\documentclass[11pt]{article}
\usepackage[a4paper,margin=1in]{geometry}
\usepackage{amsmath,amssymb,amsthm}
\usepackage{mathtools}
\usepackage{bm}
\usepackage{xcolor}
\usepackage{booktabs}
\usepackage{graphicx}
\usepackage{enumitem}
\usepackage{microtype}
\usepackage{setspace}
\usepackage[authoryear,round]{natbib}
\usepackage{hyperref}
\hypersetup{
  colorlinks=true,
  linkcolor=blue,
  citecolor=blue,
  urlcolor=blue
}
\newtheorem{theorem}{Theorem}
\newtheorem{proposition}{Proposition}
\newtheorem{lemma}{Lemma}
\newtheorem{corollary}{Corollary}
\theoremstyle{definition}
\newtheorem{assumption}{Assumption}
\newtheorem{remark}{Remark}
\title{Robust Global Testing for Mean Functions in Sparse and Dense Functional and Longitudinal Data}
\author{
Behdad Mostafaiy$^{1}$\thanks{Corresponding author. Email: \texttt{behdad.mostafaiy@gmail.com}}
\and
Mustafa Abduljabbar Dawood$^{2}$\thanks{Email: \texttt{mostafa.abduljabar1201a@coadec.uobaghdad.edu.iq}}
}
\date{}

\begin{document}
\maketitle

\begin{center}
$^{1}$Department of Statistics and Applications, University of Mohaghegh Ardabili, Iran\\
$^{2}$Department of Statistics, College of Administration and Economics, University of Baghdad, Iraq
\end{center}
\begin{abstract}
We develop a robust quadratic score test for a specified mean function in functional and longitudinal data. The construction combines a convex loss with subject level scores, retaining dependence among repeated measurements. An explicit score centering condition ensures that the target remains the population mean. We establish the null distribution and Sobolev separation rates across sparse and dense sampling, identifying a transition
from a nonparametric testing rate to the root $n$ boundary. A Gaussian multiplier bootstrap remains valid across these regimes, including settings with non-Gaussian quadratic limits. For implementation, we use a finite spline approximation and estimate the curvature through symmetric score increments. Under the stated conditions, the feasible procedure preserves null calibration and attains a consistency separation rate matching the lower bound over model classes containing a Gaussian submodel. Simulations illustrate the performance of squared, Huber, and absolute losses under regular errors and contamination, and an analysis of longitudinal CD4 counts demonstrates the method.
\end{abstract}

\section{Introduction}
Functional and longitudinal studies often ask whether a population mean curve agrees with a specified function. Let $g$ denote the mean of a stochastic process observed at subject specific, irregular time points. We consider the global hypothesis $H_0:g=g_0$. When each subject contributes only a few measurements, the test must account for both dependence within subjects and measurement variability. Heavy tailed errors and atypical responses
can further affect procedures based on squared loss. Estimation of functional means across sparse and dense sampling has been studied extensively; see \citet{yao2005}, \citet{lihsing2010}, \citet{caiyuan2011}, \citet{zhangwang2016}, \citet{xiao2020}, \citet{kalogridisvanaelst2023}, \citet{bergeretal2025}, and \citet{dawoodmostafaiy2025}. These contributions cover smoothing methods, convergence rates, and robust estimation.
Global testing, however, requires separate calibration and power analysis.

Existing tests approach this problem through integrated squared deviations, supremum statistics, or likelihood ratio type comparisons. Integrated statistics accumulate evidence across the domain, whereas supremum statistics emphasize the largest standardized departure. In nonparametric regression, generalized likelihood ratio methods can have Wilks type limits and attain optimal testing rates; see \citet{fan2001} and \citet{shangcheng2013}. For longitudinal and functional data, the independent sampling unit is the subject, so the covariance structure of repeated measurements must also enter the analysis.

\citet{staicu2014plrt} developed a pseudo likelihood ratio test that estimates and pre-whitens the within subject covariance, and \citet{staicu2015} studied significance tests under complex functional dependence. Supremum procedures include the $F_{\max}$ construction of \citet{zhangetal2019}. More recently, \citet{louwu2026} developed simultaneous inference for mean curves across sparse and increasingly dense sampling, using subject weighted local linear smoothing, jackknife bias correction, and Gaussian multiplier calibration. Their statistic is a studentized supremum. Here we study how to construct and calibrate a robust quadratic score test when the sampling density changes both the covariance scale and the form of the null distribution.

The paper makes three connected contributions. First, we construct a subject level quadratic score test from a convex loss, with an explicit score centering condition that keeps the population mean as its target. Second, we establish its null behavior and Sobolev separation rates across sparse, transition, and dense sampling, and prove validity of a subject multiplier bootstrap without assuming a common limiting null law. Third, we justify a feasible spline implementation by controlling projection and curvature estimation errors on the testing scale. Under the additional assumptions needed for power under alternatives, the feasible test attains a consistency separation rate matching the lower bound over model classes containing the specified Gaussian submodel.

The construction starts from $g=g_0+h$ and penalizes the roughness of $h$. The penalty derivative then vanishes under $H_0:h=0$, giving a score statistic evaluated directly at the null. Averaging scores within each subject retains both dependence among repeated observations and the diagonal measurement contribution, which remains first order under sparse sampling. We also study an objective improvement statistic. Its null equivalence to the score statistic follows from a functional Bahadur representation, with separate remainder arguments for Lipschitz scores and the discontinuous score of absolute loss.

The testing rate reflects the balance between measurement variability and variation across trajectories. Write $m_H$ for the harmonic mean of the cluster sizes and $r$ for the Sobolev smoothness order. The consistency separation rate is
\[
\rho_n^*=(nm_H)^{-2r/(4r+1)}\vee n^{-1/2}.
\]
The two terms balance at $m_H\asymp n^{1/(4r)}$. Below this scale, the rate agrees with regression testing based on an effective sample size $nm_H$. Above it, additional measurements reveal each trajectory more accurately, but the number of independent subjects remains $n$. A Gaussian random intercept plus measurement error submodel supplies the matching lower bound and explains the root $n$ barrier.

The null distribution also changes with sampling density. Under sparse sampling, the normalized covariance spectrum is diffuse and the centered quadratic statistic has a Gaussian limit. With dense sampling, nonnegligible trajectory eigenvalues may persist, producing a weighted quadratic limit; transition sequences may contain both components. Assigning one Gaussian multiplier to each subject preserves the covariance structure needed for calibration across these regimes. The proof compares the full score statistic with a Gaussian quadratic form and then controls the effect of estimating its covariance.

For computation, we use a B-spline Galerkin approximation and estimate the curvature matrix through symmetric score increments. This accommodates smooth scores and absolute loss within the same implementation. Explicit conditions on spline dimension and curvature estimation ensure asymptotic null calibration. A shifted secant condition provides the additional control needed to establish power for the feasible statistic under alternatives.

The numerical study compares the quadratic score test with supremum and pseudo likelihood ratio procedures under regular errors and contamination. Diffuse, localized, and oscillatory alternatives show how performance depends on the shape of the departure as well as the loss. An analysis of longitudinal CD4 counts illustrates the procedure with irregular repeated measurements.

Section~\ref{sec:model} introduces the model, robust target, and Sobolev geometry. Sections~\ref{sec:criterion} and~\ref{sec:test-statistics} define the criterion, subject scores, and test statistics. Sections~\ref{sec:asymptotic-theory}--\ref{sec:finite-basis-theory} develop the null theory, power analysis, sparse to dense transition, bootstrap validity, and feasible spline approximation. Section~\ref{sec:simulation} reports the simulations, Section~\ref{sec:cd4-application} presents the CD4 application, and Section~\ref{sec:discussion} concludes.

\section{Observation model, robust target, and Sobolev geometry}
\label{sec:model}

Let \(X_1,\ldots,X_n\) be independent copies of a process \(X\)
on \([0,1]\), with mean \(g(t)=E\{X(t)\}\). For subject \(i\),
we observe
\[
U_{ij}=g(T_{ij})+Z_i(T_{ij})+\varepsilon_{ij},
\qquad j=1,\ldots,m_i,
\]
where \(Z_i(t)=X_i(t)-g(t)\) is the centered subject specific
process and \(\varepsilon_{ij}\) is measurement error.
The blocks
\(\{(Z_i,T_{i1},\ldots,T_{im_i},
\varepsilon_{i1},\ldots,\varepsilon_{im_i})\}_{i=1}^n\)
are independent across subjects. Measurements within a subject
may be dependent, since they share the trajectory \(Z_i\).
This sampling structure is standard in sparse longitudinal and
functional data analysis; see \citet{yao2005} and
\citet{kalogridisvanaelst2023}.

Conditional on \(m_i\), the observation times are independent
and identically distributed, independently of \(Z_i\) and the
measurement errors, with density \(\pi\) satisfying
\[
0<c_\pi\leq\pi(t)\leq C_\pi<\infty,
\qquad t\in[0,1].
\]
The conditional residual law is assumed common across subjects and measurement indices. When cluster sizes are random, the stated marginal conditions also hold conditional on those sizes.
Write \(e_{ij}=U_{ij}-g(T_{ij})
=Z_i(T_{ij})+\varepsilon_{ij}\).
Let \(\rho:\mathbb R\to[0,\infty)\) be convex, and choose a score
\(\psi(u)\in\partial\rho(u)\); when \(\rho\) is differentiable,
\(\psi=\rho'\).

\begin{assumption}\label{ass:mean-consistency}
For every \(t\in[0,1]\), the following conditional expectations
exist and satisfy
\[
\begin{aligned}
E(e_{ij}\mid T_{ij}=t)&=0,\\
E\{\psi(e_{ij})\mid T_{ij}=t\}&=0.
\end{aligned}
\]
\end{assumption}

The first condition centers the observations at the mean, while
the second makes that mean a solution of the population score
equation. Both hold under conditional symmetry about zero when
\(\psi\) is odd and the expectations exist. Without score
centering, the criterion generally targets an \(M\)-location
that differs from the mean; see \citet{huberronchetti2009}.

Define \(\Psi_t(a)=E\{\psi(e_{ij}-a)\mid T_{ij}=t\}\).
The local curvature \(I(t)\) is determined by
\[
\Psi_t(a)=-I(t)a+o(|a|),
\qquad a\to0.
\]
Assumption~\ref{ass:loss-primitive} specifies the uniform remainder
and score increment bounds used below. For differentiable scores,
\(I(t)=E\{\psi'(e_{ij})\mid T_{ij}=t\}\), provided differentiation
can be interchanged with conditional expectation.

For an integer \(r\geq1\), we work in the Sobolev space
\[
W^{r,2}[0,1]
=
\left\{
f:f^{(0)},\ldots,f^{(r-1)}
\text{ are absolutely continuous},
\ f^{(r)}\in L^2[0,1]
\right\},
\]
using the usual absolutely continuous representatives and weak
derivatives; see \citet{adamsfournier2003}.
The roughness penalty is defined by
\[
J(f,h)=\int_0^1 f^{(r)}(t)h^{(r)}(t)\,dt.
\]
The curvature and design density determine the weighted inner product
\[
V(f,h)=\int_0^1 I(t)f(t)h(t)\pi(t)\,dt.
\]
Under the assumed upper and positive lower bounds on \(I\) and
\(\pi\), \(V(f,f)^{1/2}\) is equivalent to the \(L^2\) norm.
For \(\lambda>0\), set
\[
\langle f,h\rangle_\lambda=V(f,h)+\lambda J(f,h),
\]
with induced norm
\(\|f\|_\lambda=\langle f,f\rangle_\lambda^{1/2}\).
For each fixed \(\lambda>0\), this norm is equivalent to the
Sobolev norm. Point evaluation is therefore continuous, and
\(\mathcal H_\lambda
=(W^{r,2}[0,1],\langle\cdot,\cdot\rangle_\lambda)\)
is a reproducing kernel Hilbert space; see \citet{wahba1990}
and \citet{gu2013}. Denote its kernel by \(K_\lambda(s,t)\).

\begin{proposition}\label{prop:eigensystem}
Suppose \(w=I\pi\in C^{2r-1}[0,1]\) and
\(0<c_w\leq w(t)\leq C_w<\infty\) on \([0,1]\).
There exist functions
\(\{\phi_\nu\}_{\nu\geq1}\subset W^{r,2}[0,1]\)
and nonnegative eigenvalues \(\{\gamma_\nu\}_{\nu\geq1}\),
ordered nondecreasingly, such that
\begin{align}
V(\phi_\mu,\phi_\nu)&=\delta_{\mu\nu},
\label{eq:eigensystem-V}\\
J(\phi_\mu,\phi_\nu)&=\gamma_\nu\delta_{\mu\nu}.
\label{eq:eigensystem-J}
\end{align}
The null eigenspace is \(\mathcal P_{r-1}\), with
\(\gamma_1=\cdots=\gamma_r=0\), and the positive eigenvalues
satisfy \(\gamma_\nu\asymp\nu^{2r}\) as \(\nu\to\infty\).
The family \(\{\phi_\nu\}_{\nu\geq1}\) is complete in
\(L^2([0,1],w(t)\,dt)\). For every \(\lambda>0\), the functions
\(e_{\nu,\lambda}=(1+\lambda\gamma_\nu)^{-1/2}\phi_\nu\)
form an orthonormal basis of \(\mathcal H_\lambda\), and
\begin{equation}
K_\lambda(s,t)
=
\sum_{\nu\geq1}
\frac{\phi_\nu(s)\phi_\nu(t)}{1+\lambda\gamma_\nu}.
\label{eq:kernel-spectral-expansion}
\end{equation}
\end{proposition}

\begin{proof}
The bounds on \(w\) make \(V(f,f)^{1/2}\) equivalent to the
ordinary \(L^2\) norm. The form \(J\), with domain
\(W^{r,2}[0,1]\), is closed and nonnegative in the weighted
\(L^2\) space. Its null space is \(\mathcal P_{r-1}\), since
\(J(f,f)=0\) precisely when \(f^{(r)}=0\).
The compact Sobolev embedding and the spectral theorem for
closed forms therefore give a complete \(V\)-orthonormal
family satisfying
\begin{equation}
J(\phi_\nu,f)=\gamma_\nu V(\phi_\nu,f),
\qquad f\in W^{r,2}[0,1].
\label{eq:weak-eigenproblem}
\end{equation}
The zero eigenvalue has multiplicity \(r\).

For a positive eigenvalue, the weak equation implies
\[
(-1)^r\phi_\nu^{(2r)}(t)
=
\gamma_\nu w(t)\phi_\nu(t),
\qquad 0<t<1.
\]
The regularity of \(w\) gives the regularity needed to integrate
by parts. Since the endpoint traces of the test functions
are arbitrary, the resulting boundary conditions are
\[
\phi_\nu^{(j)}(0)=\phi_\nu^{(j)}(1)=0,
\qquad j=r,\ldots,2r-1.
\]
This is a regular self adjoint boundary value problem of order
\(2r\). Its eigenvalue asymptotics give
\(\gamma_\nu\asymp\nu^{2r}\); see \citet{birkhoff1908}
and \citet{coddingtonlevinson1987}. Related spline arguments
appear in \citet{utreras1988} and \citet{shangcheng2013}.

Equations \eqref{eq:eigensystem-V} and
\eqref{eq:eigensystem-J} yield
\[
\langle\phi_\mu,\phi_\nu\rangle_\lambda
=
(1+\lambda\gamma_\nu)\delta_{\mu\nu}.
\]
Thus the normalized functions \(e_{\nu,\lambda}\) are
orthonormal in \(\mathcal H_\lambda\). To establish completeness,
suppose \(f\in\mathcal H_\lambda\) is orthogonal to every
\(e_{\nu,\lambda}\). By \eqref{eq:weak-eigenproblem},
\(\langle f,\phi_\nu\rangle_\lambda
=(1+\lambda\gamma_\nu)V(f,\phi_\nu)\), so \(f\) is
\(V\)-orthogonal to every \(\phi_\nu\).
Completeness in the weighted \(L^2\) space then gives \(f=0\).

Expanding the evaluation representer \(K_\lambda(s,\cdot)\)
in this orthonormal basis gives
\eqref{eq:kernel-spectral-expansion}.
For each fixed argument, the series converges in
\(\mathcal H_\lambda\), and hence pointwise by continuity
of point evaluation.
\end{proof}

\section{Null centered penalized \texorpdfstring{\(M\)}{M}-estimation and subject level score structure}
\label{sec:criterion}

Let \(g_0\in W^{r,2}[0,1]\) be fixed and known. Writing
\(g=g_0+h\), we test \(H_0:h=0\) against \(H_1:h\neq0\).
For the null residuals \(e_{ij}^{0}=U_{ij}-g_0(T_{ij})\), define
\begin{equation*}
Q_{n,\lambda}(h)
=
\frac{1}{n}\sum_{i=1}^{n}\frac{1}{m_i}
\sum_{j=1}^{m_i}\rho\{e_{ij}^{0}-h(T_{ij})\}
+\frac{\lambda}{2}J(h,h).
\end{equation*}
Each subject receives equal weight. For \(\lambda>0\), the estimator is
\begin{equation*}
\widehat h_{n,\lambda}
\in\operatorname*{arg\,min}_{h\in W^{r,2}[0,1]}Q_{n,\lambda}(h),
\end{equation*}
with corresponding mean estimator
\(\widehat g_{n,\lambda}=g_0+\widehat h_{n,\lambda}\).

Under the assumed convexity and coercivity of \(\rho\), a minimizer
exists whenever there are at least \(r\) distinct design points.
Indeed, bounded criterion values control the Sobolev seminorm and
the fitted values at these points. Since evaluation at \(r\) distinct
points is injective on the null space \(\mathcal P_{r-1}\) of \(J\),
these bounds control the full Sobolev norm. Weak compactness and
weak lower semicontinuity then yield existence. Under the continuous
design assumptions, distinctness holds almost surely when the total
sample size is at least \(r\). Uniqueness is unnecessary for the
subsequent localization arguments.

The criterion follows the penalized \(M\)-spline approach of
\citet{kalogridis2022} and \citet{kalogridisvanaelst2023}.
Penalizing \(h\) makes the penalty derivative vanish under the null,
so the null score contains no term \(\lambda J(g_0,v)\).
For \(v\in\mathcal H_\lambda\), differentiation gives
\[
DQ_{n,\lambda}(h)[v]
=
-\frac{1}{n}\sum_{i=1}^{n}\frac{1}{m_i}
\sum_{j=1}^{m_i}\psi\{e_{ij}^{0}-h(T_{ij})\}v(T_{ij})
+\lambda J(h,v).
\]
At points of nondifferentiability, choosing
\(\psi(u)\in\partial\rho(u)\) gives a subgradient functional.

Under \(H_0\), \(e_{ij}^{0}=e_{ij}\), and the reproducing property
represents the subject score by
\begin{equation}
G_{i,\lambda}
=
\frac{1}{m_i}\sum_{j=1}^{m_i}
\psi(e_{ij})K_\lambda(T_{ij},\cdot).
\label{eq:cluster-score-element}
\end{equation}
Thus \(\mathcal S_i(v)=\langle G_{i,\lambda},v\rangle_\lambda\).
Conditional score centering and integrability imply
\(E(G_{i,\lambda})=0\). The first order stochastic term is therefore
determined by
\(\overline G_{n,\lambda}=n^{-1}\sum_{i=1}^{n}G_{i,\lambda}\),
an average of independent subject contributions.

Assuming \(E\|G_{i,\lambda}\|_\lambda^2<\infty\), define
\(\Omega_{i,\lambda}=E(G_{i,\lambda}\otimes G_{i,\lambda})\),
where \((a\otimes b)v=\langle b,v\rangle_\lambda a\).
The covariance form is
\(\Gamma_{i,\lambda}(v,w)
=\langle\Omega_{i,\lambda}v,w\rangle_\lambda\).
Independence across subjects gives the covariance operator of
\(\sqrt n\,\overline G_{n,\lambda}\) as
\begin{equation*}
\Omega_{n,\lambda}
=
\frac{1}{n}\sum_{i=1}^{n}\Omega_{i,\lambda}.
\end{equation*}
This allows unequal cluster sizes and different subject distributions.
Cluster sizes are treated as fixed; for random sizes, the calculations
and assumptions are understood conditionally on the observed sizes.

Expanding the subject scores yields
\begin{equation*}
\Gamma_{i,\lambda}(v,w)
=
\frac{1}{m_i^2}
E\left\{
\sum_{j,k=1}^{m_i}
\psi(e_{ij})\psi(e_{ik})v(T_{ij})w(T_{ik})
\right\}.
\end{equation*}
Suppose the design points within each subject are independent with
density \(\pi\), and the conditional moments
\(q_\psi(t)=E\{\psi^2(e_{ij})\mid T_{ij}=t\}\) and
\(c_\psi(s,t)=E\{\psi(e_{ij})\psi(e_{ik})
\mid T_{ij}=s,T_{ik}=t\}\), for \(j\neq k\), are common across
subjects, indices, and cluster sizes. Separating diagonal and
off diagonal terms gives
\begin{equation*}
\begin{aligned}
\Gamma_{i,\lambda}(v,w)
={}&\frac{1}{m_i}\int_0^1 q_\psi(t)v(t)w(t)\pi(t)\,dt\\
&+\frac{m_i-1}{m_i}\int_0^1\int_0^1
c_\psi(s,t)v(s)w(t)\pi(s)\pi(t)\,ds\,dt.
\end{aligned}
\end{equation*}
Subject specific moments give the same decomposition when these
moments vary across subjects. The diagonal term retains individual
score variation and cannot generally be omitted when \(m_i\) is
bounded. For fixed \(v,w\), it vanishes as \(m_i\to\infty\) if its
integral remains bounded. Negligibility of its trace contribution
when \(\lambda\to0\) requires additional bounds.

The trace identity
\(\operatorname{tr}(\Omega_{i,\lambda})
=E\|G_{i,\lambda}\|_\lambda^2\) and the reproducing property imply
\[
\operatorname{tr}(\Omega_{i,\lambda})
=
\frac{1}{m_i^2}E\left\{
\sum_{j,k=1}^{m_i}
\psi(e_{ij})\psi(e_{ik})K_\lambda(T_{ij},T_{ik})
\right\}.
\]
Consequently, the quadratic score statistic has exact mean
\(\mu_{n,\lambda}=\operatorname{tr}(\Omega_{n,\lambda})\).
Its reference variance is
\(\sigma_{n,\lambda}^2=2\operatorname{tr}(\Omega_{n,\lambda}^2)\),
which is the leading variance under the fourth moment negligibility
conditions in Section~\ref{sec:asymptotic-theory}.

The effective smoothing dimension
\(d_\lambda=\sum_{\nu\geq1}(1+\lambda\gamma_\nu)^{-1}\)
is analogous to the effective degrees of freedom in smoothing spline
theory; see \citet{wahba1990}, \citet{gu2013}, and
\citet{shangcheng2013}. By Lemma~\ref{lem:kernel-orders},
\(d_\lambda\asymp\lambda^{-1/(2r)}\) as \(\lambda\to0\).
It describes the spectral contribution of smoothing, while the
operators \(\Omega_{i,\lambda}\) retain dependence within subjects.

\section{Robust global test and finite sample implementation}
\label{sec:test-statistics}
\label{sec:test}

Under the null hypothesis, the subject scores \(G_{i,\lambda}\)
in \eqref{eq:cluster-score-element} are independent and have mean
zero. We measure their aggregate departure from zero through
\begin{equation}
S_{n,\lambda}
=
\left\|\frac{1}{\sqrt n}\sum_{i=1}^nG_{i,\lambda}\right\|_\lambda^2.
\label{eq:score-statistic}
\end{equation}
Large values provide evidence against \(H_0\). This statistic
measures a quadratic departure from the null, whereas
\citet{louwu2026} use a studentized supremum for simultaneous
inference on the mean curve. Section~\ref{sec:simulation} compares
their behavior under localized alternatives.

The score statistic uses only null residuals, so its null analysis
does not require an expansion of the estimator under alternatives.
Its expectation is
\(E(S_{n,\lambda})=\mu_{n,\lambda}
=\operatorname{tr}(\Omega_{n,\lambda})\).
Under the sparse regime conditions of Theorem~\ref{thm:score-clt},
the centered and standardized statistic has a normal limit, with
\(\mu_{n,\lambda}\asymp m^{-1}\lambda^{-1/(2r)}\) and
\(\operatorname{Var}(S_{n,\lambda})
\asymp m^{-2}\lambda^{-1/(2r)}\).

An alternative measures the reduction in the penalized criterion:
\[
T_{n,\lambda}
=
2n\{Q_{n,\lambda}(0)-Q_{n,\lambda}(\widehat h_{n,\lambda})\}.
\]
For a general convex loss, this is an objective improvement
statistic rather than a likelihood ratio statistic.
Transferring the standardized null limit from \(S_{n,\lambda}\)
requires
\(T_{n,\lambda}-S_{n,\lambda}=o_p(\sigma_{n,\lambda})\).
When \(m\) is bounded, the required scale is
\(o_p\{\lambda^{-1/(4r)}\}\).
The functional Bahadur remainder provides the basis for the
comparison in Theorem~\ref{thm:statistic-equivalence}.

For computation, take a spline space
\(\mathcal S_q=\operatorname{span}\{B_1,\ldots,B_q\}
\subset W^{r,2}[0,1]\) containing \(\mathcal P_{r-1}\), and write
\(B(t)=\{B_1(t),\ldots,B_q(t)\}^{\prime}\).
This finite dimensional construction follows standard spline
approximation arguments; see \citet{wahba1990}, \citet{gu2013},
and \citet{sandeetal2020}. The matrix representing \(V\) is
\[
A_q=\int_0^1 I(t)B(t)B(t)^{\prime}\pi(t)\,dt.
\]
Let \(P_q\) have entries
\((P_q)_{k\ell}=\int_0^1
B_k^{(r)}(t)B_\ell^{(r)}(t)\,dt\). The penalized inner product
is then represented by
\begin{equation}
H_{q,\lambda}=A_q+\lambda P_q.
\label{eq:population-basis-hessian}
\end{equation}
For subject \(i\), form the score vector
\(s_i=m_i^{-1}\sum_{j=1}^{m_i}\psi(e_{ij}^{0})B(T_{ij})\)
and its sample average
\begin{equation}
\bar s=\frac{1}{n}\sum_{i=1}^n s_i.
\label{eq:subject-score-vector}
\end{equation}
The resulting score statistic is
\[
S_{n,\lambda}^{(q)}
=
n\bar s^{\prime}H_{q,\lambda}^{-1}\bar s.
\]
Proposition~\ref{prop:finite-basis-projection} identifies this
quantity as the squared \(\lambda\) norm of the orthogonal
projection of \(\sqrt n\,\overline G_{n,\lambda}\) onto
\(\mathcal S_q\).

To estimate \(A_q\), we use the symmetric score increment
\begin{equation}
w_\delta(u)
=
\frac{\psi(u+\delta)-\psi(u-\delta)}{2\delta},
\qquad \delta>0.
\label{eq:score-increment-curvature-weight}
\end{equation}
This gives the empirical curvature matrix
\[
\widehat A_{q,\delta}
=
\frac{1}{n}\sum_{i=1}^n\frac{1}{m_i}\sum_{j=1}^{m_i}
w_\delta(e_{ij}^{0})B(T_{ij})B(T_{ij})^{\prime}.
\]
Monotonicity of \(\psi\) ensures that \(w_\delta\geq0\).
Under \(H_0\) and Assumption~\ref{ass:loss-primitive},
\(E\{w_\delta(e_{ij})\mid T_{ij}=t\}
=I(t)+O(\delta^\alpha)\), uniformly in \(t\).
For a Lipschitz score, \(w_\delta\) is uniformly bounded.
For a bounded score with jumps, it can instead grow at rate
\(\delta^{-1}\) near a jump.

Set \(\widehat H_{q,\lambda}
=\widehat A_{q,\delta}+\lambda P_q\).
Whenever this matrix is invertible, the feasible statistic is
\[
\widehat S_{n,\lambda}^{(q)}
=
n\bar s^{\prime}\widehat H_{q,\lambda}^{-1}\bar s.
\]
We estimate the covariance of the subject scores by
\[
\widehat C_q
=
\frac{1}{n}\sum_{i=1}^n
(s_i-\bar s)(s_i-\bar s)^{\prime}.
\]
Writing
\(\widehat M_{q,\lambda}
=\widehat H_{q,\lambda}^{-1}\widehat C_q\),
the estimated centering is
\(\widehat\mu_{n,\lambda}^{(q)}
=\operatorname{tr}(\widehat M_{q,\lambda})\), and the estimated
reference variance is
\(\{\widehat\sigma_{n,\lambda}^{(q)}\}^2
=2\operatorname{tr}(\widehat M_{q,\lambda}^2)\).
The standardized statistic is therefore
\[
\widehat Z_{n,\lambda}^{(q)}
=
\frac{\widehat S_{n,\lambda}^{(q)}
-\widehat\mu_{n,\lambda}^{(q)}}
{\widehat\sigma_{n,\lambda}^{(q)}}.
\]
Section~\ref{sec:finite-basis-theory} controls the errors introduced
by spline truncation and estimation of the metric.

For bootstrap calibration, let \(\xi_1,\ldots,\xi_n\) be independent
\(N(0,1)\) variables, independent of the data. Assigning one
multiplier to each subject keeps that subject's observations
together. Form the bootstrap score vector
\[
S_n^*
=
\frac{1}{\sqrt n}\sum_{i=1}^n\xi_i(s_i-\bar s).
\]
Its conditional covariance is \(\widehat C_q\), and the bootstrap
statistic is
\[
\widehat S_{n,\lambda}^{*(q)}
=
(S_n^*)^{\prime}\widehat H_{q,\lambda}^{-1}S_n^*.
\]
From \(B\) independent multiplier replicates, compute
\[
p_{\mathrm{boot}}
=
\frac{1+\sum_{b=1}^B
\mathbf{1}\{\widehat S_{n,\lambda}^{*(q,b)}
\geq\widehat S_{n,\lambda}^{(q)}\}}
{B+1}.
\]
Section~\ref{sec:bootstrap-validity} establishes conditional
validity under the stated conditions for the sparse, transition,
and dense regimes. Section~\ref{sec:finite-basis-theory} extends
the result to the spline statistic.

The smoothing parameter determines which spectral directions
contribute most to the test. Since the effective dimension has
order \(\lambda^{-1/(2r)}\), decreasing \(\lambda\) gives greater
weight to higher frequency departures while increasing null
variability. Theorem~\ref{thm:score-clt} requires
\(\lambda\to0\) and \(n\lambda^{1/(2r)}\to\infty\), together with
its other assumptions. Section~\ref{sec:power} derives the testing
scale by balancing smoothing bias against stochastic variation.

\section{Asymptotic theory for the robust global score test}
\label{sec:asymptotic-theory}

We first derive the null distribution of the score statistic under
sparse sampling. We then study the estimator and use its local
quadratic expansion to compare the score and objective improvement
statistics. Throughout, write \(s_\lambda=\lambda^{1/(2r)}\);
the effective dimension is of order \(s_\lambda^{-1}\).
All results below are under \(H_0\). When cluster sizes are random,
expectations and probability statements are understood conditionally
on the observed sizes.

\subsection{Assumptions and preliminary bounds}

The assumptions control the local behavior of the loss and the
dependence within each subject. Recall the conditional score map
\(\Psi_t(a)=E\{\psi(e_{ij}-a)\mid T_{ij}=t\}\).

\begin{assumption}\label{ass:loss-primitive}
The loss \(\rho\) is proper, finite, lower semicontinuous, convex,
and coercive. Let \(\psi\) be a measurable nondecreasing selection
from its subgradient. The following conditions hold.
\begin{enumerate}
\item[(i)]
For some \(\alpha\in(0,1]\), \(\beta\in(0,2]\), \(a_0>0\),
and \(C<\infty\), the score is centered:
\begin{equation*}
\Psi_t(0)=0.
\end{equation*}
Uniformly in \(t\in[0,1]\) and \(|a|\leq a_0\), its local
linearization satisfies
\begin{equation}
|\Psi_t(a)+I(t)a|\leq C|a|^{1+\alpha},
\label{eq:score-linearization}
\end{equation}
and its increments satisfy
\begin{equation}
E\!\left[
\{\psi(e_{ij}-a)-\psi(e_{ij})\}^2
\mid T_{ij}=t
\right]\leq C|a|^\beta.
\label{eq:score-increment}
\end{equation}

\item[(ii)]
For some \(\eta>0\) and every fixed \(M<\infty\), there is a
constant \(C_M<\infty\) such that
\[
\sup_{t\in[0,1]}\sup_{|a|\leq M}
E\{|\psi(e_{ij}-a)|^{4+\eta}\mid T_{ij}=t\}
\leq C_M.
\]

\item[(iii)]
The curvature satisfies \(I\in C^{2r-1}[0,1]\) and
\(0<c_I\leq I(t)\leq C_I<\infty\) on \([0,1]\).

\item[(iv)]
The conditional distribution of \(e_{ij}\) given \(T_{ij}=t\)
assigns zero probability to every discontinuity point of \(\psi\).
\end{enumerate}
\end{assumption}

\begin{assumption}\label{ass:sparse-cluster}
The following conditions hold.
\begin{enumerate}
\item[(i)]
The cluster sizes are deterministic or independent of the
observations, with \(1\leq m_i\leq m_{\max}\) for a fixed
\(m_{\max}<\infty\).

\item[(ii)]
Conditional on \(m_i\), the design points are independent and
identically distributed with density
\(\pi\in C^{2r-1}[0,1]\), where
\(0<c_\pi\leq\pi(t)\leq C_\pi<\infty\).
The design independence conditions of Section~\ref{sec:model}
remain in force.

\item[(iii)]
The conditional second moment
\(q_\psi(t)=E\{\psi^2(e_{ij})\mid T_{ij}=t\}\) satisfies
\(0<c_q\leq q_\psi(t)\leq C_q<\infty\).

\item[(iv)]
Let \(\Sigma_i(\boldsymbol T_i)\) be the conditional covariance
matrix of the score vector given
\(\boldsymbol T_i=(T_{i1},\ldots,T_{im_i})\), and set
\(D_i=\operatorname{diag}\{q_\psi(T_{i1}),\ldots,
q_\psi(T_{im_i})\}\). For some \(c_\Sigma>0\),
\begin{equation*}
D_i^{-1/2}\Sigma_i(\boldsymbol T_i)D_i^{-1/2}
\succeq c_\Sigma I_{m_i}
\end{equation*}
almost surely, uniformly in \(i\) and \(n\).
\end{enumerate}
\end{assumption}

The last condition permits dependence within subjects but rules
out degenerating conditional score covariance. By design
independence and score centering, \(D_i\) is the diagonal of
\(\Sigma_i(\boldsymbol T_i)\). Thus
\(R_i=D_i^{-1/2}\Sigma_i(\boldsymbol T_i)D_i^{-1/2}\)
has unit diagonal and satisfies
\[
R_i\preceq\operatorname{tr}(R_i)I_{m_i}
=m_iI_{m_i}\preceq m_{\max}I_{m_i}.
\]
An upper covariance bound therefore follows from bounded cluster
sizes; only the lower bound needs to be imposed.

\begin{lemma}\label{lem:population-loss}
Under Assumption~\ref{ass:loss-primitive}(i)--(ii), uniformly in
\(t\in[0,1]\) and \(|a|\leq a_0\),
\begin{equation*}
E\{\rho(e_{ij}-a)-\rho(e_{ij})\mid T_{ij}=t\}
=
\frac12 I(t)a^2+O(|a|^{2+\alpha}).
\end{equation*}
\end{lemma}

\begin{proof}
A finite convex function is locally absolutely continuous, so
\[
\rho(e-a)-\rho(e)=-\int_0^a\psi(e-u)\,du.
\]
The selected subgradient can differ from the derivative only on
a set of Lebesgue measure zero. The moment assumption permits
conditional Fubini, giving
\[
E\{\rho(e_{ij}-a)-\rho(e_{ij})\mid T_{ij}=t\}
=-\int_0^a\Psi_t(u)\,du.
\]
Write \(\Psi_t(u)=-I(t)u+r_t(u)\), where
\(|r_t(u)|\leq C|u|^{1+\alpha}\).
Integrating the linear term gives \(I(t)a^2/2\), while
\(\left|\int_0^a r_t(u)\,du\right|
\leq C|a|^{2+\alpha}/(2+\alpha)\), uniformly in \(t\).
\end{proof}

\begin{lemma}\label{lem:kernel-orders}
Under the conditions of Proposition~\ref{prop:eigensystem},
for every fixed \(p>1/(2r)\),
\begin{equation*}
\sum_{\nu\geq1}(1+\lambda\gamma_\nu)^{-p}
\asymp s_\lambda^{-1},
\qquad \lambda\to0.
\end{equation*}
In particular, both \(d_\lambda\) and
\(\sum_{\nu\geq1}(1+\lambda\gamma_\nu)^{-2}\) have order
\(s_\lambda^{-1}\). For sufficiently small \(\lambda\),
\begin{equation}
\sup_{t\in[0,1]}K_\lambda(t,t)\leq Cs_\lambda^{-1},
\label{eq:kernel-diagonal-order}
\end{equation}
and every \(f\in\mathcal H_\lambda\) satisfies
\begin{equation}
\|f\|_\infty\leq Cs_\lambda^{-1/2}\|f\|_\lambda.
\label{eq:evaluation-inequality}
\end{equation}
\end{lemma}

\begin{proof}
The eigenvalue order \(\gamma_\nu\asymp\nu^{2r}\) and integral
comparison give
\[
\sum_{\nu\geq1}(1+\lambda\gamma_\nu)^{-p}
\asymp
s_\lambda^{-1}\int_{s_\lambda}^{\infty}
(1+u^{2r})^{-p}\,du.
\]
The integral converges to a finite positive limit when \(2rp>1\).
Taking \(p=1\) and \(p=2\) gives the stated special cases.

The Sobolev interpolation inequality in \citet{adamsfournier2003}
gives
\[
\|f\|_\infty
\leq C\|f\|_{L^2}^{1-1/(2r)}
\|f^{(r)}\|_{L^2}^{1/(2r)}
+C\|f\|_{L^2}.
\]
The bounds
\(\|f\|_{L^2}\leq C\|f\|_\lambda\) and
\(\|f^{(r)}\|_{L^2}\leq s_\lambda^{-r}\|f\|_\lambda\)
then yield \eqref{eq:evaluation-inequality}.
Finally, the norm of the evaluation functional is
\[
K_\lambda(t,t)^{1/2}
=
\sup_{\|f\|_\lambda\leq1}|f(t)|.
\]
Applying \eqref{eq:evaluation-inequality} proves
\eqref{eq:kernel-diagonal-order}.
\end{proof}

\subsection{Covariance bounds and the null distribution}

We compare the score covariance with the operator induced by
\(V\). This gives the centering and variance orders and identifies
the part of the score statistic that determines its null limit.
Define \(A_\lambda:\mathcal H_\lambda\to\mathcal H_\lambda\) by
\begin{equation*}
\langle A_\lambda f,g\rangle_\lambda=V(f,g).
\end{equation*}
Its eigenvalues are
\(a_{\nu,\lambda}=(1+\lambda\gamma_\nu)^{-1}\).
Lemma~\ref{lem:kernel-orders} therefore gives
\(\operatorname{tr}(A_\lambda)\asymp s_\lambda^{-1}\) and
\(\operatorname{tr}(A_\lambda^2)\asymp s_\lambda^{-1}\).

\begin{proposition}\label{prop:covariance-comparison}
Suppose Assumptions~\ref{ass:loss-primitive} and
\ref{ass:sparse-cluster} hold. Let
\(m_H=\{n^{-1}\sum_{i=1}^n m_i^{-1}\}^{-1}\).
There are constants \(0<c<C<\infty\) such that
\begin{equation}
\frac{c}{m_H}A_\lambda
\preceq\Omega_{n,\lambda}
\preceq\frac{C}{m_H}A_\lambda.
\label{eq:operator-comparison}
\end{equation}
Consequently,
\begin{align}
\mu_{n,\lambda}
&\asymp m_H^{-1}s_\lambda^{-1},
\label{eq:mu-order}\\
\operatorname{tr}(\Omega_{n,\lambda}^2)
&\asymp m_H^{-2}s_\lambda^{-1},
\label{eq:variance-trace-order}\\
\|\Omega_{n,\lambda}\|_{\mathrm{op}}
&\leq Cm_H^{-1}.
\label{eq:omega-op-order}
\end{align}
In particular, as \(\lambda\to0\),
\begin{equation}
\frac{\|\Omega_{n,\lambda}\|_{\mathrm{op}}^2}
{\operatorname{tr}(\Omega_{n,\lambda}^2)}
=O(s_\lambda)\to0.
\label{eq:derived-spectral-dispersion}
\end{equation}
\end{proposition}

\begin{proof}
For \(f\in\mathcal H_\lambda\), write
\(\boldsymbol f_i=\{f(T_{i1}),\ldots,f(T_{im_i})\}^{\prime}\).
Design independence and score centering imply
\(E\{\psi(e_{ij})\mid\boldsymbol T_i\}=0\). Hence
\[
E\{\mathcal S_i(f)^2\mid\boldsymbol T_i\}
=
m_i^{-2}\boldsymbol f_i^{\prime}
\Sigma_i(\boldsymbol T_i)\boldsymbol f_i.
\]
The conditional covariance bounds and integration over the design
give
\begin{equation*}
\begin{aligned}
\frac{c_\Sigma}{m_i}\int_0^1q_\psi(t)f^2(t)\pi(t)\,dt
&\leq\langle\Omega_{i,\lambda}f,f\rangle_\lambda\\
&\leq
\frac{m_{\max}}{m_i}\int_0^1q_\psi(t)f^2(t)\pi(t)\,dt.
\end{aligned}
\end{equation*}
Since \(q_\psi\) and \(I\) are bounded above and away from zero,
the integral is comparable to \(V(f,f)\). Thus
\(cm_i^{-1}A_\lambda\preceq\Omega_{i,\lambda}
\preceq Cm_i^{-1}A_\lambda\). Averaging proves
\eqref{eq:operator-comparison}.

Let \(\omega_{\nu,\lambda}\) be the ordered eigenvalues of
\(\Omega_{n,\lambda}\). The min--max principle gives
\[
\frac{c}{m_H}a_{\nu,\lambda}
\leq\omega_{\nu,\lambda}
\leq\frac{C}{m_H}a_{\nu,\lambda}.
\]
Summing these bounds and their squares, then applying
Lemma~\ref{lem:kernel-orders}, proves
\eqref{eq:mu-order} and \eqref{eq:variance-trace-order}.
Since \(a_{\nu,\lambda}\leq1\), the operator norm bound follows.
Combining it with \eqref{eq:variance-trace-order} proves
\eqref{eq:derived-spectral-dispersion}.
\end{proof}

\begin{lemma}\label{lem:G-fourth}
Under Assumption~\ref{ass:loss-primitive}(ii)--(iii) and
Assumption~\ref{ass:sparse-cluster}(ii),
\begin{equation*}
\sup_{1\leq i\leq n}E\|G_{i,\lambda}\|_\lambda^4
\leq Cs_\lambda^{-2}.
\end{equation*}
\end{lemma}

\begin{proof}
The triangle inequality and
\(\|K_\lambda(t,\cdot)\|_\lambda^2=K_\lambda(t,t)\) give
\[
\|G_{i,\lambda}\|_\lambda
\leq Cs_\lambda^{-1/2}
\frac{1}{m_i}\sum_{j=1}^{m_i}|\psi(e_{ij})|.
\]
Jensen's inequality therefore yields
\[
E\|G_{i,\lambda}\|_\lambda^4
\leq Cs_\lambda^{-2}
\frac{1}{m_i}\sum_{j=1}^{m_i}E|\psi(e_{ij})|^4.
\]
The uniform moment bound in
Assumption~\ref{ass:loss-primitive}(ii) completes the proof.
\end{proof}

Expanding the squared norm in \eqref{eq:score-statistic} gives
\begin{equation}
S_{n,\lambda}-\mu_{n,\lambda}
=D_{n,\lambda}+U_{n,\lambda},
\label{eq:score-decomposition}
\end{equation}
where the centered diagonal term is
\[
D_{n,\lambda}
=
\frac{1}{n}\sum_{i=1}^n
\{\|G_{i,\lambda}\|_\lambda^2
-E\|G_{i,\lambda}\|_\lambda^2\}.
\]
The remaining term is the degenerate quadratic form
\[
U_{n,\lambda}
=
\frac{2}{n}\sum_{i<j}
\langle G_{i,\lambda},G_{j,\lambda}\rangle_\lambda.
\]
Independence and centering of the subject scores give its variance:
\begin{equation*}
\tau_{n,\lambda}^2
=
\frac{4}{n^2}\sum_{i<j}
\operatorname{tr}(\Omega_{i,\lambda}\Omega_{j,\lambda}).
\end{equation*}

\begin{lemma}\label{lem:variance-equivalence}
Under Assumptions~\ref{ass:loss-primitive} and
\ref{ass:sparse-cluster}, if \(s_\lambda\to0\) and
\(ns_\lambda\to\infty\), then
\begin{equation}
\tau_{n,\lambda}^2
\sim2\operatorname{tr}(\Omega_{n,\lambda}^2)
\asymp m_H^{-2}s_\lambda^{-1},
\label{eq:tau-equivalent}
\end{equation}
and
\begin{equation}
D_{n,\lambda}/\tau_{n,\lambda}\xrightarrow{p}0.
\label{eq:diagonal-negligible}
\end{equation}
\end{lemma}

\begin{proof}
The definition of \(\Omega_{n,\lambda}\) gives the exact identity
\[
2\operatorname{tr}(\Omega_{n,\lambda}^2)
=
\tau_{n,\lambda}^2
+\frac{2}{n^2}\sum_{i=1}^n
\operatorname{tr}(\Omega_{i,\lambda}^2).
\]
The individual covariance comparison implies
\(\operatorname{tr}(\Omega_{i,\lambda}^2)
\leq C\operatorname{tr}(A_\lambda^2)
\leq Cs_\lambda^{-1}\).
Thus the last term is \(O(n^{-1}s_\lambda^{-1})\), whereas
\(\operatorname{tr}(\Omega_{n,\lambda}^2)
\asymp s_\lambda^{-1}\), since \(m_H\) is bounded.
This proves \eqref{eq:tau-equivalent}.

Independence across subjects and Lemma~\ref{lem:G-fourth} give
\[
E(D_{n,\lambda}^2)
\leq\frac{1}{n^2}\sum_{i=1}^n
E\|G_{i,\lambda}\|_\lambda^4
\leq Cn^{-1}s_\lambda^{-2}.
\]
Dividing by \(\tau_{n,\lambda}^2\) gives
\(O\{(ns_\lambda)^{-1}\}=o(1)\).
Chebyshev's inequality proves \eqref{eq:diagonal-negligible}.
\end{proof}

\begin{theorem}\label{thm:score-clt}
Suppose Assumptions~\ref{ass:loss-primitive} and
\ref{ass:sparse-cluster} hold. If \(s_\lambda\to0\) and
\(ns_\lambda\to\infty\), then
\begin{equation*}
\frac{S_{n,\lambda}-\mu_{n,\lambda}}
{\{2\operatorname{tr}(\Omega_{n,\lambda}^2)\}^{1/2}}
\xrightarrow{d}N(0,1).
\end{equation*}
The same limit holds with denominator \(\tau_{n,\lambda}\).
\end{theorem}

\begin{proof}
By Lemma~\ref{lem:variance-equivalence}, it suffices to establish
a normal limit for \(U_{n,\lambda}/\tau_{n,\lambda}\).
Set \(W_{ij}=\langle G_{i,\lambda},G_{j,\lambda}\rangle_\lambda\)
and \(Q_{n,\lambda}=\sum_{i<j}W_{ij}\).
The kernels are degenerate, and
\(\operatorname{Var}(Q_{n,\lambda})
=n^2\tau_{n,\lambda}^2/4\asymp n^2s_\lambda^{-1}\).
We verify the conditions of Proposition~3.2 in \citet{dejong1987} for the quadratic form considered here.

First, independence and Lemma~\ref{lem:G-fourth} imply
\(E(W_{ij}^4)\leq Cs_\lambda^{-4}\). Hence
\[
\frac{\sum_{i<j}E(W_{ij}^4)}
{\{\operatorname{Var}(Q_{n,\lambda})\}^2}
=
O\{(ns_\lambda)^{-2}\}\to0.
\]
For three distinct indices, conditioning on \(G_{i,\lambda}\)
gives
\[
E(W_{ij}^2W_{ik}^2)
=
E\!\left[
\langle\Omega_{j,\lambda}G_{i,\lambda},G_{i,\lambda}\rangle_\lambda
\langle\Omega_{k,\lambda}G_{i,\lambda},G_{i,\lambda}\rangle_\lambda
\right].
\]
The individual covariance bound gives
\(\sup_i\|\Omega_{i,\lambda}\|_{\mathrm{op}}\leq C\).
The last expectation is therefore \(O(s_\lambda^{-2})\).
Summing over the \(O(n^3)\) configurations and normalizing by
\(\{\operatorname{Var}(Q_{n,\lambda})\}^2\) gives \(O(n^{-1})\).

For four distinct indices, a representative cycle satisfies
\[
E(W_{ij}W_{ik}W_{\ell j}W_{\ell k})
=
\operatorname{tr}
(\Omega_{j,\lambda}\Omega_{\ell,\lambda}
\Omega_{k,\lambda}\Omega_{i,\lambda}).
\]
The trace ideal inequalities in \citet{Simon2005} give
\[
|\operatorname{tr}(ABCD)|
\leq
\|A\|_{\mathrm{op}}\|B\|_{\mathrm{HS}}
\|C\|_{\mathrm{op}}\|D\|_{\mathrm{HS}}.
\]
Since
\(\sup_i\operatorname{tr}(\Omega_{i,\lambda}^2)
\leq Cs_\lambda^{-1}\), each cycle is \(O(s_\lambda^{-1})\).
The normalized sum over \(O(n^4)\) cycles is therefore
\(O(s_\lambda)\to0\).

These bounds verify the \(G_I\), \(G_{II}\), and \(G_{IV}\) conditions of Proposition~3.2 in \citet{dejong1987}.
Thus \(Q_{n,\lambda}\), divided by its standard deviation,
converges to \(N(0,1)\). Since \(U_{n,\lambda}=2Q_{n,\lambda}/n\),
Lemma~\ref{lem:variance-equivalence} and Slutsky's theorem
complete the proof.
\end{proof}

\begin{remark}\label{rem:sparse-diagonal}
The negligible term \(D_{n,\lambda}\) concerns the quadratic
decomposition across subjects. It does not remove the diagonal
contribution within a subject's covariance operator.
Each \(\Omega_{i,\lambda}\) retains both \(j=k\) and \(j\neq k\)
terms. Under bounded cluster sizes, the former generally remain
part of the leading covariance.
\end{remark}

\begin{remark}\label{rem:loss-examples}
The loss conditions include familiar choices from robust spline
estimation; see \citet{kalogridis2022} and
\citet{kalogridisvanaelst2023}. For
\(\rho(u)=u^2/2\), the score is \(\psi(u)=u\),
\(I(t)=1\), the linearization is exact, and \(\beta=2\).
For the usual Huber score, \(I(t)\) is the conditional probability
of the quadratic region; suitable regularity at the clipping
points gives the local expansion, while Lipschitz continuity
gives \(\beta=2\). See \citet{huberronchetti2009}.
For absolute loss, \(\psi=\operatorname{sign}\), and a positive,
locally regular conditional density at zero gives
\(I(t)=2f_t(0)\) and \(\beta=1\); see \citet{koenker2005}.
These examples must also satisfy the stated centering,
smoothness, and moment conditions.
\end{remark}

\subsection{Estimator expansion and objective improvement}
\label{subsec:bahadur}

To compare the two statistics, we first localize the estimator
and then approximate the criterion uniformly near the null.
We begin with Lipschitz scores. Absolute loss requires a separate
argument, given in Subsection~\ref{subsec:absolute-loss-bahadur}.

\begin{assumption}\label{ass:lipschitz-score}
There is a constant \(L_\psi<\infty\) such that
\[
|\psi(u)-\psi(v)|\leq L_\psi|u-v|,
\qquad u,v\in\mathbb R.
\]
\end{assumption}

Squared, Huber, and pseudo-Huber losses satisfy this condition.
Write \(r_{n,\lambda}=(ns_\lambda)^{-1/2}\) for the stochastic
radius of the null score.

\begin{lemma}\label{lem:null-score-rate}
Under Assumptions~\ref{ass:loss-primitive} and
\ref{ass:sparse-cluster},
\begin{equation*}
E\|\overline G_{n,\lambda}\|_\lambda^2
=
\frac{\mu_{n,\lambda}}{n}
\asymp(nm_Hs_\lambda)^{-1}
\asymp r_{n,\lambda}^2.
\end{equation*}
Consequently,
\begin{equation}
\|\overline G_{n,\lambda}\|_\lambda=O_p(r_{n,\lambda}).
\label{eq:null-score-rate}
\end{equation}
\end{lemma}

\begin{proof}
Independence and centering give
\(E\|\overline G_{n,\lambda}\|_\lambda^2
=n^{-2}\sum_iE\|G_{i,\lambda}\|_\lambda^2
=\mu_{n,\lambda}/n\).
Proposition~\ref{prop:covariance-comparison} and boundedness of
\(m_H\) give the order. Markov's inequality gives
\eqref{eq:null-score-rate}.
\end{proof}

For the population criterion, use the difference
\(\mathcal Q_\lambda(h)-\mathcal Q_\lambda(0)
:=E\{Q_{n,\lambda}(h)-Q_{n,\lambda}(0)\}\).
This definition requires integrability of the loss difference,
rather than of the loss itself.

\begin{lemma}\label{lem:population-localization}
Under Assumption~\ref{ass:loss-primitive} and the conditions of
Lemma~\ref{lem:kernel-orders}, if \(\|h\|_\infty\leq a_0\), then
\[
\mathcal Q_\lambda(h)-\mathcal Q_\lambda(0)
=
\frac12\|h\|_\lambda^2+R_\lambda(h),
\]
where
\(|R_\lambda(h)|
\leq Cs_\lambda^{-\alpha/2}\|h\|_\lambda^{2+\alpha}\).
If a deterministic radius \(a_{n,\lambda}>0\) satisfies
\begin{equation*}
s_\lambda^{-1/2}a_{n,\lambda}\to0,
\end{equation*}
then the remainder is \(o(\|h\|_\lambda^2)\), uniformly over
\(0<\|h\|_\lambda\leq a_{n,\lambda}\).
\end{lemma}

\begin{proof}
Integrating Lemma~\ref{lem:population-loss} over the design gives
\[
\mathcal Q_\lambda(h)-\mathcal Q_\lambda(0)
=
\frac12V(h,h)+\frac{\lambda}{2}J(h,h)
+O\!\left(\int_0^1|h(t)|^{2+\alpha}\pi(t)\,dt\right).
\]
The integral is bounded by
\(C\|h\|_\infty^\alpha V(h,h)\).
The evaluation inequality and
\(V(h,h)\leq\|h\|_\lambda^2\) give the stated remainder bound.
On the specified ball, its ratio to \(\|h\|_\lambda^2\) is at
most \(C(s_\lambda^{-1/2}a_{n,\lambda})^\alpha\to0\).
The same condition ensures \(\|h\|_\infty\leq a_0\) throughout
the ball for sufficiently large \(n\).
\end{proof}

\begin{lemma}\label{lem:lambda-ball-entropy}
Let
\(\mathcal H_\lambda(a)
=\{h\in\mathcal H_\lambda:\|h\|_\lambda\leq a\}\).
For \(r\geq2\) and sufficiently small \(\lambda\),
\begin{equation*}
\log N\{\varepsilon,\mathcal H_\lambda(a),\|\cdot\|_\infty\}
\leq
C\left(\frac{a s_\lambda^{-r}}{\varepsilon}\right)^{1/(r-1/2)},
\qquad
0<\varepsilon\leq Ca s_\lambda^{-1/2}.
\end{equation*}
Consequently,
\begin{equation}
\int_0^{Ca s_\lambda^{-1/2}}
\sqrt{\log N\{\varepsilon,\mathcal H_\lambda(a),
\|\cdot\|_\infty\}}\,d\varepsilon
\leq Ca s_\lambda^{-1}.
\label{eq:entropy-integral-lambda-ball}
\end{equation}
\end{lemma}

\begin{proof}
The norm bounds imply
\(\|h\|_{L^2}\leq Ca\) and
\(\|h^{(r)}\|_{L^2}\leq a s_\lambda^{-r}\).
Sobolev interpolation therefore gives
\(\|h\|_{W^{r,2}}\leq Ca s_\lambda^{-r}\).
Moreover,
\[
|h^{(r-1)}(t)-h^{(r-1)}(u)|
\leq\|h^{(r)}\|_{L^2}|t-u|^{1/2}.
\]
Thus the ball is contained in a Hölder ball of smoothness
\(r-1/2\) and radius \(Ca s_\lambda^{-r}\).
The entropy bound follows from Theorem~2.7.1 of
\citet{vanderVaartWellner1996}.

Put \(\kappa=(r-1/2)^{-1}<2\).
Integrating the square root of the entropy bound up to
\(Ca s_\lambda^{-1/2}\) gives an upper bound proportional to
\[
(a s_\lambda^{-r})^{\kappa/2}
(a s_\lambda^{-1/2})^{1-\kappa/2}
=
a s_\lambda^{-1}.
\]
This proves \eqref{eq:entropy-integral-lambda-ball}.
\end{proof}

Define the nonlinear loss remainder by
\(q(e,a)=\rho(e-a)-\rho(e)+\psi(e)a\).
Convexity gives \(q(e,a)\geq0\), and the Lipschitz score condition
gives \(q(e,a)=\int_0^a\{\psi(e)-\psi(e-u)\}\,du\).
It follows that
\begin{equation}
|q(e,a)|\leq \frac{L_\psi}{2}a^2.
\label{eq:q-envelope}
\end{equation}
For \(|a|,|b|\leq R\), the same identity yields
\begin{equation}
|q(e,a)-q(e,b)|\leq L_\psi R|a-b|.
\label{eq:q-local-lipschitz}
\end{equation}
Write
\(Z_i(h)=m_i^{-1}\sum_{j=1}^{m_i}q(e_{ij},h(T_{ij}))\)
and
\(\mathbb E_{n,\lambda}(h)
=n^{-1}\sum_i\{Z_i(h)-EZ_i(h)\}\).
The criterion then has the exact decomposition
\begin{equation}
\begin{aligned}
Q_{n,\lambda}(h)-Q_{n,\lambda}(0)
={}&-\langle\overline G_{n,\lambda},h\rangle_\lambda\\
&+\{\mathcal Q_\lambda(h)-\mathcal Q_\lambda(0)\}
+\mathbb E_{n,\lambda}(h).
\end{aligned}
\label{eq:criterion-exact-decomposition}
\end{equation}

\begin{lemma}\label{lem:uniform-nonlinear-remainder}
Suppose Assumptions~\ref{ass:loss-primitive},
\ref{ass:sparse-cluster}, and \ref{ass:lipschitz-score} hold,
and let \(r\geq2\). For \(a=a_{n,\lambda}>0\),
\begin{equation}
\sup_{h\in\mathcal H_\lambda(a)}
|\mathbb E_{n,\lambda}(h)|
=
O_p\left(\frac{a^2}{\sqrt n\,s_\lambda^{3/2}}\right).
\label{eq:uniform-nonlinear-rate}
\end{equation}
\end{lemma}

\begin{proof}
Set \(R=Ca s_\lambda^{-1/2}\), so
\(\|h\|_\infty\leq R\) on \(\mathcal H_\lambda(a)\).
Equations \eqref{eq:q-envelope} and
\eqref{eq:q-local-lipschitz} give
\(0\leq Z_i(h)\leq CR^2\) and
\(|Z_i(h)-Z_i(g)|\leq CR\|h-g\|_\infty\).

Let \(\epsilon_1,\ldots,\epsilon_n\) be independent Rademacher
variables, independent of the data. Symmetrization gives
\[
E\sup_h|\sqrt n\,\mathbb E_{n,\lambda}(h)|
\leq
2E\sup_h\left|
n^{-1/2}\sum_{i=1}^n\epsilon_iZ_i(h)
\right|.
\]
Here and below the suprema are over \(\mathcal H_\lambda(a)\).
Conditionally on the data, the Rademacher process has increment
metric
\(d_n(h,g)=\{n^{-1}\sum_i[Z_i(h)-Z_i(g)]^2\}^{1/2}\),
which is bounded by \(CR\|h-g\|_\infty\).
Since \(Z_i(0)=0\), the entropy maximal inequality gives
\[
E_\epsilon\sup_h\left|
n^{-1/2}\sum_i\epsilon_iZ_i(h)
\right|
\leq
CR\int_0^{Ca s_\lambda^{-1/2}}
\sqrt{\log N\{u,\mathcal H_\lambda(a),\|\cdot\|_\infty\}}\,du.
\]
See Lemma~2.3.1 and Corollary~2.2.8 in \citet{vanderVaartWellner1996}.
Lemma~\ref{lem:lambda-ball-entropy} bounds the right side by
\(Ca^2s_\lambda^{-3/2}\).
Taking expectations and applying Markov's inequality proves
\eqref{eq:uniform-nonlinear-rate}.
\end{proof}

\begin{theorem}\label{thm:estimator-localization}
Suppose Assumptions~\ref{ass:loss-primitive},
\ref{ass:sparse-cluster}, and \ref{ass:lipschitz-score} hold,
and let \(r\geq2\). If \(s_\lambda\to0\) and
\(ns_\lambda^3\to\infty\), then
\begin{equation*}
\|\widehat h_{n,\lambda}\|_\lambda=O_p(r_{n,\lambda}).
\end{equation*}
\end{theorem}

\begin{proof}
Fix \(\varepsilon>0\). By Lemma~\ref{lem:null-score-rate},
choose \(K\) so that
\(\|\overline G_{n,\lambda}\|_\lambda\leq Kr_{n,\lambda}\)
with probability at least \(1-\varepsilon\) for large \(n\).
On this event, the linear term on the sphere
\(\|h\|_\lambda=Mr_{n,\lambda}\) is bounded in absolute value
by \(KMr_{n,\lambda}^2\).

Since \(s_\lambda^{-1/2}r_{n,\lambda}
=(ns_\lambda^2)^{-1/2}\to0\), the population expansion applies
uniformly on the ball. The empirical remainder is
\(O_p\{M^2r_{n,\lambda}^2(ns_\lambda^3)^{-1/2}\}
=o_p(M^2r_{n,\lambda}^2)\).
Thus \eqref{eq:criterion-exact-decomposition} gives, uniformly
on the sphere,
\[
Q_{n,\lambda}(h)-Q_{n,\lambda}(0)
\geq
\{\tfrac12+o_p(1)\}M^2r_{n,\lambda}^2
-KMr_{n,\lambda}^2.
\]
Choose \(M>4K\). The right side is positive with probability
at least \(1-\varepsilon-o(1)\).
If a minimizer lay outside the ball, convexity would give a point
on the segment from zero to that minimizer, lying on the sphere,
with criterion value at most \(Q_{n,\lambda}(0)\).
This contradicts the boundary inequality and proves the result
for every minimizer.
\end{proof}

The two remainder terms are summarized by
\(\delta_{n,\lambda}
=(ns_\lambda^2)^{-\alpha/2}+(ns_\lambda^3)^{-1/2}\).
The first comes from the population expansion and the second
from the centered nonlinear empirical term.

\begin{theorem}\label{thm:bahadur}
Under the assumptions of Theorem~\ref{thm:estimator-localization},
\begin{equation*}
\|\widehat h_{n,\lambda}-\overline G_{n,\lambda}\|_\lambda
=
O_p\{r_{n,\lambda}\delta_{n,\lambda}^{1/2}\}.
\end{equation*}
Consequently,
\(\widehat h_{n,\lambda}
=\overline G_{n,\lambda}+o_p(r_{n,\lambda})\)
in \(\mathcal H_\lambda\).
\end{theorem}

\begin{proof}
For each fixed \(C<\infty\), the population and empirical
remainder bounds give
\begin{equation}
Q_{n,\lambda}(h)-Q_{n,\lambda}(0)
=
\frac12\|h\|_\lambda^2
-\langle\overline G_{n,\lambda},h\rangle_\lambda
+\mathcal R_{n,\lambda}(h),
\label{eq:uniform-criterion-bahadur}
\end{equation}
uniformly over \(\|h\|_\lambda\leq Cr_{n,\lambda}\), with
\begin{equation}
\sup_{\|h\|_\lambda\leq Cr_{n,\lambda}}
|\mathcal R_{n,\lambda}(h)|
=
O_p\{r_{n,\lambda}^2\delta_{n,\lambda}\}.
\label{eq:uniform-criterion-remainder-rate}
\end{equation}
Localization and the null score rate place both
\(\widehat h_{n,\lambda}\) and \(\overline G_{n,\lambda}\)
in such a ball with arbitrarily high probability when \(C\)
is sufficiently large.

Complete the square in
\eqref{eq:uniform-criterion-bahadur} and compare the criterion
at these two points. The minimizing property gives
\[
\|\widehat h_{n,\lambda}-\overline G_{n,\lambda}\|_\lambda^2
\leq
4\sup_{\|h\|_\lambda\leq Cr_{n,\lambda}}
|\mathcal R_{n,\lambda}(h)|.
\]
Equation \eqref{eq:uniform-criterion-remainder-rate} proves the
rate. Finally, the bandwidth conditions imply
\(\delta_{n,\lambda}\to0\).
\end{proof}

Related functional Bahadur representations for smoothing spline
\(M\)-estimators were developed by \citet{shang2010} and
\citet{shangcheng2013}. For sparse functional data,
\citet{dawoodmostafaiy2025} derived a Bahadur representation
for a smoothing spline estimator of the mean function.
Here the leading term is an average of independent subject
scores, each retaining the dependence among observations
from the same trajectory.

\begin{theorem}\label{thm:statistic-equivalence}
Under the assumptions of Theorem~\ref{thm:bahadur}, suppose
\(n^\alpha s_\lambda^{2\alpha+1}\to\infty\) and
\(ns_\lambda^4\to\infty\). Then
\begin{equation}
T_{n,\lambda}-S_{n,\lambda}=o_p(s_\lambda^{-1/2}),
\label{eq:TS-equivalence}
\end{equation}
and
\begin{equation}
\frac{T_{n,\lambda}-\mu_{n,\lambda}}
{\{2\operatorname{tr}(\Omega_{n,\lambda}^2)\}^{1/2}}
\xrightarrow{d}N(0,1).
\label{eq:objective-improvement-clt}
\end{equation}
\end{theorem}

\begin{proof}
Evaluate \eqref{eq:uniform-criterion-bahadur} at the estimator
and complete the square:
\[
T_{n,\lambda}-S_{n,\lambda}
=
-n\|\widehat h_{n,\lambda}-\overline G_{n,\lambda}\|_\lambda^2
-2n\mathcal R_{n,\lambda}(\widehat h_{n,\lambda}).
\]
Theorem~\ref{thm:bahadur} and
\eqref{eq:uniform-criterion-remainder-rate} imply
\begin{equation*}
|T_{n,\lambda}-S_{n,\lambda}|
=
O_p(s_\lambda^{-1}\delta_{n,\lambda}).
\end{equation*}
The additional conditions make both terms in
\(\delta_{n,\lambda}\) equal to \(o(s_\lambda^{1/2})\),
proving \eqref{eq:TS-equivalence}.
The reference standard deviation has order \(s_\lambda^{-1/2}\).
Theorem~\ref{thm:score-clt} and Slutsky's theorem therefore give
\eqref{eq:objective-improvement-clt}.
\end{proof}

\begin{remark}\label{rem:equivalence-rates}
The two additional conditions control different errors.
The condition involving \(\alpha\) makes the population remainder
negligible on the testing scale; \(ns_\lambda^4\to\infty\)
does the same for the empirical remainder.
When \(\alpha=1\), the latter is stronger.
For example, if \(s_\lambda\asymp n^{-1/(2r+1)}\), it holds
for every integer \(r\geq2\).
\end{remark}

\subsection{Absolute loss}
\label{subsec:absolute-loss-bahadur}

For \(\rho(u)=|u|\), the score
\(\psi(u)=\operatorname{sign}(u)\), with
\(\operatorname{sign}(0)=0\), is discontinuous.
We replace the Lipschitz score argument by the threshold crossing
identity of \citet{knight1998}; see also
\citet{koenker2005}.
Conditional symmetry and a finite conditional mean ensure that
the mean and median coincide, so the target remains the mean
function.

\begin{assumption}\label{ass:absolute-loss}
Conditional on \(T_{ij}=t\), the residual has a distribution
symmetric about zero, a finite first absolute moment, and a
density \(f_t\) on \([-a_0,a_0]\).
For some \(\alpha\in(0,1]\) and constants
\(0<c_f<C_f<\infty\) and \(L_f<\infty\),
\begin{equation*}
c_f\leq f_t(0)\leq C_f
\end{equation*}
uniformly in \(t\). The density also satisfies
\begin{equation*}
|f_t(u)-f_t(0)|\leq L_f|u|^\alpha,
\end{equation*}
uniformly in \(t\in[0,1]\) and \(|u|\leq a_0\).
The function \(t\mapsto f_t(0)\) belongs to \(C^{2r-1}[0,1]\),
and the conditional sign score covariance matrices satisfy
Assumption~\ref{ass:sparse-cluster}(iv).
\end{assumption}

Symmetry centers the sign score, while local boundedness of
the density gives
\[
E[
\{\operatorname{sign}(e_{ij}-a)-\operatorname{sign}(e_{ij})\}^2
\mid T_{ij}=t]
\leq C|a|.
\]
Thus the score increment condition holds with \(\beta=1\).
The score is bounded, so all required shifted moments exist.
The curvature is \(I(t)=2f_t(0)\), and density regularity gives
\[
E\{\operatorname{sign}(e_{ij}-a)\mid T_{ij}=t\}
=
-I(t)a+O(|a|^{1+\alpha}),
\]
uniformly in \(t\).

Define \(q_A(e,a)=|e-a|-|e|+a\operatorname{sign}(e)\).
The absolute loss identity, written to include the convention
\(\operatorname{sign}(0)=0\), is
\begin{equation}
q_A(e,a)
=
\int_0^a
\{\operatorname{sign}(e)-\operatorname{sign}(e-u)\}\,du.
\label{eq:knight-absolute}
\end{equation}

\begin{lemma}\label{lem:absolute-crossing}
Under Assumption~\ref{ass:absolute-loss}, for
\(0<R\leq a_0\) and \(|a|,|b|\leq R\),
\begin{equation*}
|q_A(e,a)|\leq2R\mathbf1\{|e|\leq R\},
\end{equation*}
and
\begin{equation*}
|q_A(e,a)-q_A(e,b)|
\leq2|a-b|\mathbf1\{|e|\leq R\}.
\end{equation*}
Uniformly in \(t\in[0,1]\), these bounds imply
\begin{equation*}
E\{q_A(e_{ij},a)^2\mid T_{ij}=t\}\leq CR^3,
\end{equation*}
and
\begin{equation*}
E[\{q_A(e_{ij},a)-q_A(e_{ij},b)\}^2\mid T_{ij}=t]
\leq CR|a-b|^2.
\end{equation*}
\end{lemma}

\begin{proof}
For \(|u|\leq R\), the integrand in
\eqref{eq:knight-absolute} vanishes when \(|e|>R\) and is
bounded in absolute value by two. Integrating over the interval
from zero to \(a\), or from \(b\) to \(a\), proves the pointwise
bounds. The density assumptions give
\(P(|e_{ij}|\leq R\mid T_{ij}=t)\leq CR\).
Squaring and taking conditional expectations proves the remaining
claims.
\end{proof}

Write
\(Z_i^A(h)=m_i^{-1}\sum_jq_A(e_{ij},h(T_{ij}))\)
and
\(\mathbb E_{n,\lambda}^A(h)
=n^{-1}\sum_i\{Z_i^A(h)-EZ_i^A(h)\}\).

\begin{lemma}\label{lem:absolute-uniform-remainder}
Suppose Assumptions~\ref{ass:sparse-cluster} and
\ref{ass:absolute-loss} hold, and let \(r\geq2\).
If \(a=a_{n,\lambda}>0\) satisfies
\(s_\lambda^{-1/2}a\leq c_0\) for sufficiently small \(c_0>0\),
then
\begin{equation}
\sup_{h\in\mathcal H_\lambda(a)}
|\mathbb E_{n,\lambda}^A(h)|
=
O_p\left(\frac{a^{3/2}}{\sqrt n\,s_\lambda^{5/4}}\right).
\label{eq:absolute-uniform-rate}
\end{equation}
\end{lemma}

\begin{proof}
Set \(R=Ca s_\lambda^{-1/2}\), so
\(\|h\|_\infty\leq R\leq a_0\) on the ball.
For each subject, define
\(L_i=m_i^{-1}\sum_j\mathbf1\{|e_{ij}|\leq R\}\).
Lemma~\ref{lem:absolute-crossing} gives
\[
|Z_i^A(h)-Z_i^A(g)|
\leq2L_i\|h-g\|_\infty.
\]
Also, \(0\leq L_i\leq1\) and \(EL_i^2\leq EL_i\leq CR\).

Symmetrize over the independent subject blocks, as in the proof
of Lemma~\ref{lem:uniform-nonlinear-remainder}.
Conditionally on the data, the Rademacher increment metric is
bounded by \(2L_n\|h-g\|_\infty\), where
\(L_n=\{n^{-1}\sum_iL_i^2\}^{1/2}\).
Since \(Z_i^A(0)=0\), the entropy maximal inequality and
Lemma~\ref{lem:lambda-ball-entropy} give
\[
E_\epsilon\sup_{h\in\mathcal H_\lambda(a)}
\left|n^{-1/2}\sum_i\epsilon_iZ_i^A(h)\right|
\leq C L_n a s_\lambda^{-1}.
\]
Jensen's inequality yields \(EL_n\leq CR^{1/2}\).
Taking expectations and applying symmetrization therefore gives
\[
E\sup_{h\in\mathcal H_\lambda(a)}
|\sqrt n\,\mathbb E_{n,\lambda}^A(h)|
\leq
CR^{1/2}a s_\lambda^{-1}
=
Ca^{3/2}s_\lambda^{-5/4}.
\]
Markov's inequality proves \eqref{eq:absolute-uniform-rate}.
\end{proof}

Let \(F_t\) denote the conditional distribution function.
Taking expectations in \eqref{eq:knight-absolute} gives
\[
E\{q_A(e,a)\mid T=t\}
=
2\int_0^a\{F_t(u)-F_t(0)\}\,du.
\]
The density condition makes this
\(f_t(0)a^2+O(|a|^{2+\alpha})\), uniformly in \(t\).
Thus Lemma~\ref{lem:population-localization} applies with
\(I(t)=2f_t(0)\).
Define
\(\delta_{n,\lambda}^A
=(ns_\lambda^2)^{-\alpha/2}+(ns_\lambda^4)^{-1/4}\).
The superscript \(A\) below denotes quantities formed using
absolute loss.

\begin{theorem}\label{thm:absolute-bahadur}
Suppose Assumptions~\ref{ass:sparse-cluster} and
\ref{ass:absolute-loss} hold, and let \(r\geq2\).
If \(s_\lambda\to0\) and \(ns_\lambda^4\to\infty\), then
\begin{equation}
\|\widehat h_{n,\lambda}^A\|_\lambda=O_p(r_{n,\lambda}).
\label{eq:absolute-localization}
\end{equation}
Moreover,
\begin{equation}
\|\widehat h_{n,\lambda}^A-\overline G_{n,\lambda}^A\|_\lambda
=
O_p\{r_{n,\lambda}(\delta_{n,\lambda}^A)^{1/2}\},
\label{eq:absolute-bahadur}
\end{equation}
and hence
\begin{equation}
\widehat h_{n,\lambda}^A
=
\overline G_{n,\lambda}^A+o_p(r_{n,\lambda})
\label{eq:absolute-bahadur-op}
\end{equation}
in \(\mathcal H_\lambda\).
\end{theorem}

\begin{proof}
On a ball of radius \(Cr_{n,\lambda}\), the population remainder,
relative to \(r_{n,\lambda}^2\), is
\(O\{(ns_\lambda^2)^{-\alpha/2}\}\).
By Lemma~\ref{lem:absolute-uniform-remainder}, the corresponding
empirical remainder is
\[
O_p\left(
\frac{r_{n,\lambda}^{3/2}}
{\sqrt n\,s_\lambda^{5/4}}
\right)
=
O_p\{r_{n,\lambda}^2(ns_\lambda^4)^{-1/4}\}.
\]
Both relative remainders tend to zero.
The null score rate and the convexity argument from
Theorem~\ref{thm:estimator-localization} therefore give
\eqref{eq:absolute-localization}.

More precisely, the criterion satisfies
\[
Q_{n,\lambda}^A(h)-Q_{n,\lambda}^A(0)
=
\frac12\|h\|_\lambda^2
-\langle\overline G_{n,\lambda}^A,h\rangle_\lambda
+\mathcal R_{n,\lambda}^A(h),
\]
where, for every fixed \(C\),
\[
\sup_{\|h\|_\lambda\leq Cr_{n,\lambda}}
|\mathcal R_{n,\lambda}^A(h)|
=
O_p\{r_{n,\lambda}^2\delta_{n,\lambda}^A\}.
\]
Localization puts the estimator and leading score in this ball
with arbitrarily high probability for large \(C\).
Completing the square and comparing their criterion values gives
\[
\|\widehat h_{n,\lambda}^A-\overline G_{n,\lambda}^A\|_\lambda^2
\leq
4\sup_{\|h\|_\lambda\leq Cr_{n,\lambda}}
|\mathcal R_{n,\lambda}^A(h)|.
\]
This proves \eqref{eq:absolute-bahadur}.
The bandwidth conditions imply
\(\delta_{n,\lambda}^A\to0\), yielding
\eqref{eq:absolute-bahadur-op}.
\end{proof}

\begin{corollary}\label{cor:absolute-equivalence}
Under the assumptions of Theorem~\ref{thm:absolute-bahadur},
suppose additionally that
\(n^\alpha s_\lambda^{2\alpha+1}\to\infty\) and
\(ns_\lambda^6\to\infty\). Then
\[
T_{n,\lambda}^A-S_{n,\lambda}^A=o_p(s_\lambda^{-1/2}),
\]
and
\begin{equation}
\frac{T_{n,\lambda}^A-\mu_{n,\lambda}^A}
{\{2\operatorname{tr}[(\Omega_{n,\lambda}^A)^2]\}^{1/2}}
\xrightarrow{d}N(0,1).
\label{eq:absolute-objective-clt}
\end{equation}
\end{corollary}

\begin{proof}
The completed square argument used for
Theorem~\ref{thm:statistic-equivalence} gives
\[
|T_{n,\lambda}^A-S_{n,\lambda}^A|
=
O_p(s_\lambda^{-1}\delta_{n,\lambda}^A).
\]
The additional conditions make each term of
\(\delta_{n,\lambda}^A\) equal to \(o(s_\lambda^{1/2})\).
The first assertion follows. Applying
Theorem~\ref{thm:score-clt} to the sign score and then
Slutsky's theorem proves \eqref{eq:absolute-objective-clt}.
\end{proof}

\begin{remark}\label{rem:jump-cost}
Both loss classes have the leading radius \(r_{n,\lambda}\),
but their empirical remainder bounds differ.
For Lipschitz scores, the relative bound is
\((ns_\lambda^3)^{-1/2}\); for the sign score, it is
\((ns_\lambda^4)^{-1/4}\).
Accordingly, the sufficient stochastic rate conditions for
statistic equivalence are \(ns_\lambda^4\to\infty\) and
\(ns_\lambda^6\to\infty\), respectively.
The stronger condition obtained for absolute loss reflects the
threshold crossing term in this argument; it is not asserted
to be necessary.
\end{remark}

\section{Local power and minimax separation in the sparse regime}
\label{sec:power}

We study power through the mean shift of the score evaluated at
the null. We first control this shift and the accompanying change
in covariance, then derive the local power function and the
Sobolev separation rate. Throughout, the alternatives change the
mean from \(g_0\) to \(g_0+h\), while leaving the joint law of the
centered process, measurement errors, and design unchanged.
As in Section~\ref{sec:asymptotic-theory}, calculations are
conditional on the cluster sizes when these are random.

\subsection{Score mean and covariance under alternatives}

Under \(g=g_0+h\), write \(e_{ij}=U_{ij}-g(T_{ij})\).
The null residual is then \(e_{ij}^0=e_{ij}+h(T_{ij})\), so the
subject score becomes
\[
G_{i,\lambda}^{(h)}
=
\frac{1}{m_i}\sum_{j=1}^{m_i}
\psi\{e_{ij}+h(T_{ij})\}K_\lambda(T_{ij},\cdot).
\]
Under the common marginal model, its expectation is the same
for every subject. Denote it by
\(b_\lambda(h)=E\{G_{i,\lambda}^{(h)}\}\).

\begin{lemma}\label{lem:alternative-score-mean}
Suppose Assumptions~\ref{ass:loss-primitive} and
\ref{ass:sparse-cluster} hold, and let
\(\|h\|_\infty\leq a_0\). Then
\(b_\lambda(h)=A_\lambda h+r_\lambda(h)\), where
\begin{equation*}
\|r_\lambda(h)\|_\lambda
\leq C\|h\|_\infty^\alpha\|h\|_\lambda.
\end{equation*}
\end{lemma}

\begin{proof}
For \(f\in\mathcal H_\lambda\), the reproducing property gives
\[
\langle b_\lambda(h),f\rangle_\lambda
=
\int_0^1\Psi_t\{-h(t)\}f(t)\pi(t)\,dt.
\]
By \eqref{eq:score-linearization},
\(\Psi_t\{-h(t)\}=I(t)h(t)+\widetilde r_t\{h(t)\}\),
with
\(|\widetilde r_t\{h(t)\}|\leq C|h(t)|^{1+\alpha}\).
The linear term integrates to
\(V(h,f)=\langle A_\lambda h,f\rangle_\lambda\).
The remainder satisfies
\[
|\langle r_\lambda(h),f\rangle_\lambda|
\leq
C\|h\|_\infty^\alpha
\|h\|_{L^2(\pi)}\|f\|_{L^2(\pi)}
\leq
C\|h\|_\infty^\alpha\|h\|_\lambda\|f\|_\lambda,
\]
where the last inequality uses the positive lower bound on \(I\).
Taking the supremum over \(\|f\|_\lambda\leq1\) proves the claim.
\end{proof}

The operator \(A_\lambda\) attenuates the higher frequency
components of the signal. For
\(h=\sum_{\nu\geq1}h_\nu\phi_\nu\), its retained energy is
\[
\|A_\lambda h\|_\lambda^2
=
\sum_{\nu\geq1}\frac{h_\nu^2}{1+\lambda\gamma_\nu}.
\]
Using \((1+x)^{-1}\geq1-x\) gives
\begin{equation}
\|A_\lambda h\|_\lambda^2
\geq V(h,h)-\lambda J(h,h).
\label{eq:smoothed-signal-lower}
\end{equation}
Thus \(\lambda J(h,h)\) bounds the signal energy lost through
smoothing.

To describe the fluctuations around the mean shift, define
\(\widetilde G_{i,\lambda}^{(h)}
=G_{i,\lambda}^{(h)}-b_\lambda(h)\).
Its covariance operator is
\[
\Omega_{i,\lambda}^{(h)}
=
E\{\widetilde G_{i,\lambda}^{(h)}
\otimes\widetilde G_{i,\lambda}^{(h)}\}.
\]
Write
\(\Omega_{n,\lambda}^{(h)}
=n^{-1}\sum_i\Omega_{i,\lambda}^{(h)}\), so that
\(\Omega_{n,\lambda}^{(0)}=\Omega_{n,\lambda}\).

\begin{lemma}\label{lem:alternative-covariance}
Suppose Assumptions~\ref{ass:loss-primitive} and
\ref{ass:sparse-cluster} hold. For sufficiently small
\(\|h\|_\infty\),
\begin{equation}
\Omega_{n,\lambda}^{(h)}\preceq CA_\lambda.
\label{eq:alternative-cov-upper}
\end{equation}
Moreover,
\begin{equation}
-C\|h\|_\infty^{\beta/2}A_\lambda
\preceq
\Omega_{n,\lambda}^{(h)}-\Omega_{n,\lambda}
\preceq
C\|h\|_\infty^{\beta/2}A_\lambda.
\label{eq:alternative-cov-perturbation}
\end{equation}
Consequently,
\begin{equation}
|\operatorname{tr}(\Omega_{n,\lambda}^{(h)})
-\operatorname{tr}(\Omega_{n,\lambda})|
\leq Cs_\lambda^{-1}\|h\|_\infty^{\beta/2},
\label{eq:alternative-trace-shift}
\end{equation}
and
\begin{equation}
|\operatorname{tr}\{(\Omega_{n,\lambda}^{(h)})^2\}
-\operatorname{tr}(\Omega_{n,\lambda}^2)|
\leq Cs_\lambda^{-1}\|h\|_\infty^{\beta/2}.
\label{eq:alternative-trace2-shift}
\end{equation}
\end{lemma}

\begin{proof}
For \(f\in\mathcal H_\lambda\), let
\[
\mathcal S_i^{(h)}(f)
=
\frac{1}{m_i}\sum_{j=1}^{m_i}
\psi\{e_{ij}+h(T_{ij})\}f(T_{ij}).
\]
Jensen's inequality and \eqref{eq:score-increment} give
\[
E[\{\mathcal S_i^{(h)}(f)-\mathcal S_i^{(0)}(f)\}^2]
\leq C\|h\|_\infty^\beta V(f,f).
\]
Center these scores by writing
\(X_i^{(h)}(f)=\mathcal S_i^{(h)}(f)
-E\{\mathcal S_i^{(h)}(f)\}\).
Centering cannot increase the \(L^2\) norm, so
\[
\|X_i^{(h)}(f)-X_i^{(0)}(f)\|_{L^2}
\leq C\|h\|_\infty^{\beta/2}V(f,f)^{1/2}.
\]
The null covariance bound and the triangle inequality also give
\(\|X_i^{(h)}(f)\|_{L^2}
+\|X_i^{(0)}(f)\|_{L^2}\leq CV(f,f)^{1/2}\)
for small \(\|h\|_\infty\). Factoring the difference of their
squared norms yields
\[
|\langle
(\Omega_{i,\lambda}^{(h)}-\Omega_{i,\lambda})f,f
\rangle_\lambda|
\leq C\|h\|_\infty^{\beta/2}V(f,f).
\]
Averaging proves \eqref{eq:alternative-cov-perturbation}.
The null comparison with \(A_\lambda\) then gives
\eqref{eq:alternative-cov-upper}.

Put
\(\Delta=\Omega_{n,\lambda}^{(h)}-\Omega_{n,\lambda}\).
Taking traces in the operator bounds gives
\(|\operatorname{tr}(\Delta)|
\leq C\|h\|_\infty^{\beta/2}\operatorname{tr}(A_\lambda)\),
which proves \eqref{eq:alternative-trace-shift}.

For either \(B=\Omega_{n,\lambda}^{(h)}\) or
\(B=\Omega_{n,\lambda}\), positivity and \(B\preceq CA_\lambda\)
give
\[
|\operatorname{tr}(\Delta B)|
\leq
C\|h\|_\infty^{\beta/2}\operatorname{tr}(A_\lambda B)
\leq
C\|h\|_\infty^{\beta/2}\operatorname{tr}(A_\lambda^2).
\]
Apply this bound to the identity
\((\Omega_{n,\lambda}^{(h)})^2-\Omega_{n,\lambda}^2
=\Delta\Omega_{n,\lambda}^{(h)}+\Omega_{n,\lambda}\Delta\).
Since \(\operatorname{tr}(A_\lambda^2)\asymp s_\lambda^{-1}\),
\eqref{eq:alternative-trace2-shift} follows.
\end{proof}

\subsection{Local power and consistency}

Let
\(\sigma_{0,n,\lambda}
=\{2\operatorname{tr}(\Omega_{n,\lambda}^2)\}^{1/2}\).
Under bounded cluster sizes,
\(\sigma_{0,n,\lambda}\asymp s_\lambda^{-1/2}\).
For \(\gamma\in(0,1)\), the theoretical score test rejects when
\[
\frac{S_{n,\lambda}-\mu_{n,\lambda}}
{\sigma_{0,n,\lambda}}>z_{1-\gamma},
\]
where \(z_{1-\gamma}\) is the standard normal quantile.
The following result concerns this test with population null
centering and scaling.

\begin{theorem}\label{thm:local-power}
Suppose Assumptions~\ref{ass:loss-primitive} and
\ref{ass:sparse-cluster} hold, \(s_\lambda\to0\), and
\(ns_\lambda\to\infty\). Let \(h=h_n\) satisfy
\(\|h\|_\infty^{\beta/2}=o(s_\lambda^{1/2})\) and
\(n\|h\|_\infty^\alpha\|h\|_\lambda^2/
\sigma_{0,n,\lambda}\to0\).
If
\begin{equation}
\Lambda_n(h)
=
\frac{n\|A_\lambda h\|_\lambda^2}{\sigma_{0,n,\lambda}}
\to\kappa\in[0,\infty),
\label{eq:noncentrality}
\end{equation}
then
\begin{equation*}
\frac{S_{n,\lambda}-\mu_{n,\lambda}}
{\sigma_{0,n,\lambda}}
\xrightarrow{d}N(\kappa,1).
\end{equation*}
The limiting rejection probability is
\(1-\Phi(z_{1-\gamma}-\kappa)\).
\end{theorem}

\begin{proof}
Define
\(Z_{n,\lambda}^{(h)}
=n^{-1/2}\sum_i\widetilde G_{i,\lambda}^{(h)}\).
Expanding the statistic under the alternative gives
\begin{equation}
\begin{aligned}
S_{n,\lambda}
={}&n\|b_\lambda(h)\|_\lambda^2
+2\sqrt n\langle b_\lambda(h),Z_{n,\lambda}^{(h)}
\rangle_\lambda\\
&+\|Z_{n,\lambda}^{(h)}\|_\lambda^2.
\end{aligned}
\label{eq:alternative-statistic-decomposition}
\end{equation}
Lemma~\ref{lem:alternative-score-mean} and
\(\|A_\lambda h\|_\lambda\leq\|h\|_\lambda\) imply
\[
|\|b_\lambda(h)\|_\lambda^2-\|A_\lambda h\|_\lambda^2|
\leq C\|h\|_\infty^\alpha\|h\|_\lambda^2.
\]
The assumed rate therefore gives
\begin{equation}
\frac{n\{\|b_\lambda(h)\|_\lambda^2
-\|A_\lambda h\|_\lambda^2\}}
{\sigma_{0,n,\lambda}}\to0.
\label{eq:mean-signal-equivalence}
\end{equation}

Lemma~\ref{lem:alternative-covariance} shows that
\(\operatorname{tr}(\Omega_{n,\lambda}^{(h)})
-\mu_{n,\lambda}=o(\sigma_{0,n,\lambda})\) and that
\(\operatorname{tr}\{(\Omega_{n,\lambda}^{(h)})^2\}/
\operatorname{tr}(\Omega_{n,\lambda}^2)\to1\).
The shifted moment condition supplies the same fourth moment
bounds as under the null. The individual covariance bounds
in the preceding proof also hold uniformly in \(i\).
Thus the quadratic form argument in
Theorem~\ref{thm:score-clt} applies to the centered alternative
scores, giving
\begin{equation*}
\frac{\|Z_{n,\lambda}^{(h)}\|_\lambda^2
-\operatorname{tr}(\Omega_{n,\lambda}^{(h)})}
{\sigma_{0,n,\lambda}}
\xrightarrow{d}N(0,1).
\end{equation*}

For the linear term,
\[
\operatorname{Var}\{
2\sqrt n\langle b_\lambda(h),Z_{n,\lambda}^{(h)}
\rangle_\lambda\}
\leq Cn\|b_\lambda(h)\|_\lambda^2.
\]
Equations \eqref{eq:noncentrality} and
\eqref{eq:mean-signal-equivalence} make this
\(O(\sigma_{0,n,\lambda})\).
After division by \(\sigma_{0,n,\lambda}^2\), it tends to zero.
Combining the three terms in
\eqref{eq:alternative-statistic-decomposition} proves the limit
and the stated power function.
\end{proof}

\begin{corollary}\label{cor:score-consistency}
Under the assumptions and rate conditions of
Theorem~\ref{thm:local-power}, replace the finite limit condition
\eqref{eq:noncentrality} by \(\Lambda_n(h)\to\infty\).
Then the rejection probability tends to one.
\end{corollary}

\begin{proof}
Equation \eqref{eq:mean-signal-equivalence} gives
\(n\|b_\lambda(h)\|_\lambda^2/
\sigma_{0,n,\lambda}\to\infty\).
The covariance centering error is
\(o(\sigma_{0,n,\lambda})\), and the centered quadratic term is
\(O_p(\sigma_{0,n,\lambda})\).
Both are negligible relative to the mean signal.
The linear term has variance at most
\(Cn\|b_\lambda(h)\|_\lambda^2\), so it is also negligible
relative to \(n\|b_\lambda(h)\|_\lambda^2\).
The decomposition \eqref{eq:alternative-statistic-decomposition}
therefore shows that the standardized statistic diverges to
\(+\infty\) in probability.
\end{proof}

\subsection{Sobolev separation and the minimax rate}

We now measure the smallest detectable signal over the ellipsoid
\[
\mathcal B_R
=
\{h\in W^{r,2}[0,1]:V(h,h)+J(h,h)\leq R^2\},
\qquad R>0.
\]
The bounds on \(I\) and \(\pi\) make \(V(h,h)^{1/2}\)
equivalent to the \(L^2\) norm.
For a given smoothing parameter, define
\begin{equation}
\eta_{n,\lambda}^2
=
s_\lambda^{2r}+\frac{1}{n s_\lambda^{1/2}}.
\label{eq:generic-separation-rate}
\end{equation}
The first term bounds the smoothing loss over a fixed ellipsoid;
the second is the signal energy needed to match the null
standard deviation after multiplication by \(n\).

\begin{lemma}\label{lem:global-identification}
Under Assumption~\ref{ass:loss-primitive}, for every fixed
\(R<\infty\), there is \(c_R>0\) such that
\begin{equation}
a\Psi_t(-a)\geq c_Ra^2
\label{eq:uniform-score-identification}
\end{equation}
uniformly in \(t\in[0,1]\) and \(a=h(t)\), \(h\in\mathcal B_R\).
Consequently,
\begin{equation}
\langle b_\lambda(h),h\rangle_\lambda\geq c_RV(h,h),
\label{eq:global-score-direction}
\end{equation}
and, for \(h\neq0\),
\begin{equation}
\|b_\lambda(h)\|_\lambda^2
\geq
c_R\frac{V(h,h)^2}{V(h,h)+\lambda J(h,h)}.
\label{eq:global-score-norm-lower}
\end{equation}
\end{lemma}

\begin{proof}
Sobolev embedding gives
\(\sup_{h\in\mathcal B_R}\|h\|_\infty\leq M_R<\infty\).
The local expansion and \(I(t)\geq c_I\) yield
\(a\Psi_t(-a)\geq c_Ia^2/2\) for \(|a|\leq a_1\),
with \(a_1>0\) independent of \(t\).
Since \(\psi\) is nondecreasing, \(\Psi_t\) is nonincreasing.
For \(a_1\leq a\leq M_R\),
\(\Psi_t(-a)\geq c_Ia_1/2\), and the corresponding negative
bound holds for \(-M_R\leq a\leq-a_1\).
Together with the local bound, this proves
\eqref{eq:uniform-score-identification}.

The integral representation of the score mean gives
\[
\langle b_\lambda(h),h\rangle_\lambda
=
\int_0^1h(t)\Psi_t\{-h(t)\}\pi(t)\,dt.
\]
Apply \eqref{eq:uniform-score-identification} and the upper
bound on \(I\) to obtain \eqref{eq:global-score-direction}.
Cauchy--Schwarz then gives
\(\|b_\lambda(h)\|_\lambda\|h\|_\lambda\geq c_RV(h,h)\).
Squaring and using
\(\|h\|_\lambda^2=V(h,h)+\lambda J(h,h)\)
proves \eqref{eq:global-score-norm-lower}.
\end{proof}

\begin{theorem}\label{thm:sparse-detection-upper}
Suppose Assumptions~\ref{ass:loss-primitive} and
\ref{ass:sparse-cluster} hold, let \(r\geq2\), and assume
\begin{equation}
\beta(2r-1)>2.
\label{eq:beta-power-condition}
\end{equation}
Choose
\begin{equation}
s_\lambda\asymp n^{-2/(4r+1)}.
\label{eq:optimal-s}
\end{equation}
Then
\begin{equation*}
\eta_{n,\lambda}\asymp n^{-2r/(4r+1)}.
\end{equation*}
For every \(\gamma\in(0,1)\) and
\(\delta\in(0,1-\gamma)\), there is
\(C=C(R,\gamma,\delta)>0\), with dependence on the fixed
model bounds suppressed, such that
\begin{equation}
\liminf_{n\to\infty}
\inf_{\substack{h\in\mathcal B_R:\\
V(h,h)^{1/2}\geq C\eta_{n,\lambda}}}
P_{g_0+h}\left\{
\frac{S_{n,\lambda}-\mu_{n,\lambda}}
{\sigma_{0,n,\lambda}}>z_{1-\gamma}
\right\}
\geq1-\delta.
\label{eq:uniform-power-separation}
\end{equation}
For any sequence \(h_n\in\mathcal B_R\) with
\(V(h_n,h_n)^{1/2}/\eta_{n,\lambda}\to\infty\),
the rejection probability tends to one.
\end{theorem}

\begin{proof}
Balancing the terms in \eqref{eq:generic-separation-rate}
gives \(s_\lambda^{2r}\asymp
(ns_\lambda^{1/2})^{-1}\).
Hence \eqref{eq:optimal-s} holds,
\(\eta_{n,\lambda}\asymp s_\lambda^r\), and
\(ns_\lambda^{2r+1/2}\asymp1\).

First consider sequences with
\(V(h_n,h_n)^{1/2}=O(\eta_{n,\lambda})\).
Sobolev interpolation and \(J(h_n,h_n)\leq R^2\) give
\[
\|h_n\|_\infty
\leq C_R V(h_n,h_n)^{(2r-1)/(4r)}
=
O(s_\lambda^{r-1/2}).
\]
Condition \eqref{eq:beta-power-condition} therefore implies
\(\|h_n\|_\infty^{\beta/2}=o(s_\lambda^{1/2})\).
Also,
\(\|h_n\|_\lambda^2=O(s_\lambda^{2r})\), so
\[
\frac{n\|h_n\|_\infty^\alpha\|h_n\|_\lambda^2}
{\sigma_{0,n,\lambda}}
=
O(\|h_n\|_\infty^\alpha)\to0.
\]
Thus the local power conditions hold for every sequence in
a bounded multiple of the separation scale.

If \(V(h,h)^{1/2}\geq C\eta_{n,\lambda}\), then
\eqref{eq:smoothed-signal-lower} yields
\(\|A_\lambda h\|_\lambda^2
\geq C^2\eta_{n,\lambda}^2-R^2s_\lambda^{2r}\).
Consequently, for fixed constants \(c_1,c_2>0\),
\begin{equation*}
\Lambda_n(h)\geq c_1C^2-c_2R^2.
\end{equation*}
Choose \(C\) so that
\(1-\Phi\{z_{1-\gamma}-(c_1C^2-c_2R^2)\}\geq1-\delta\).
For sequences whose separation ratio is bounded,
\(\Lambda_n(h_n)\) is also bounded. Every subsequence therefore
has a further subsequence along which it converges.
Theorem~\ref{thm:local-power} gives the required power bound
along each such subsequence.

Now suppose
\(V(h_n,h_n)^{1/2}/\eta_{n,\lambda}\to\infty\).
Put \(x_n=V(h_n,h_n)\).
Since \(\eta_{n,\lambda}^2\asymp\lambda\), global identification
gives
\[
\|b_\lambda(h_n)\|_\lambda^2
\geq c_R\frac{x_n^2}{x_n+\lambda R^2}
\geq c_Rx_n
\]
for large \(n\), after adjusting the constant.
It follows that
\(n\|b_\lambda(h_n)\|_\lambda^2/
\sigma_{0,n,\lambda}\to\infty\).

We next compare the covariance shift with this mean signal.
If \(x_n\to0\), interpolation and
Lemma~\ref{lem:alternative-covariance} give
\[
|\operatorname{tr}(\Omega_{n,\lambda}^{(h_n)})
-\mu_{n,\lambda}|
\leq Cs_\lambda^{-1}x_n^a,
\qquad
a=\frac{\beta(2r-1)}{8r}.
\]
Here \(0<a<1\). Writing \(x_n=s_\lambda^{2r}L_n\),
where \(L_n\to\infty\), the ratio of this bound to \(nx_n\)
is at most a constant times
\[
s_\lambda^{\,2ra-1/2}L_n^{a-1}\to0,
\]
because \eqref{eq:beta-power-condition} gives \(2ra>1/2\).
If \(x_n\) is bounded away from zero, the shifted moment bound
gives a covariance shift of order \(s_\lambda^{-1}\),
which is negligible relative to \(nx_n\) since
\(ns_\lambda\to\infty\).
These two cases cover arbitrary sequences by subsequence
selection.

The shifted moment bounds and Sobolev embedding also give,
uniformly over \(\mathcal B_R\),
\(\Omega_{i,\lambda}^{(h)}\preceq C_RA_\lambda\) and
\(E\|\widetilde G_{i,\lambda}^{(h)}\|_\lambda^4
\leq C_Rs_\lambda^{-2}\).
The quadratic decomposition used in Section~\ref{sec:asymptotic-theory}
therefore yields
\[
\operatorname{Var}\{\|Z_{n,\lambda}^{(h)}\|_\lambda^2\}
\leq C_R(s_\lambda^{-1}+n^{-1}s_\lambda^{-2})
=O(s_\lambda^{-1}).
\]
Thus its centered version is
\(O_p(s_\lambda^{-1/2})\), uniformly over the ellipsoid.
The linear term in
\eqref{eq:alternative-statistic-decomposition} has variance
at most \(C_Rn\|b_\lambda(h)\|_\lambda^2\).
Both stochastic terms are negligible relative to the mean
signal along the present sequence. Hence the standardized
statistic diverges to \(+\infty\).

Finally, any sequence of alternatives in the stated separation
set has a subsequence with bounded separation ratios or a
further subsequence whose ratios diverge.
The two arguments above establish
\eqref{eq:uniform-power-separation}.
\end{proof}

The matching lower bound follows from Gaussian regression with
one observation per subject. Since this experiment is contained
in the sparse model, a lower bound for it also limits what can
be achieved over the larger class.

\begin{theorem}\label{thm:minimax-lower}
Consider the submodel with \(m_i=1\), \(Z_i=0\), independent
standard normal measurement errors, and independent uniform
design points on \([0,1]\).
Let \(\rho_n=o\{n^{-2r/(4r+1)}\}\).
No sequence of asymptotic level-\(\gamma\) tests is uniformly
consistent over
\[
\{h\in\mathcal B_R:V(h,h)^{1/2}\geq\rho_n\}.
\]
Thus the minimax separation rate over any sparse model class
containing this submodel is bounded below by
\begin{equation*}
n^{-2r/(4r+1)}.
\end{equation*}
\end{theorem}

\begin{proof}
Subtracting the known null function reduces the observations to
\[
Y_i=h(T_i)+\varepsilon_i,
\qquad i=1,\ldots,n,
\]
with independent uniform design and standard normal errors.
The bounds on \(I\) make the ellipsoid constraint and separation
norm equivalent, up to fixed constants, to their usual Sobolev
and \(L^2\) counterparts.

Restrict to a fixed-radius periodic Sobolev ball contained in $\mathcal B_R$. The random uniform design regression lower bound of \citet{ingstersapatinas2009}, applied with dimension one and smoothness $r$, gives the separation rate $n^{-2r/(4r+1)}$ on this smaller class; see also \citet{ingstersuslina2003}. Norm equivalence changes only fixed constants.
Therefore no asymptotic level-\(\gamma\) test can be uniformly
consistent at a smaller order in this submodel.
A uniformly consistent test over the full sparse model would
also be uniformly consistent on the submodel, giving the
claimed lower bound.
\end{proof}

Together, Theorems~\ref{thm:sparse-detection-upper} and
\ref{thm:minimax-lower} identify the score test separation rate
as \(n^{-2r/(4r+1)}\), under the stated conditions.
The upper bound concerns the population calibrated test;
transferring it to estimated or bootstrap critical values also
requires control of those critical values under alternatives.

For the objective improvement statistic, Section~\ref{sec:asymptotic-theory} establishes
equivalence under the null. At the testing bandwidth, its rate
conditions simplify as follows.

\begin{corollary}\label{cor:objective-minimax}
Suppose the assumptions of Theorem~\ref{thm:statistic-equivalence}
hold in the Lipschitz score case, or those of
Corollary~\ref{cor:absolute-equivalence} hold in the absolute
loss case, apart from their bandwidth conditions.
Assume \(\alpha(4r-3)>2\) and choose \(s_\lambda\) as in
\eqref{eq:optimal-s}.
Then, under \(H_0\),
\[
T_{n,\lambda}-S_{n,\lambda}=o_p(s_\lambda^{-1/2})
\]
for Lipschitz scores when \(r\geq2\), and for absolute loss
when \(r\geq3\).
In these cases the objective statistic has the same standardized
null limit as the score statistic.
\end{corollary}

\begin{proof}
At \eqref{eq:optimal-s},
\[
n^\alpha s_\lambda^{2\alpha+1}
\asymp n^{\{\alpha(4r-3)-2\}/(4r+1)}\to\infty.
\]
For Lipschitz scores, the other required condition is
\(ns_\lambda^4\asymp n^{(4r-7)/(4r+1)}\to\infty\),
which holds for integer \(r\geq2\).
For absolute loss, it is
\(ns_\lambda^6\asymp n^{(4r-11)/(4r+1)}\to\infty\),
which holds for integer \(r\geq3\).
The respective null equivalence results complete the proof.
\end{proof}

A matching power conclusion for \(T_{n,\lambda}\) requires
an expansion valid under alternatives. In particular, transferring
the separation rate requires uniform control near the detection
boundary and a consistency argument for larger signals.
Null equivalence alone does not supply these results.

\section{Sparse to dense transition}
\label{sec:growing-m}

We now allow the number of measurements per subject to increase.
The score covariance then has two components: one from individual
measurements and another from dependence within a trajectory.
Their relative sizes determine the variance transition and the
corresponding detection rate. We retain the loss, design density,
and Sobolev conditions used earlier, but replace the bounded
cluster size assumption by the conditions below. All covariance
calculations are under the null unless an alternative is specified.

\subsection{Covariance decomposition and variance transition}

Recall that \(m_H=\{n^{-1}\sum_i m_i^{-1}\}^{-1}\).
Throughout this section, \(m_H\to\infty\), and individual cluster
sizes remain comparable to \(m_H\). For random cluster sizes,
the analysis is conditional on their observed values.

\begin{assumption}\label{ass:growing-cluster}
There are constants \(0<c_m<C_m<\infty\) such that
\(c_m m_H\leq m_i\leq C_m m_H\) for every subject.
Conditional on \(m_i\), the observation times are independent
and identically distributed with density \(\pi\), independently
of the latent trajectory and measurement errors.

The conditional second moment
\(q_\psi(t)=E\{\psi^2(e_{ij})\mid T_{ij}=t\}\)
is common across subjects and indices and is bounded above
and away from zero. For \(j\neq k\), the conditional pair moment
\[
c_\psi(s,t)
=
E\{\psi(e_{ij})\psi(e_{ik})
\mid T_{ij}=s,T_{ik}=t\}
\]
is common across subjects, distinct index pairs, and cluster
sizes, and satisfies \(c_\psi(s,t)=c_\psi(t,s)\).
The shifted moment condition in
Assumption~\ref{ass:loss-primitive} holds uniformly over subjects
and measurement indices.

For the analysis under alternatives, the conditional marginal
and pair distributions are also common across subjects and
indices. Alternatives change the mean function while leaving
the joint law of the centered residuals and design unchanged.
\end{assumption}

Write \(L^2(V)=L^2([0,1],I(t)\pi(t)\,dt)\).
Define the pair score operator \(\mathcal C_\psi\) on this space by
\[
V(\mathcal C_\psi f,g)
=
\int_0^1\int_0^1
c_\psi(s,t)f(s)g(t)\pi(s)\pi(t)\,ds\,dt.
\]
Throughout this section, assume that \(\mathcal C_\psi\) is
nonnegative, self adjoint, and Hilbert--Schmidt, with
\begin{equation}
0<\operatorname{tr}(\mathcal C_\psi^2)<\infty.
\label{eq:Cpsi-nondegenerate}
\end{equation}
The positive lower bound excludes the case without pair score
dependence. If \(c_\psi\equiv0\), the trajectory component below
vanishes.

On \(\mathcal H_\lambda\), define the diagonal operator by
\[
\langle D_\lambda f,g\rangle_\lambda
=
\int_0^1q_\psi(t)f(t)g(t)\pi(t)\,dt.
\]
The pair operator is defined by
\[
\langle B_\lambda f,g\rangle_\lambda
=
\int_0^1\int_0^1
c_\psi(s,t)f(s)g(t)\pi(s)\pi(t)\,ds\,dt.
\]
These definitions distinguish the operators on
\(\mathcal H_\lambda\) from \(\mathcal C_\psi\), which acts on
\(L^2(V)\). Their relationship is made explicit in the proof
below.

\begin{proposition}\label{prop:exact-dense-decomposition}
Under Assumption~\ref{ass:loss-primitive},
Assumption~\ref{ass:growing-cluster}, and the stated conditions
on \(\mathcal C_\psi\),
\begin{equation}
\Omega_{n,\lambda}
=
m_H^{-1}D_\lambda+(1-m_H^{-1})B_\lambda.
\label{eq:exact-dense-omega}
\end{equation}
Moreover,
\(\operatorname{tr}(D_\lambda)\asymp s_\lambda^{-1}\) and
\(\operatorname{tr}(D_\lambda^2)\asymp s_\lambda^{-1}\).
For every \(\lambda>0\), \(B_\lambda\) is trace class, and
\begin{equation}
\operatorname{tr}(B_\lambda^2)
\to\operatorname{tr}(\mathcal C_\psi^2),
\qquad \lambda\to0.
\label{eq:B-trace2}
\end{equation}
\end{proposition}

\begin{proof}
Separating the diagonal and distinct index pairs in the subject
score covariance gives
\[
\Omega_{i,\lambda}
=
m_i^{-1}D_\lambda+(1-m_i^{-1})B_\lambda.
\]
Averaging over subjects proves \eqref{eq:exact-dense-omega}.

On \(L^2(V)\), let \(M_qf=(q_\psi/I)f\), and define
\(R_\lambda\phi_\nu
=(1+\lambda\gamma_\nu)^{-1/2}\phi_\nu\).
The bounds on \(q_\psi/I\) imply
\(c\,\mathrm{Id}\preceq M_q\preceq C\,\mathrm{Id}\).

The map sending the orthonormal basis
\(e_{\nu,\lambda}\) of \(\mathcal H_\lambda\) to \(\phi_\nu\)
is unitary. Under this map, \(D_\lambda\) corresponds to
\(R_\lambda M_qR_\lambda\), and \(A_\lambda\) corresponds to
\(R_\lambda^2\). Hence
\[
cR_\lambda^2
\preceq R_\lambda M_qR_\lambda
\preceq CR_\lambda^2.
\]
The min--max principle and Lemma~\ref{lem:kernel-orders}
give the two trace orders for \(D_\lambda\).

Under the same identification, \(B_\lambda\) corresponds to
\(R_\lambda\mathcal C_\psi R_\lambda\).
For fixed \(\lambda>0\), \(R_\lambda\) is Hilbert--Schmidt,
so this product is trace class. Also,
\(R_\lambda\to\mathrm{Id}\) strongly and
\(\|R_\lambda\|_{\mathrm{op}}\leq1\).
Finite rank approximation of \(\mathcal C_\psi\) therefore gives
\[
\|R_\lambda\mathcal C_\psi R_\lambda-\mathcal C_\psi
\|_{\mathrm{HS}}\to0.
\]
Taking squared Hilbert--Schmidt norms proves \eqref{eq:B-trace2}.
\end{proof}

\begin{theorem}\label{thm:dense-variance-order}
Under the conditions of
Proposition~\ref{prop:exact-dense-decomposition}, as
\(\lambda\to0\) and \(m_H\to\infty\),
\begin{equation*}
\operatorname{tr}(\Omega_{n,\lambda}^2)
\asymp m_H^{-2}s_\lambda^{-1}+1.
\end{equation*}
The exact decomposition is
\begin{equation}
\begin{aligned}
\operatorname{tr}(\Omega_{n,\lambda}^2)
={}&m_H^{-2}\operatorname{tr}(D_\lambda^2)\\
&+2m_H^{-1}(1-m_H^{-1})
\operatorname{tr}(D_\lambda B_\lambda)\\
&+(1-m_H^{-1})^2\operatorname{tr}(B_\lambda^2).
\end{aligned}
\label{eq:dense-variance-exact}
\end{equation}
The measurement component dominates when
\(m_H^2s_\lambda\to0\), and the trajectory component dominates
when \(m_H^2s_\lambda\to\infty\).
The two have comparable orders when \(m_H^2s_\lambda\asymp1\).
\end{theorem}

\begin{proof}
Squaring \eqref{eq:exact-dense-omega} and using cyclicity of
the trace gives \eqref{eq:dense-variance-exact}.
Positivity makes the cross term nonnegative, while
Hilbert--Schmidt Cauchy--Schwarz gives
\[
0\leq\operatorname{tr}(D_\lambda B_\lambda)
\leq
\{\operatorname{tr}(D_\lambda^2)
\operatorname{tr}(B_\lambda^2)\}^{1/2}
=O(s_\lambda^{-1/2}).
\]
Thus the trace is bounded above by
\(C\{m_H^{-2}s_\lambda^{-1}
+m_H^{-1}s_\lambda^{-1/2}+1\}\).
The middle term is at most
\(\{m_H^{-2}s_\lambda^{-1}+1\}/2\).

For the lower bound, retain the two nonnegative squared terms
in \eqref{eq:dense-variance-exact}.
Proposition~\ref{prop:exact-dense-decomposition} and
\eqref{eq:Cpsi-nondegenerate} give
\(c\{m_H^{-2}s_\lambda^{-1}+1\}\) for large \(n\).
Their relative order is \((m_H^2s_\lambda)^{-1}\), which
identifies the three regimes.
\end{proof}

\begin{remark}\label{rem:dense-null-law}
The trace calculation explains why the sparse normal limit
does not automatically extend to dense sampling.
In the measurement dominated regime, the covariance operator
norm remains bounded while
\(\operatorname{tr}(\Omega_{n,\lambda}^2)\to\infty\).
Thus
\[
\frac{\|\Omega_{n,\lambda}\|_{\mathrm{op}}^2}
{\operatorname{tr}(\Omega_{n,\lambda}^2)}\to0.
\]
This preserves the spectral dispersion condition, although a
central limit theorem still requires the relevant moment bounds.

In the trajectory dominated regime, the transformed covariance
converges in Hilbert--Schmidt norm to the nonzero compact
operator \(\mathcal C_\psi\). Its leading eigenvalues need not
be negligible. A centered weighted chi square limit may then
replace the normal limit. Related changes in asymptotic behavior
appear in \citet{wangetal2018} and \citet{caihu2025}.
Calibration must therefore accommodate the sampling regime.
\end{remark}

\subsection{Covariance stability and separation rates}

To study power, we control the covariance change caused by a
mean departure \(h\). We then use the exact null quantile to
derive detection rates without imposing a normal limit in
the dense regime.

For \(h\in\mathcal B_R\), define
\[
Y_{ij}^{(h)}
=
\psi\{e_{ij}+h(T_{ij})\}K_\lambda(T_{ij},\cdot).
\]
Its mean is \(b_\lambda(h)\).
Let \(D_{\lambda,h}\) be the covariance operator of
\(Y_{ij}^{(h)}\), and, for \(j\neq k\), let
\[
B_{\lambda,h}
=
E[(Y_{ij}^{(h)}-b_\lambda(h))
\otimes(Y_{ik}^{(h)}-b_\lambda(h))].
\]
The common marginal and pair distributions give
\[
\Omega_{n,\lambda}^{(h)}
=
m_H^{-1}D_{\lambda,h}
+(1-m_H^{-1})B_{\lambda,h}.
\]
At \(h=0\), these operators reduce to their null counterparts.

\begin{lemma}\label{lem:dense-covariance-perturbation}
Under the conditions of
Proposition~\ref{prop:exact-dense-decomposition}, for
\(\|h\|_\infty\leq a_0\),
\begin{equation}
\|D_{\lambda,h}-D_\lambda\|_{\mathrm{HS}}
\leq C\|h\|_\infty^{\beta/2}s_\lambda^{-1/2},
\label{eq:D-perturb-HS}
\end{equation}
and
\begin{equation}
\|B_{\lambda,h}-B_\lambda\|_{\mathrm{HS}}
\leq C\|h\|_\infty^{\beta/2}.
\label{eq:B-perturb-HS}
\end{equation}
Consequently,
\begin{equation}
\|\Omega_{n,\lambda}^{(h)}-\Omega_{n,\lambda}\|_{\mathrm{HS}}
\leq C\|h\|_\infty^{\beta/2}
\{m_H^{-1}s_\lambda^{-1/2}+1\}.
\label{eq:Omega-perturb-HS}
\end{equation}
The corresponding trace bounds are
\begin{equation}
|\operatorname{tr}(\Omega_{n,\lambda}^{(h)})
-\operatorname{tr}(\Omega_{n,\lambda})|
\leq C\|h\|_\infty^{\beta/2}
\{m_H^{-1}s_\lambda^{-1}+s_\lambda^{-1/2}\},
\label{eq:dense-trace-perturb}
\end{equation}
and
\begin{equation}
\begin{aligned}
|\operatorname{tr}\{(\Omega_{n,\lambda}^{(h)})^2\}
-\operatorname{tr}(\Omega_{n,\lambda}^2)|
\leq{}&
C\|h\|_\infty^{\beta/2}\\
&\times
\{m_H^{-2}s_\lambda^{-1}
+m_H^{-1}s_\lambda^{-1/2}+1\}.
\end{aligned}
\label{eq:dense-trace2-perturb}
\end{equation}
\end{lemma}

\begin{proof}
Write
\(\Delta_h(t)=\psi\{e+h(t)\}-\psi(e)\).
The score increment condition gives
\(E\{\Delta_h(t)^2\mid T=t\}\leq C|h(t)|^\beta\).
Conditional Cauchy--Schwarz and the shifted moment bound imply
\[
|q_{\psi,h}(t)-q_\psi(t)|
\leq C\|h\|_\infty^{\beta/2},
\]
where \(q_{\psi,h}(t)
=E[\psi^2\{e+h(t)\}\mid T=t]\).

The uncentered diagonal operators correspond on \(L^2(V)\) to
\(R_\lambda M_{q,h}R_\lambda\) and
\(R_\lambda M_qR_\lambda\), where
\(M_{q,h}f=(q_{\psi,h}/I)f\).
The multiplication operator difference has norm
\(O(\|h\|_\infty^{\beta/2})\).
The Schatten norm product inequality therefore gives
\[
\|R_\lambda(M_{q,h}-M_q)R_\lambda\|_{\mathrm{HS}}
\leq
C\|h\|_\infty^{\beta/2}
\{\operatorname{tr}(R_\lambda^4)\}^{1/2}
\leq C\|h\|_\infty^{\beta/2}s_\lambda^{-1/2}.
\]
The score increment bound also implies
\(\|b_\lambda(h)\|_\lambda
\leq C\|h\|_\infty^{\beta/2}\), by testing its mean against
unit vectors in \(\mathcal H_\lambda\).
Subtracting the rank one operator
\(b_\lambda(h)\otimes b_\lambda(h)\) therefore preserves
the stated order and proves \eqref{eq:D-perturb-HS}.

For distinct measurements, expand the difference between the
shifted and null score products. Design independence allows
the conditional marginal bounds to be used after conditioning
on both times. Cauchy--Schwarz gives
\[
|c_{\psi,h}(s,t)-c_\psi(s,t)|
\leq C\{|h(s)|^{\beta/2}+|h(t)|^{\beta/2}\}.
\]
The corresponding operators on \(L^2(V)\) consequently differ
in Hilbert--Schmidt norm by
\(O(\|h\|_\infty^{\beta/2})\).
Conjugating by the contraction \(R_\lambda\) and subtracting
the same rank one centering term proves
\eqref{eq:B-perturb-HS}.
The covariance decomposition yields
\eqref{eq:Omega-perturb-HS}.

For the diagonal trace, the multiplication operator bound gives
\(C\|h\|_\infty^{\beta/2}\operatorname{tr}(R_\lambda^2)\).
For the pair trace, cyclicity and Cauchy--Schwarz give
\[
|\operatorname{tr}\{
(\mathcal C_{\psi,h}-\mathcal C_\psi)R_\lambda^2\}|
\leq
\|\mathcal C_{\psi,h}-\mathcal C_\psi\|_{\mathrm{HS}}
\|R_\lambda^2\|_{\mathrm{HS}}
\leq C\|h\|_\infty^{\beta/2}s_\lambda^{-1/2}.
\]
Including the centering corrections and the cluster weights
proves \eqref{eq:dense-trace-perturb}.

Finally, for self adjoint Hilbert--Schmidt operators,
\[
|\operatorname{tr}(A^2)-\operatorname{tr}(B^2)|
\leq(\|A\|_{\mathrm{HS}}+\|B\|_{\mathrm{HS}})
\|A-B\|_{\mathrm{HS}}.
\]
Apply this inequality with
\eqref{eq:Omega-perturb-HS} and
Theorem~\ref{thm:dense-variance-order} to obtain
\eqref{eq:dense-trace2-perturb}.
\end{proof}

Let \(c_{n,\gamma}(\lambda)\) be the \((1-\gamma)\) null
quantile of \(S_{n,\lambda}\).
The population calibrated test rejects when
\(S_{n,\lambda}>c_{n,\gamma}(\lambda)\).
This quantile is used to isolate the detection rate from the
separate problem of estimating critical values.

A compact expression covering all sampling regimes is
\begin{equation}
\rho_n^*
=
\max\{(nm_H)^{-2r/(4r+1)},\,n^{-1/2}\}.
\label{eq:unified-target-rate}
\end{equation}
Choose
\begin{equation}
s_n^*\asymp
\max\{(nm_H)^{-2/(4r+1)},\,n^{-1/(2r)}\},
\label{eq:unified-optimal-s}
\end{equation}
and set \(\lambda=(s_n^*)^{2r}\).
The two terms coincide in order at
\(m_H\asymp n^{1/(4r)}\).

\begin{theorem}\label{thm:unified-upper-rate}
Let \(r\geq2\) and \(\beta\geq1\), and suppose the conditions
of Proposition~\ref{prop:exact-dense-decomposition} hold.
Let \(M_n\to\infty\) satisfy \(M_n\rho_n^*\to0\), and choose
\(\lambda=(s_n^*)^{2r}\). Then
\begin{equation}
\inf_{\substack{h\in\mathcal B_R:\\
V(h,h)^{1/2}\geq M_n\rho_n^*}}
P_{g_0+h}\{S_{n,\lambda}>c_{n,\gamma}(\lambda)\}
\to1.
\label{eq:unified-upper-power}
\end{equation}
Thus the consistency separation rate is
\((nm_H)^{-2r/(4r+1)}\) when
\(m_H\ll n^{1/(4r)}\), and \(n^{-1/2}\) when
\(m_H\gg n^{1/(4r)}\).
At the transition, both expressions have order \(n^{-1/2}\).
\end{theorem}

\begin{proof}
We first bound the null critical value.
Use the centered decomposition
\(S_{n,\lambda}-\mu_{n,\lambda}
=D_{n,\lambda}+U_{n,\lambda}\)
from \eqref{eq:score-decomposition}.
The two terms are uncorrelated. The exact covariance
decomposition and comparability of cluster sizes give
\[
\operatorname{Var}(U_{n,\lambda})
=
2\operatorname{tr}(\Omega_{n,\lambda}^2)
+O\{n^{-1}(m_H^{-2}s_\lambda^{-1}+1)\}.
\]
The kernel bound, Jensen's inequality, and the fourth score
moment bound also give
\(\operatorname{Var}(D_{n,\lambda})
\leq Cn^{-1}s_\lambda^{-2}\).
At \eqref{eq:unified-optimal-s},
\begin{equation}
n^{-1}s_\lambda^{-2}
=o\{m_H^{-2}s_\lambda^{-1}+1\}.
\label{eq:dense-diagonal-negligible-optimal}
\end{equation}
Indeed, below the transition the relevant ratio is
\(m_H^2/(ns_\lambda)\), whose largest order at the transition
is \(n^{1/r-1}\). Above the transition,
\(n^{-1}s_\lambda^{-2}\asymp n^{1/r-1}\).
Both tend to zero for \(r\geq2\).

Consequently, with
\(v_{n,\lambda}^2=\operatorname{Var}_0(S_{n,\lambda})\),
we have
\(v_{n,\lambda}^2\asymp m_H^{-2}s_\lambda^{-1}+1\).
Chebyshev's inequality yields
\begin{equation}
c_{n,\gamma}(\lambda)
\leq\mu_{n,\lambda}+\gamma^{-1/2}v_{n,\lambda}.
\label{eq:null-quantile-chebyshev}
\end{equation}
At the chosen bandwidth,
\(\lambda\asymp(\rho_n^*)^2\) and
\(v_{n,\lambda}/n\asymp(\rho_n^*)^2\).

Write \(x=V(h,h)\). The global identification argument in
Lemma~\ref{lem:global-identification} uses only the marginal
score conditions, so it remains valid here:
\[
\|b_\lambda(h)\|_\lambda^2
\geq c_R\frac{x^2}{x+\lambda R^2}.
\]
Uniformly over the separation set,
\(x/\lambda\to\infty\). Hence
\(\|b_\lambda(h)\|_\lambda^2\geq c_Rx\) for large \(n\), and
\(n\|b_\lambda(h)\|_\lambda^2/v_{n,\lambda}\to\infty\).

For shrinking alternatives, Sobolev interpolation gives
\begin{equation*}
\|h\|_\infty
\leq C_R x^{(2r-1)/(4r)}.
\end{equation*}
Put \(a=\beta(2r-1)/(8r)\), so \(0<a<1\).
By \eqref{eq:dense-trace-perturb},
the covariance centering error is at most
\[
Cx^a\{m_H^{-1}s_\lambda^{-1}+s_\lambda^{-1/2}\}
\leq Cx^a s_\lambda^{-1/2}v_{n,\lambda}.
\]
Write \(x=s_\lambda^{2r}L\).
On the separation set, \(L\to\infty\) uniformly.
Since \(v_{n,\lambda}\asymp ns_\lambda^{2r}\), the ratio
of the preceding bound to \(nx\) is at most
\[
C s_\lambda^{2ra-1/2}L^{a-1}\to0.
\]
Here \(2ra>1/2\) follows from
\(\beta(2r-1)>2\), which holds when \(r\geq2\) and
\(\beta\geq1\).

The perturbation bounds also give
\begin{equation*}
\operatorname{tr}\{(\Omega_{n,\lambda}^{(h)})^2\}
=
O(m_H^{-2}s_\lambda^{-1}+1)
\end{equation*}
for shrinking alternatives.
Together with the shifted fourth moment bound and
\eqref{eq:dense-diagonal-negligible-optimal}, this shows that
the centered quadratic term in
\eqref{eq:alternative-statistic-decomposition} has standard
deviation \(O(v_{n,\lambda})\).
The linear term has variance at most
\(Cn\|b_\lambda(h)\|_\lambda^2\).
Both are negligible relative to
\(n\|b_\lambda(h)\|_\lambda^2\).

If \(x\) is bounded away from zero, Sobolev embedding bounds
\(\|h\|_\infty\) uniformly over \(\mathcal B_R\).
The shifted moment condition then gives a covariance trace
bound \(O(s_\lambda^{-1})=o(n)\), an operator norm bound
\(O(1)\), and a centered quadratic variance bound
\(O(s_\lambda^{-1}+n^{-1}s_\lambda^{-2})=o(n^2)\).
Global identification supplies a mean signal bounded below
by a positive multiple of \(n\).
Consistency follows in this case as well.

These bounds, the alternative decomposition, and
\eqref{eq:null-quantile-chebyshev} prove
\eqref{eq:unified-upper-power}.
Equivalently, the sufficient signal condition is
\begin{equation*}
V(h,h)\gg
s_\lambda^{2r}
+\frac1n\{m_H^{-2}s_\lambda^{-1}+1\}^{1/2}.
\end{equation*}
Balancing the two terms below the transition gives
\(s_\lambda\asymp(nm_H)^{-2/(4r+1)}\).
Above it, \(s_\lambda\asymp n^{-1/(2r)}\) makes both terms
of order \(n^{-1}\). This gives the stated rates.
\end{proof}

The lower bound has two sources. Even if the subject effects
were observed, estimation of a nonconstant signal would retain
the usual regression testing difficulty. Constant alternatives,
in contrast, cannot be distinguished faster than the rate
allowed by the number of independent subjects.

\begin{theorem}\label{thm:unified-minimax-lower}
Suppose the model class contains the Gaussian submodel
\[
U_{ij}=g_0(T_{ij})+h(T_{ij})+Z_i+\varepsilon_{ij},
\]
where \(Z_i\) are independent \(N(0,\tau^2)\) variables for a
fixed \(\tau^2>0\), the errors are independent \(N(0,1)\),
and the design points have density \(\pi\).
All components are mutually independent apart from the common
subject effect \(Z_i\).

If \(\rho_n=o(\rho_n^*)\), no sequence of asymptotic
level-\(\gamma\) tests is uniformly consistent over
\[
\{h\in\mathcal B_R:V(h,h)^{1/2}\geq\rho_n\}.
\]
Together with Theorem~\ref{thm:unified-upper-rate}, this gives
the minimax consistency separation rate
\eqref{eq:unified-target-rate}.
\end{theorem}

\begin{proof}
First consider the more informative experiment in which all
\(Z_i\) are observed. Subtracting them and the known null
function gives
\[
U_{ij}-g_0(T_{ij})-Z_i
=
h(T_{ij})+\varepsilon_{ij}.
\]
There are $N_n=\sum_i m_i\asymp nm_H$ independent random design regression observations. To reduce the design to the uniform case, put $F(t)=\int_0^t\pi(u)\,du$ and write $h(t)=f\{F(t)\}$. The assumed smoothness and positive bounds on $\pi$ make this composition a bounded isomorphism on $W^{r,2}[0,1]$. A sufficiently small fixed-radius periodic Sobolev ball for $f$ therefore maps into $\mathcal B_R$, with equivalent separation norm. The additional observed subject effects are independent of $h$ and provide no information in this regression experiment.
The Sobolev testing lower bound for this Gaussian experiment
has order \(N_n^{-2r/(4r+1)}\); see
\citet{ingstersapatinas2009} and
\citet{ingstersuslina2003}.
Impossibility in the more informative experiment implies
impossibility for the original data, giving the lower bound
\((nm_H)^{-2r/(4r+1)}\).

For the second bound, restrict to constant alternatives
\(h_\theta(t)=\theta\).
The subject covariance is
\(\Sigma_i=\mathrm{Id}_{m_i}
+\tau^2\mathbf1_{m_i}\mathbf1_{m_i}^{\prime}\).
The joint Kullback--Leibler divergence is
\[
\mathrm{KL}(P_\theta,P_0)
=
\frac{\theta^2}{2}\sum_{i=1}^n
\mathbf1_{m_i}^{\prime}\Sigma_i^{-1}\mathbf1_{m_i}.
\]
The Sherman--Morrison formula gives
\(\mathbf1_{m_i}^{\prime}\Sigma_i^{-1}\mathbf1_{m_i}
=m_i/(1+\tau^2m_i)\leq\tau^{-2}\).
Thus
\[
\mathrm{KL}(P_\theta,P_0)
\leq\frac{n\theta^2}{2\tau^2}.
\]
For \(|\theta|=o(n^{-1/2})\), Pinsker's inequality makes the
total variation distance tend to zero.
Any asymptotic level-\(\gamma\) test then has power at most
\(\gamma+o(1)\) along such alternatives.
Since \(V(h_\theta,h_\theta)\asymp\theta^2\), this gives the
\(n^{-1/2}\) lower bound.

Combining the two bounds gives their maximum,
\(\rho_n^*\). The upper result establishes consistency above
diverging multiples of this rate.
\end{proof}

\begin{remark}\label{rem:transition-interpretation}
The boundary \(m_H\asymp n^{1/(4r)}\) is a testing transition.
Below it, measurement variability determines the rate, which
matches regression testing with effective sample size \(nm_H\).
Above it, denser measurements improve information about each
trajectory, but the number of independent trajectories remains
\(n\). The random subject effect in the lower bound produces
the resulting \(n^{-1/2}\) barrier.

Related transitions were studied by \citet{wangetal2018}
for functional testing and by \citet{caihu2025} for simultaneous
mean inference. The boundary here depends on the quadratic
statistic and Sobolev order, so it need not coincide with the
boundaries for other procedures.

\end{remark}

\section{Uniform validity of the subject level multiplier bootstrap}
\label{sec:bootstrap-validity}

When the trajectory covariance dominates, the centered score statistic need not have an asymptotically Gaussian null distribution. The subject level multiplier bootstrap is therefore designed to approximate the full quadratic null law rather than only its first two moments.

We first work with the population Hilbert space scores. In the growing cluster case, we retain all conditions of Proposition~\ref{prop:exact-dense-decomposition}, including positivity and nondegeneracy of the pair score operator. This convention also applies in Section~\ref{sec:finite-basis-theory}. Uniformity of the distributional approximation refers to its argument $x$; uniformity over data generating laws would require additional bounds.

Center the scores as $\widetilde G_{i,\lambda}=G_{i,\lambda}-\overline G_{n,\lambda}$ and define
\[
\widehat\Omega_{n,\lambda}
=\frac1n\sum_{i=1}^n\widetilde G_{i,\lambda}\otimes\widetilde G_{i,\lambda}.
\]
Its trace is $\widehat\mu_{n,\lambda}=n^{-1}\sum_i\|\widetilde G_{i,\lambda}\|_\lambda^2$.
For independent $N(0,1)$ multipliers $\xi_i$, independent of the data, set
\[
\mathbb Z_{n,\lambda}^*=\frac1{\sqrt n}\sum_{i=1}^n\xi_i\widetilde G_{i,\lambda}.
\]
The bootstrap statistic is $S_{n,\lambda}^*=\|\mathbb Z_{n,\lambda}^*\|_\lambda^2$. Conditional on the data, the score sum is Gaussian with covariance $\widehat\Omega_{n,\lambda}$. Hence
\[
\begin{aligned}
E^*(S_{n,\lambda}^*)&=\widehat\mu_{n,\lambda},\\
\operatorname{Var}^*(S_{n,\lambda}^*)&=2\operatorname{tr}(\widehat\Omega_{n,\lambda}^2).
\end{aligned}
\]
Here $P^*$ and $E^*$ condition on the observed data. Write $a_{n,\lambda}^2=2\operatorname{tr}(\Omega_{n,\lambda}^2)$ for the population quadratic fluctuation scale.
The bootstrap approximation will be established under
\begin{equation}
\chi_{n,\lambda}
=
\frac{n^{-1}s_\lambda^{-2}}
{m_H^{-2}s_\lambda^{-1}+1}
\to0.
\label{eq:bootstrap-rate-condition}
\end{equation}
The numerator is the order of the subject diagonal fourth moment contribution, whereas the denominator is the order of the leading quadratic variance from Theorem~\ref{thm:dense-variance-order}. Thus \eqref{eq:bootstrap-rate-condition} requires the former to be negligible relative to the latter.

When the cluster sizes remain bounded, $m_H\asymp1$, so
\eqref{eq:bootstrap-rate-condition} reduces to $ns_\lambda\to\infty$.
For growing trajectories, the same ratio appears in
\eqref{eq:dense-diagonal-negligible-optimal}. Hence the bootstrap rate condition is the natural extension of the variance negligibility condition across the sparse to dense regimes.

\begin{lemma}
\label{lem:bootstrap-covariance-consistency}
Suppose Assumption~\ref{ass:loss-primitive} holds and either
Assumption~\ref{ass:sparse-cluster} or Assumption~\ref{ass:growing-cluster}
is satisfied. If $s_\lambda\to0$ and \eqref{eq:bootstrap-rate-condition} holds, then under
$H_0$,
\begin{equation}
\frac{|\widehat\mu_{n,\lambda}-\mu_{n,\lambda}|}
{a_{n,\lambda}}
\xrightarrow{p}0,
\label{eq:bootstrap-mean-consistency}
\end{equation}
\begin{equation}
\frac{\|\widehat\Omega_{n,\lambda}-\Omega_{n,\lambda}\|_{\mathrm{HS}}}
{a_{n,\lambda}}
\xrightarrow{p}0,
\label{eq:bootstrap-HS-consistency}
\end{equation}
and
\begin{equation}
\frac{\operatorname{tr}(\widehat\Omega_{n,\lambda}^2)}
{\operatorname{tr}(\Omega_{n,\lambda}^2)}
\xrightarrow{p}1.
\label{eq:bootstrap-trace2-consistency}
\end{equation}
\end{lemma}

\begin{proof}
The argument used in Lemma~\ref{lem:G-fourth} relies on Jensen's
inequality for the within subject average and does not require bounded
cluster size. Hence, uniformly in $i$,
\begin{equation}
E\|G_{i,\lambda}\|_\lambda^4
\leq Cs_\lambda^{-2}.
\label{eq:g-fourth-bootstrap}
\end{equation}

Consider first the uncentered empirical covariance
\[
\widetilde\Omega_{n,\lambda}^{\,0}
=
\frac1n\sum_{i=1}^n
G_{i,\lambda}\otimes G_{i,\lambda}.
\]
Since $\|u\otimes u\|_{\mathrm{HS}}=\|u\|_\lambda^2$ and the subjects
are independent,
\[
E\big\|
\widetilde\Omega_{n,\lambda}^{\,0}
-\Omega_{n,\lambda}
\big\|_{\mathrm{HS}}^2
\leq
\frac1{n^2}\sum_{i=1}^n
E\|G_{i,\lambda}\|_\lambda^4
\leq Cn^{-1}s_\lambda^{-2}.
\]
In both the sparse and growing cluster regimes,
\[
a_{n,\lambda}^2
\asymp
m_H^{-2}s_\lambda^{-1}+1.
\]
Therefore \eqref{eq:bootstrap-rate-condition} implies
\begin{equation}
\big\|
\widetilde\Omega_{n,\lambda}^{\,0}
-\Omega_{n,\lambda}
\big\|_{\mathrm{HS}}
=o_p(a_{n,\lambda}).
\label{eq:uncentered-covariance-HS}
\end{equation}

Empirical centering gives
\[
\widehat\Omega_{n,\lambda}
=
\widetilde\Omega_{n,\lambda}^{\,0}
-\overline G_{n,\lambda}\otimes\overline G_{n,\lambda}.
\]
Under $H_0$,
\[
E\|\overline G_{n,\lambda}\|_\lambda^2
=
\frac{\mu_{n,\lambda}}{n}.
\]
In the sparse regime,
$\mu_{n,\lambda}=O(s_\lambda^{-1})$. In the growing cluster regime,
Proposition~\ref{prop:exact-dense-decomposition} gives
\[
\mu_{n,\lambda}
=
m_H^{-1}\operatorname{tr}(D_\lambda)
+(1-m_H^{-1})\operatorname{tr}(B_\lambda).
\]
Under the unitary representation of $B_\lambda$,
\[
\operatorname{tr}(B_\lambda)
=
\operatorname{tr}(\mathcal C_\psi R_\lambda^2),
\]
and the Hilbert--Schmidt Cauchy--Schwarz inequality yields
\[
|\operatorname{tr}(B_\lambda)|
\leq
\|\mathcal C_\psi\|_{\mathrm{HS}}
\|R_\lambda^2\|_{\mathrm{HS}}
=O(s_\lambda^{-1/2}).
\]
Thus, in both regimes,
\[
\mu_{n,\lambda}
\leq
C\{m_H^{-1}s_\lambda^{-1}+s_\lambda^{-1/2}\}.
\]
Since
\[
\frac{m_H^{-1}s_\lambda^{-1}+s_\lambda^{-1/2}}
{\{m_H^{-2}s_\lambda^{-1}+1\}^{1/2}}
\leq Cs_\lambda^{-1/2},
\]
we obtain
\[
\frac{\mu_{n,\lambda}/n}{a_{n,\lambda}}
\leq
\frac{C}{ns_\lambda^{1/2}}
\to0.
\]
Here the final convergence follows from
\eqref{eq:bootstrap-rate-condition}, which in particular implies
$ns_\lambda\to\infty$. Consequently,
\[
\|\overline G_{n,\lambda}\|_\lambda^2
=o_p(a_{n,\lambda}),
\]
and \eqref{eq:bootstrap-HS-consistency} follows from
\eqref{eq:uncentered-covariance-HS}.

For the trace,
\[
\widehat\mu_{n,\lambda}
=
\frac1n\sum_{i=1}^n
\|G_{i,\lambda}\|_\lambda^2
-\|\overline G_{n,\lambda}\|_\lambda^2.
\]
By \eqref{eq:g-fourth-bootstrap},
\[
\operatorname{Var}\left(
\frac1n\sum_{i=1}^n
\|G_{i,\lambda}\|_\lambda^2
\right)
\leq Cn^{-1}s_\lambda^{-2}.
\]
The square root of the right hand side divided by $a_{n,\lambda}$ is
$O(\chi_{n,\lambda}^{1/2})=o(1)$. Together with the preceding bound for
$\|\overline G_{n,\lambda}\|_\lambda^2$, this proves
\eqref{eq:bootstrap-mean-consistency}.

Finally,
\[
\left|
\|\widehat\Omega_{n,\lambda}\|_{\mathrm{HS}}^2
-\|\Omega_{n,\lambda}\|_{\mathrm{HS}}^2
\right|
\leq
\big(
\|\widehat\Omega_{n,\lambda}\|_{\mathrm{HS}}
+\|\Omega_{n,\lambda}\|_{\mathrm{HS}}
\big)
\|\widehat\Omega_{n,\lambda}-\Omega_{n,\lambda}\|_{\mathrm{HS}}.
\]
By \eqref{eq:bootstrap-HS-consistency},
$\|\widehat\Omega_{n,\lambda}\|_{\mathrm{HS}}=O_p(a_{n,\lambda})$.
Since
\[
\|\Omega_{n,\lambda}\|_{\mathrm{HS}}^2
=
\operatorname{tr}(\Omega_{n,\lambda}^2)
=
a_{n,\lambda}^2/2,
\]
division by $\operatorname{tr}(\Omega_{n,\lambda}^2)$ proves
\eqref{eq:bootstrap-trace2-consistency}.
\end{proof}

The next comparison treats the entire quadratic statistic at once. It uses independent Gaussian replacements with the same subject covariances, so it does not require a separate joint limit for leading eigendirections and a spectral tail.

\begin{lemma}\label{lem:quadratic-gaussian-replacement}
For each $n$, let $X_{n1},\ldots,X_{nn}$ be independent centered random elements of a separable real Hilbert space, with finite fourth moments. Write $V_{ni}=E(X_{ni}\otimes X_{ni})$ and $\Lambda_n=\sum_iV_{ni}$. Suppose $\|\Lambda_n\|_{\mathrm{HS}}=1$, and put $\kappa_n=\sum_iE\|X_{ni}\|^4$. Let $Y_{ni}$ be independent centered Gaussian elements with covariance $V_{ni}$. Define
\[
Q_n^X=2^{-1/2}\left\{\left\|\sum_iX_{ni}\right\|^2-\operatorname{tr}(\Lambda_n)\right\},
\]
and define $Q_n^Y$ by replacing $X_{ni}$ with $Y_{ni}$. For every bounded three times continuously differentiable $f$ with bounded derivatives,
\begin{equation}
|Ef(Q_n^X)-Ef(Q_n^Y)|
\leq C\|f'\|_\infty\kappa_n^{1/2}
+C\|f'''\|_\infty\{\kappa_n(1+\kappa_n)\}^{1/2}.
\label{eq:quadratic-replacement-bound}
\end{equation}
The constant is independent of $n$ and the Hilbert space.
\end{lemma}

\begin{proof}
Suppress $n$ and set $a_i=E\|X_i\|^4$. Gaussian fourth moments and Jensen's inequality give
\[
E\|Y_i\|^4=(\operatorname{tr}V_i)^2+2\operatorname{tr}(V_i^2)
\leq3(\operatorname{tr}V_i)^2\leq3a_i.
\]
Write $Q^X=U^X+D^X$, where $U^X=\sqrt2\sum_{i<j}\langle X_i,X_j\rangle$ and $D^X=2^{-1/2}\sum_i(\|X_i\|^2-E\|X_i\|^2)$. Independence gives $E|D^X|\leq(\kappa/2)^{1/2}$. The corresponding Gaussian bound is $E|D^Y|\leq(3\kappa/2)^{1/2}$. These terms contribute the first term in \eqref{eq:quadratic-replacement-bound}.

Replace the arguments of $U^X$ successively. At step $i$, let $Z_j$ denote the other arguments, each an independent $X_j$ or $Y_j$, and put $W_i=\sum_{j\ne i}Z_j$. The two quadratic forms being compared are $U_{-i}+\Delta_i^X$ and $U_{-i}+\Delta_i^Y$, where $U_{-i}$ depends only on the other arguments and
\[
\Delta_i^X=\sqrt2\langle X_i,W_i\rangle.
\]
Define $\Delta_i^Y$ in the same way. Conditional on the other arguments, these increments have the same mean zero and the same second moment. Taylor expansion through order two therefore leaves an error at most
\[
\frac{\|f'''\|_\infty}{6}\{E|\Delta_i^X|^3+E|\Delta_i^Y|^3\}.
\]
It remains to sum these third moments. Every hybrid sum has covariance $\Lambda-V_i\preceq\Lambda$, so
\[
E(\Delta_i^X)^2=E(\Delta_i^Y)^2
=2\operatorname{tr}\{V_i(\Lambda-V_i)\}.
\]
Summing gives a bound of $2\operatorname{tr}(\Lambda^2)=2$ for either family of increments.

For any fixed vector $x$, the scalar fourth-moment expansion for independent centered summands gives
\[
E\langle x,W_i\rangle^4
\leq3\langle\Lambda x,x\rangle^2
+\sum_{j\ne i}E\langle x,Z_j\rangle^4.
\]
Since $\|\Lambda\|_{\mathrm{op}}\leq1$ and $E\|Z_j\|^4\leq3a_j$, the right side is bounded by $3(1+\kappa)\|x\|^4$. Independence of $W_i$ and the replaced argument now implies
\[
\sum_i\{E(\Delta_i^X)^4+E(\Delta_i^Y)^4\}
\leq C\kappa(1+\kappa).
\]
Apply Cauchy--Schwarz first to each third moment and then to the sum over $i$. Together with the second-moment bound, this gives
\[
\sum_i\{E|\Delta_i^X|^3+E|\Delta_i^Y|^3\}
\leq C\{\kappa(1+\kappa)\}^{1/2}.
\]
Telescoping the replacements and adding the diagonal bounds proves the result.
\end{proof}

\begin{lemma}\label{lem:spectral-subsequence}
Let $\Lambda_n$ be positive trace class operators with
$0<\|\Lambda_n\|_{\mathrm{HS}}<\infty$, and let
$\ell_{n1}\geq\ell_{n2}\geq\cdots\geq0$ denote their eigenvalues. Define
\[
b_{nk}
=\frac{\ell_{nk}}{\|\Lambda_n\|_{\mathrm{HS}}}.
\]
Every subsequence has a further subsequence along which
$b_{nk}\to b_k$ for every fixed $k$, where
$\sum_{k\geq1}b_k^2\leq1$. If $Z_n\sim N(0,\Lambda_n)$, then along this
further subsequence,
\[
\frac{\|Z_n\|^2-\operatorname{tr}(\Lambda_n)}
{\{2\operatorname{tr}(\Lambda_n^2)\}^{1/2}}
\xrightarrow{d}\mathcal L_{\boldsymbol b},
\]
where
\begin{equation}
\mathcal L_{\boldsymbol b}
=
2^{-1/2}\sum_{k\geq1}b_k(Z_k^2-1)
+
\left(1-\sum_{k\geq1}b_k^2\right)^{1/2}Z_0,
\label{eq:mixed-chaos-limit}
\end{equation}
and $Z_0,Z_1,\ldots$ are independent $N(0,1)$ variables. The random
variable $\mathcal L_{\boldsymbol b}$ has a continuous distribution and
variance one.
\end{lemma}

\begin{proof}
Since
\(\sum_{k\geq1}b_{nk}^2=1\)
for every $n$, each $b_{nk}$ belongs to $[0,1]$. A diagonal
subsequence argument therefore gives a further subsequence on which
$b_{nk}\to b_k$ for every fixed $k$. Fatou's lemma yields
\(\sum_{k\geq1}b_k^2\leq1.\)

Using an eigenbasis of $\Lambda_n$, we may write
\[
\frac{\|Z_n\|^2-\operatorname{tr}(\Lambda_n)}
{\{2\operatorname{tr}(\Lambda_n^2)\}^{1/2}}
=
\frac1{\sqrt2}\sum_{k\geq1}b_{nk}(Z_k^2-1).
\]
For fixed $K$, the first $K$ summands converge in distribution to
\[
\frac1{\sqrt2}\sum_{k=1}^K b_k(Z_k^2-1).
\]

It remains to identify the contribution of the spectral tail. Since
$\sum_k b_k^2<\infty$, we have $b_k\to0$. Using the coordinatewise
convergence and the monotonicity
$b_{n1}\geq b_{n2}\geq\cdots$, choose integers $K_n\to\infty$ sufficiently
slowly that
\[
\max_{k>K_n}b_{nk}=b_{n,K_n+1}\to0
\]
and
\[
\sum_{k=1}^{K_n}b_{nk}^2
\to
\sum_{k\geq1}b_k^2.
\]
Consider
\[
R_n
=
\frac1{\sqrt2}\sum_{k>K_n}b_{nk}(Z_k^2-1).
\]
Its variance satisfies
\[
\operatorname{Var}(R_n)
=
\sum_{k>K_n}b_{nk}^2
\to
1-\sum_{k\geq1}b_k^2.
\]
If the limiting tail variance is positive, then, since $E(Z_k^2-1)^4<\infty$ and
$\max_{k>K_n}b_{nk}\to0$, the Lindeberg condition holds for the
triangular array defining $R_n$. Hence
\[
R_n
\xrightarrow{d}
\left(1-\sum_{k\geq1}b_k^2\right)^{1/2}Z_0.
\]
The cutoffs may also be chosen so slowly that $\sum_{k\leq K_n}(b_{nk}-b_k)^2\to0$. Coupling the leading terms with one Gaussian sequence then gives their $L^2$ convergence to $2^{-1/2}\sum_{k\geq1}b_k(Z_k^2-1)$. For each $n$, this leading sum and $R_n$ are independent, so their joint characteristic function factors. Their limits are therefore independent, proving \eqref{eq:mixed-chaos-limit}. If the limiting tail variance is zero, $R_n\to0$ in $L^2$ and no central limit theorem is needed.

Finally,
\[
\operatorname{Var}\left(
2^{-1/2}\sum_{k\geq1}b_k(Z_k^2-1)
\right)
=
\sum_{k\geq1}b_k^2,
\]
so independence gives $\operatorname{Var}(\mathcal L_{\boldsymbol b})=1$.
If $\sum_kb_k^2<1$, the limit contains a nondegenerate Gaussian
component and therefore has a continuous distribution. If
$\sum_kb_k^2=1$, at least one $b_k$ is positive, and the corresponding
weighted centered $\chi_1^2$ variable has a continuous distribution;
adding an independent random variable preserves continuity. This
completes the proof.
\end{proof}

The representation \eqref{eq:mixed-chaos-limit} gives the probabilistic form of the phase transition described in Remark~\ref{rem:dense-null-law}. If every normalized eigenvalue vanishes, then $\sum_{k\geq1}b_k^2=0$, the quadratic-chaos component disappears, and the limit is standard normal. If a nonvanishing compact spectral component remains, the corresponding weighted centered chi square terms persist. Along transition sequences, both the Gaussian and quadratic-chaos components may contribute to the limit.

We now compare the score statistic with its multiplier analogue. Related results for degenerate $U$ and $V$ statistics are given by \citet{arconesgine1992} and \citet{leuchtneumann2009,leuchtneumann2013}. Here the kernel changes with $\lambda$, so the argument must also control the subject diagonal term and the changing covariance spectrum.

\begin{theorem}\label{thm:uniform-bootstrap-validity}
Suppose Assumption~\ref{ass:loss-primitive} holds and either
Assumption~\ref{ass:sparse-cluster} or Assumption~\ref{ass:growing-cluster}
is satisfied. If $s_\lambda\to0$ and
\eqref{eq:bootstrap-rate-condition} holds, then under $H_0$,
\begin{equation}
\sup_{x\in\mathbb R}
\left|
P_0\left\{
\frac{S_{n,\lambda}-\mu_{n,\lambda}}
{a_{n,\lambda}}
\leq x
\right\}
-
P^*\left\{
\frac{S_{n,\lambda}^*-\widehat\mu_{n,\lambda}}
{\{2\operatorname{tr}(\widehat\Omega_{n,\lambda}^2)\}^{1/2}}
\leq x
\right\}
\right|
\xrightarrow{p}0.
\label{eq:uniform-bootstrap-kolmogorov}
\end{equation}
The conclusion holds throughout the sparse, transition, and dense regimes.
\end{theorem}

\begin{proof}
We first compare the original quadratic statistic with a Gaussian quadratic
form having the same covariance operator.

Let $Z_{n,\lambda}^{G}$ be a centered Gaussian element in
$\mathcal H_\lambda$ with covariance $\Omega_{n,\lambda}$, and define
\[
Q_{n,\lambda}^{G}
=
\frac{\|Z_{n,\lambda}^{G}\|_\lambda^2-\mu_{n,\lambda}}
{a_{n,\lambda}}.
\]
We show that
\begin{equation}
\sup_x
\left|
P_0\left\{
\frac{S_{n,\lambda}-\mu_{n,\lambda}}
{a_{n,\lambda}}
\leq x
\right\}
-
P(Q_{n,\lambda}^{G}\leq x)
\right|
\to0.
\label{eq:gaussianization-quadratic}
\end{equation}

Set $c_{n,\lambda}=\|\Omega_{n,\lambda}\|_{\mathrm{HS}}$ and apply Lemma~\ref{lem:quadratic-gaussian-replacement} to
\[
X_{ni}=\frac{G_{i,\lambda}}{\sqrt{n c_{n,\lambda}}}.
\]
Their aggregate covariance is $\Lambda_n=\Omega_{n,\lambda}/c_{n,\lambda}$, whose Hilbert--Schmidt norm is one. The fourth-moment bound yields
\[
\kappa_n=\frac{\sum_iE\|G_{i,\lambda}\|_\lambda^4}{n^2c_{n,\lambda}^2}
\leq C\chi_{n,\lambda}\to0.
\]
The quadratic forms $Q_n^X$ and $Q_n^Y$ in that lemma are exactly $(S_{n,\lambda}-\mu_{n,\lambda})/a_{n,\lambda}$ and $Q_{n,\lambda}^G$, respectively. Their expectations against every bounded smooth test function therefore differ by $o(1)$.

To pass to distribution functions, take any subsequence. Lemma~\ref{lem:spectral-subsequence} gives a further subsequence along which $Q_{n,\lambda}^G$ converges to a continuous law $\mathcal L_{\boldsymbol b}$. Smooth approximations to indicators, together with \eqref{eq:quadratic-replacement-bound}, give the same weak limit for the original statistic. Convergence of distribution functions to a continuous limit is uniform on $\mathbb R$. The subsequence principle therefore proves \eqref{eq:gaussianization-quadratic}. This argument covers the Gaussian, compact quadratic, and mixed limits without a separate asymptotic independence assertion.

We next compare the Gaussian quadratic form with the conditional bootstrap law. Write $\widehat\ell_{nk}$ for the ordered eigenvalues of $\widehat\Omega_{n,\lambda}$ and set $\widehat b_{nk}=\widehat\ell_{nk}/\|\widehat\Omega_{n,\lambda}\|_{\mathrm{HS}}$. Lemma~\ref{lem:bootstrap-covariance-consistency} and the reverse triangle inequality give
\[
\frac{\|\widehat\Omega_{n,\lambda}-\Omega_{n,\lambda}\|_{\mathrm{HS}}}{\|\Omega_{n,\lambda}\|_{\mathrm{HS}}}\to_p0,
\qquad
\frac{\|\widehat\Omega_{n,\lambda}\|_{\mathrm{HS}}}{\|\Omega_{n,\lambda}\|_{\mathrm{HS}}}\to_p1.
\]
Apply the Hoffman--Wielandt inequality to the normalized positive operators; see \citet{hoffmanwielandt1953,cochranhinds1972}. It yields
\[
\begin{aligned}
\left\{\sum_{k\geq1}(\widehat b_{nk}-b_{nk})^2\right\}^{1/2}
&\leq\left\|\frac{\widehat\Omega_{n,\lambda}}{\|\widehat\Omega_{n,\lambda}\|_{\mathrm{HS}}}
-\frac{\Omega_{n,\lambda}}{\|\Omega_{n,\lambda}\|_{\mathrm{HS}}}\right\|_{\mathrm{HS}}\\
&\leq\frac{\|\widehat\Omega_{n,\lambda}-\Omega_{n,\lambda}\|_{\mathrm{HS}}}{\|\widehat\Omega_{n,\lambda}\|_{\mathrm{HS}}}
+\left|\frac{\|\Omega_{n,\lambda}\|_{\mathrm{HS}}}{\|\widehat\Omega_{n,\lambda}\|_{\mathrm{HS}}}-1\right|\to_p0.
\end{aligned}
\]
Conditionally on the data, the standardized bootstrap statistic has the
same distribution as
\[
\widehat Q_{n,\lambda}^{G}
=
2^{-1/2}
\sum_{k\geq1}
\widehat b_{nk}(Z_k^2-1),
\]
where $Z_1,Z_2,\ldots$ are independent standard normal variables.
Similarly,
\[
Q_{n,\lambda}^{G}
\stackrel{d}{=}
2^{-1/2}
\sum_{k\geq1}
b_{nk}(Z_k^2-1).
\]
Coupling the two quadratic forms with the same Gaussian sequence gives
\[
E_Z
\left|
\widehat Q_{n,\lambda}^{G}
-
Q_{n,\lambda}^{G}
\right|^2
=
\sum_{k\geq1}
(\widehat b_{nk}-b_{nk})^2
\xrightarrow{p}0.
\]
Thus their distributions become asymptotically identical along every
subsequence on which the normalized spectrum converges. Since every
such subsequential limit is one of the continuous laws
$\mathcal L_{\boldsymbol b}$ from
Lemma~\ref{lem:spectral-subsequence}, the corresponding conditional and
unconditional distribution functions converge uniformly.

Combining this result with
\eqref{eq:gaussianization-quadratic} proves
\eqref{eq:uniform-bootstrap-kolmogorov}.
\end{proof}

\begin{remark}
\label{rem:bootstrap-three-regimes}
The spectral description makes the role of the bootstrap transparent. If
$m_H^2s_\lambda\to0$, the normalized covariance spectrum becomes diffuse
and $b_k=0$ for every fixed $k$, so \eqref{eq:mixed-chaos-limit} reduces
to $N(0,1)$. If $m_H^2s_\lambda\to\infty$, the compact trajectory
spectrum may persist and the limit can contain the weighted centered
chi square component generated by $\mathcal C_\psi$. At
$m_H^2s_\lambda\asymp1$, part of the spectral mass may remain compact
while the remainder escapes to high frequencies, producing the mixed
limit in \eqref{eq:mixed-chaos-limit}. The multiplier procedure estimates
the covariance spectrum rather than selecting one of these limiting laws
in advance.
\end{remark}

\begin{corollary}
\label{cor:bootstrap-optimal-scale}
Let $r\geq2$ and choose $s_\lambda\asymp s_n^*$ according to
\eqref{eq:unified-optimal-s}. Under the assumptions of
Theorem~\ref{thm:uniform-bootstrap-validity},
\eqref{eq:bootstrap-rate-condition} holds throughout the sparse,
transition, and dense regimes. Let $c_{n,\gamma}^*$ denote the
conditional $(1-\gamma)$ quantile of $S_{n,\lambda}^*$. Then
\begin{equation}
P_0\{S_{n,\lambda}>c_{n,\gamma}^*\}\to\gamma.
\label{eq:bootstrap-asymptotic-size}
\end{equation}
This result concerns null calibration. Power with an estimated critical value requires bounds under alternatives; Section~\ref{subsec:feasible-power} supplies a separate argument for the feasible spline procedure under its additional assumptions.
\end{corollary}

\begin{proof}
Consider first the measurement dominated regime
$m_H\ll n^{1/(4r)}$. Then
\[
s_\lambda\asymp(nm_H)^{-2/(4r+1)}
\]
and $m_H^2s_\lambda\to0$. Hence
\[
\chi_{n,\lambda}
\leq
C\frac{m_H^2}{ns_\lambda}.
\]
The right hand side increases with $m_H$, and at the transition scale
$m_H\asymp n^{1/(4r)}$ it has order
$n^{1/r-1}$, which tends to zero for $r\geq2$.

In the trajectory dominated regime
$m_H\gg n^{1/(4r)}$,
\[
s_\lambda\asymp n^{-1/(2r)}.
\]
Since the denominator in \eqref{eq:bootstrap-rate-condition} is bounded
away from zero,
\[
\chi_{n,\lambda}
\leq
Cn^{-1}s_\lambda^{-2}
\asymp
n^{-1+1/r}
\to0.
\]
At the transition
$m_H\asymp n^{1/(4r)}$, the two smoothing rates coincide and the same
bound $n^{-1+1/r}\to0$ applies. Thus
Theorem~\ref{thm:uniform-bootstrap-validity} is valid in all three
regimes.

Let
\[
\widehat a_{n,\lambda}^2
=
2\operatorname{tr}(\widehat\Omega_{n,\lambda}^2).
\]
Lemma~\ref{lem:bootstrap-covariance-consistency} gives
\[
\frac{\widehat\mu_{n,\lambda}-\mu_{n,\lambda}}
{a_{n,\lambda}}
\to_p0
\]
and
\[
\frac{\widehat a_{n,\lambda}}
{a_{n,\lambda}}
\to_p1.
\]
Together with \eqref{eq:uniform-bootstrap-kolmogorov} and continuity of
every possible subsequential limit in
Lemma~\ref{lem:spectral-subsequence}, these relations imply
\eqref{eq:bootstrap-asymptotic-size}.

\end{proof}

These size statements use the exact conditional bootstrap quantile. Monte Carlo implementation has the same limit when the number of independent multiplier replicates tends to infinity; a fixed number retains Monte Carlo discretization error.

\section{Finite basis approximation and feasible B-spline implementation}
\label{sec:finite-basis-theory}

The finite spline statistic involves errors from truncating the function space and estimating the curvature matrix. We first bound these errors under the null, then study power using the curvature metric induced by shifted residuals.

\subsection{Galerkin representation and spline projection error}

Let $\mathcal S_q$ be a B-spline space of fixed polynomial degree $p\geq2r-1$ defined on a quasi-uniform partition of $[0,1]$ with mesh width $h_q\asymp q^{-1}$, and suppose that $\mathcal P_{r-1}\subset\mathcal S_q$. Let $\Pi_{q,\lambda}$ denote the orthogonal projection from $\mathcal H_\lambda$ onto $\mathcal S_q$ with respect to $\langle\cdot,\cdot\rangle_\lambda$. The Ritz--Galerkin approximation bounds used below follow from standard spline approximation results; see \citet{schumaker2007} and \citet{sandeetal2020,sandeetal2022}.

Define
\begin{equation*}
\vartheta_{q,\lambda}=q^{1-2r}s_\lambda^{-2r}.
\end{equation*}

\begin{lemma}\label{lem:spline-ritz-tail}
Suppose the eigensystem conditions of Proposition~\ref{prop:eigensystem} hold and $qs_\lambda\geq1$. Then
\begin{equation}
\operatorname{tr}\{(I-\Pi_{q,\lambda})A_\lambda\}\leq C\vartheta_{q,\lambda}.
\label{eq:A-projection-trace}
\end{equation}
Moreover, for every $f\in\mathcal H_\lambda$,
\begin{equation}
\|(I-\Pi_{q,\lambda})A_\lambda f\|_\lambda\leq C(qs_\lambda)^{-r}\|f\|_\lambda.
\label{eq:A-projection-operator}
\end{equation}
\end{lemma}

\begin{proof}
Let $e_\nu=(1+\lambda\gamma_\nu)^{-1/2}\phi_\nu$. Then $\{e_\nu\}$ is orthonormal in $\mathcal H_\lambda$ and $A_\lambda e_\nu=a_{\nu,\lambda}e_\nu$, where $a_{\nu,\lambda}=(1+\lambda\gamma_\nu)^{-1}$. The null space eigenfunctions belong to $\mathcal P_{r-1}\subset\mathcal S_q$ and are therefore reproduced exactly.

For $\gamma_\nu>0$, the differential equation in Proposition~\ref{prop:eigensystem}, together with the regularity and boundedness of $I\pi$, gives $\|\phi_\nu\|_{H^{2r}}\leq C(1+\gamma_\nu)$. Corollary~2 of \citet{sandeetal2020}, with derivative order $r$ and smoothness $2r$, yields $v_{\nu,q}\in\mathcal S_q$ satisfying
\[
\|\phi_\nu-v_{\nu,q}\|_{L^2}\leq Cq^{-2r}(1+\gamma_\nu),\qquad \|\phi_\nu^{(r)}-v_{\nu,q}^{(r)}\|_{L^2}\leq Cq^{-r}(1+\gamma_\nu).
\]
Since $\Pi_{q,\lambda}$ is the best approximation in the $\lambda$ norm, $\lambda=s_\lambda^{2r}$, $\gamma_\nu\asymp\nu^{2r}$, and $qs_\lambda\geq1$,
\begin{equation}
\|(I-\Pi_{q,\lambda})e_\nu\|_\lambda^2\leq C\min\{1,(\nu/q)^{2r}\}.
\label{eq:eigen-spline-projection}
\end{equation}
For $\nu\leq q$, this follows by bounding the $L^2$ and penalty contributions separately; for $\nu>q$, it follows from the contraction property of the orthogonal projection.

Using orthogonality of $\Pi_{q,\lambda}$,
\[
\operatorname{tr}\{(I-\Pi_{q,\lambda})A_\lambda\}=\sum_{\nu\geq1}a_{\nu,\lambda}\|(I-\Pi_{q,\lambda})e_\nu\|_\lambda^2.
\]
Hence, by \eqref{eq:eigen-spline-projection},
\[
\operatorname{tr}\{(I-\Pi_{q,\lambda})A_\lambda\}\leq C\sum_{\nu\geq1}\frac{\min\{1,(\nu/q)^{2r}\}}{1+(s_\lambda\nu)^{2r}}.
\]
Since $s_\lambda^{-1}\leq q$, splitting the sum over $\nu\leq s_\lambda^{-1}$, $s_\lambda^{-1}<\nu\leq q$, and $\nu>q$ gives
\[
\sum_{\nu\geq1}\frac{\min\{1,(\nu/q)^{2r}\}}{1+(s_\lambda\nu)^{2r}}\leq Cq^{1-2r}s_\lambda^{-2r}=C\vartheta_{q,\lambda},
\]
which proves \eqref{eq:A-projection-trace}.

For \eqref{eq:A-projection-operator}, put $g=A_\lambda f$. Its spectral representation gives $\|g\|_{H^{2r}}\leq C\lambda^{-1}\|f\|_\lambda$. The spline approximation bound and the best approximation property yield
\[
\|(I-\Pi_{q,\lambda})g\|_\lambda^2\leq C\{q^{-4r}+\lambda q^{-2r}\}\|g\|_{H^{2r}}^2.
\]
Therefore
\[
\|(I-\Pi_{q,\lambda})g\|_\lambda^2\leq C\{(qs_\lambda)^{-4r}+(qs_\lambda)^{-2r}\}\|f\|_\lambda^2\leq C(qs_\lambda)^{-2r}\|f\|_\lambda^2,
\]
and taking square roots proves \eqref{eq:A-projection-operator}.
\end{proof}

\begin{proposition}\label{prop:finite-basis-projection}
Under $H_0$ and the conditions of Lemma~\ref{lem:spline-ritz-tail}, let $H_{q,\lambda}$ and $s_i$ be defined by \eqref{eq:population-basis-hessian} and \eqref{eq:subject-score-vector}. The restricted Riesz representative of $G_{i,\lambda}$ on $\mathcal S_q$ is $\Pi_{q,\lambda}G_{i,\lambda}$, and
\begin{equation}
S_{n,\lambda}^{(q)}=\left\|\Pi_{q,\lambda}n^{-1/2}\sum_{i=1}^nG_{i,\lambda}\right\|_\lambda^2=n\bar s^{\prime}H_{q,\lambda}^{-1}\bar s.
\label{eq:finite-projection-identity}
\end{equation}
Furthermore,
\begin{equation}
0\leq S_{n,\lambda}-S_{n,\lambda}^{(q)}=\left\|(I-\Pi_{q,\lambda})n^{-1/2}\sum_{i=1}^nG_{i,\lambda}\right\|_\lambda^2.
\label{eq:finite-statistic-exact-error}
\end{equation}
Under Assumption~\ref{ass:sparse-cluster},
\begin{equation}
E\{S_{n,\lambda}-S_{n,\lambda}^{(q)}\}\leq Cm_H^{-1}\vartheta_{q,\lambda},
\label{eq:sparse-finite-statistic-error}
\end{equation}
whereas under Assumption~\ref{ass:growing-cluster},
\begin{equation}
E\{S_{n,\lambda}-S_{n,\lambda}^{(q)}\}\leq C\vartheta_{q,\lambda}.
\label{eq:dense-finite-statistic-error}
\end{equation}
\end{proposition}

\begin{proof}
For $v(t)=B(t)^{\prime}\theta\in\mathcal S_q$,
\[
\langle G_{i,\lambda},v\rangle_\lambda=s_i^{\prime}\theta.
\]
The element $B(\cdot)^{\prime}H_{q,\lambda}^{-1}s_i$ has the same inner product with every $v\in\mathcal S_q$, so uniqueness of the Riesz representative gives
\[
\Pi_{q,\lambda}G_{i,\lambda}=B(\cdot)^{\prime}H_{q,\lambda}^{-1}s_i.
\]
This proves \eqref{eq:finite-projection-identity}, and orthogonality of $\Pi_{q,\lambda}$ gives \eqref{eq:finite-statistic-exact-error}.

Under Assumption~\ref{ass:sparse-cluster}, Proposition~\ref{prop:covariance-comparison} gives $\Omega_{n,\lambda}\preceq Cm_H^{-1}A_\lambda$. Hence
\[
E\{S_{n,\lambda}-S_{n,\lambda}^{(q)}\}=\operatorname{tr}\{(I-\Pi_{q,\lambda})\Omega_{n,\lambda}\}\leq Cm_H^{-1}\operatorname{tr}\{(I-\Pi_{q,\lambda})A_\lambda\},
\]
and \eqref{eq:sparse-finite-statistic-error} follows from Lemma~\ref{lem:spline-ritz-tail}.

Under Assumption~\ref{ass:growing-cluster}, $D_\lambda\preceq CA_\lambda$. In the $L^2(V)$ representation of Proposition~\ref{prop:exact-dense-decomposition}, $B_\lambda$ is unitarily equivalent to $R_\lambda\mathcal C_\psi R_\lambda$, whereas $A_\lambda$ is equivalent to $R_\lambda^2$. Since $\mathcal C_\psi$ is positive and bounded,
\[
R_\lambda\mathcal C_\psi R_\lambda\preceq\|\mathcal C_\psi\|_{\mathrm{op}}R_\lambda^2,
\]
and therefore $B_\lambda\preceq CA_\lambda$. Equation~\eqref{eq:exact-dense-omega} then gives $\Omega_{n,\lambda}\preceq CA_\lambda$. Applying Lemma~\ref{lem:spline-ritz-tail} proves \eqref{eq:dense-finite-statistic-error}.
\end{proof}

Recall that
\(a_{n,\lambda}=\{2\operatorname{tr}(\Omega_{n,\lambda}^2)\}^{1/2}.\)
The preceding proposition yields the dimension condition required for the finite basis statistic to be asymptotically equivalent to its Hilbert space counterpart.

\begin{corollary}\label{cor:spline-truncation-scale}
Suppose the conditions of Proposition~\ref{prop:finite-basis-projection} hold. Under Assumption~\ref{ass:sparse-cluster}, suppose
\begin{equation}
\frac{m_H^{-1}\vartheta_{q,\lambda}}{a_{n,\lambda}}\to0.
\label{eq:sparse-q-condition}
\end{equation}
Under Assumption~\ref{ass:growing-cluster}, suppose instead
\begin{equation}
\frac{\vartheta_{q,\lambda}}{a_{n,\lambda}}\to0.
\label{eq:growing-q-condition}
\end{equation}
Then, under $H_0$ and the corresponding regime,
\begin{equation}
\frac{S_{n,\lambda}-S_{n,\lambda}^{(q)}}{a_{n,\lambda}}\xrightarrow{p}0.
\label{eq:finite-statistic-equivalence}
\end{equation}
In the bounded sparse regime, $a_{n,\lambda}\asymp m_H^{-1}s_\lambda^{-1/2}$, so \eqref{eq:sparse-q-condition} is equivalent to
\[
q^{2r-1}s_\lambda^{2r-1/2}\to\infty.
\]
In the transition and dense regimes, $a_{n,\lambda}$ is bounded away from zero. Hence the stronger explicit condition
\[
q^{2r-1}s_\lambda^{2r}\to\infty
\]
is sufficient for \eqref{eq:growing-q-condition}.
\end{corollary}

\begin{proof}
Equations~\eqref{eq:sparse-finite-statistic-error} and \eqref{eq:dense-finite-statistic-error}, together with Markov's inequality, give \eqref{eq:finite-statistic-equivalence}. Put
\[
t_{q,\lambda}=\operatorname{tr}\{(I-\Pi_{q,\lambda})\Omega_{n,\lambda}\}.
\]
Since $\Omega_{n,\lambda}$ is positive,
\[
\|\Omega_{n,\lambda}-\Pi_{q,\lambda}\Omega_{n,\lambda}\Pi_{q,\lambda}\|_{\mathrm{HS}}\leq2\{\|\Omega_{n,\lambda}\|_{\mathrm{op}}t_{q,\lambda}\}^{1/2}\leq C\{a_{n,\lambda}t_{q,\lambda}\}^{1/2},
\]
where $\|\Omega_{n,\lambda}\|_{\mathrm{op}}\leq\|\Omega_{n,\lambda}\|_{\mathrm{HS}}=a_{n,\lambda}/\sqrt2$.

Define $\mu_{n,\lambda}^{(q)}=\operatorname{tr}(\Pi_{q,\lambda}\Omega_{n,\lambda}\Pi_{q,\lambda})$ and
\[
\{a_{n,\lambda}^{(q)}\}^2=2\operatorname{tr}\{(\Pi_{q,\lambda}\Omega_{n,\lambda}\Pi_{q,\lambda})^2\}.
\]
Then $\mu_{n,\lambda}-\mu_{n,\lambda}^{(q)}=t_{q,\lambda}$. Hence $t_{q,\lambda}/a_{n,\lambda}\to0$ implies
\[
\frac{\mu_{n,\lambda}-\mu_{n,\lambda}^{(q)}}{a_{n,\lambda}}\to0
\]
and
\[
\frac{a_{n,\lambda}^{(q)}}{a_{n,\lambda}}\to1.
\]
Thus the finite basis centering and quadratic scale agree with their full RKHS counterparts on the testing scale. The explicit dimension conditions follow by substituting $\vartheta_{q,\lambda}=q^{1-2r}s_\lambda^{-2r}$ and the variance orders from Proposition~\ref{prop:covariance-comparison} and Theorem~\ref{thm:dense-variance-order}.
\end{proof}

\subsection{Curvature estimation and feasible calibration}
For $\delta>0$, let $w_\delta$ be defined by \eqref{eq:score-increment-curvature-weight}. In addition to Assumption~\ref{ass:loss-primitive}, suppose that for all sufficiently small $\delta$ there exists a deterministic $L_\delta$ such that
\begin{equation}
0\leq w_\delta(u)\leq L_\delta,\qquad u\in\mathbb R.
\label{eq:curvature-envelope}
\end{equation}
This is a primitive envelope condition on the score increment. For a globally Lipschitz score, one may take $L_\delta=L_\psi$, whereas for a bounded monotone score one may take $L_\delta=C/\delta$. Thus the condition covers both Huber and absolute loss.

Define
\[
\epsilon_{n,q,\lambda,\delta}=\delta^\alpha+\left\{\frac{L_\delta\log(2q)}{ns_\lambda}\right\}^{1/2}+\frac{L_\delta\log(2q)}{ns_\lambda}.
\]
\begin{theorem}\label{thm:curvature-matrix-rate}
Suppose Assumption~\ref{ass:loss-primitive} and \eqref{eq:curvature-envelope} hold under either cluster design. Then, under $H_0$,
\begin{equation}
\left\|H_{q,\lambda}^{-1/2}(\widehat H_{q,\lambda}-H_{q,\lambda})H_{q,\lambda}^{-1/2}\right\|_{\mathrm{op}}=O_p(\epsilon_{n,q,\lambda,\delta}).
\label{eq:relative-curvature-rate}
\end{equation}
\end{theorem}

\begin{proof}
Set $x_q(t)=H_{q,\lambda}^{-1/2}B(t)$. Since $H_{q,\lambda}$ represents the $\lambda$ inner product on $\mathcal S_q$,
\[
\|x_q(t)\|^2=B(t)^{\prime}H_{q,\lambda}^{-1}B(t)=\sup_{\substack{f\in\mathcal S_q\\ \|f\|_\lambda\leq1}}|f(t)|^2\leq K_\lambda(t,t)\leq Cs_\lambda^{-1}.
\]
Define
\[
Y_i=\frac1{m_i}\sum_{j=1}^{m_i}w_\delta(e_{ij}^0)x_q(T_{ij})x_q(T_{ij})^{\prime}.
\]
By \eqref{eq:curvature-envelope}, $Y_i$ is positive semidefinite and $\|Y_i\|_{\mathrm{op}}\leq CL_\delta s_\lambda^{-1}$.

Let
\[
I_\delta(t)=E\{w_\delta(e_{ij})\mid T_{ij}=t\}.
\]
By the definition of $w_\delta$ and \eqref{eq:score-linearization},
\[
I_\delta(t)=\frac{\Psi_t(-\delta)-\Psi_t(\delta)}{2\delta}=I(t)+O(\delta^\alpha)
\]
uniformly in $t$. Since the penalty part is identical in $\widehat H_{q,\lambda}$ and $H_{q,\lambda}$, it cancels from their difference. Hence
\begin{equation}
\left\|E(Y_i)-H_{q,\lambda}^{-1/2}A_qH_{q,\lambda}^{-1/2}\right\|_{\mathrm{op}}\leq C\delta^\alpha.
\label{eq:curvature-bias-bound}
\end{equation}
Moreover, $H_{q,\lambda}\succeq A_q$, so $\|H_{q,\lambda}^{-1/2}A_qH_{q,\lambda}^{-1/2}\|_{\mathrm{op}}\leq1$. Thus $\|E(Y_i)\|_{\mathrm{op}}\leq C$ for sufficiently small $\delta$.

Put $R=CL_\delta s_\lambda^{-1}$. Since $0\preceq Y_i\preceq RI_q$,
\[
E(Y_i^2)\preceq R E(Y_i)\preceq CRI_q.
\]
Therefore the matrix variance proxy for the independent centered matrices $Y_i-E(Y_i)$ is bounded by $CnL_\delta s_\lambda^{-1}$. The self adjoint matrix Bernstein inequality of \citet{tropp2012} yields
\[
\left\|\frac1n\sum_{i=1}^n\{Y_i-E(Y_i)\}\right\|_{\mathrm{op}}=O_p\left[\left\{\frac{L_\delta\log(2q)}{ns_\lambda}\right\}^{1/2}+\frac{L_\delta\log(2q)}{ns_\lambda}\right].
\]
Combining this bound with \eqref{eq:curvature-bias-bound} and the definition of $\epsilon_{n,q,\lambda,\delta}$ proves \eqref{eq:relative-curvature-rate}.
\end{proof}

We first introduce the projected covariance operator and its associated centering and quadratic scale. Let
\[
\Omega_{n,\lambda}^{(q)}
=\Pi_{q,\lambda}\Omega_{n,\lambda}\Pi_{q,\lambda},
\]
\[
\mu_{n,\lambda}^{(q)}
=\operatorname{tr}(\Omega_{n,\lambda}^{(q)}),
\]
and
\[
a_{n,\lambda}^{(q)}
=\{2\operatorname{tr}[(\Omega_{n,\lambda}^{(q)})^2]\}^{1/2}.
\]
The next result shows that estimating the curvature matrix does not affect the statistic or its bootstrap calibration on the testing scale.

\begin{theorem}\label{thm:feasible-finite-basis}
Suppose the assumptions of Theorems~\ref{thm:uniform-bootstrap-validity} and~\ref{thm:curvature-matrix-rate} hold and the appropriate truncation condition in Corollary~\ref{cor:spline-truncation-scale} is satisfied. If
\begin{equation}
\epsilon_{n,q,\lambda,\delta}\frac{\mu_{n,\lambda}^{(q)}+a_{n,\lambda}^{(q)}}{a_{n,\lambda}^{(q)}}\to0,
\label{eq:metric-estimation-testing-condition}
\end{equation}
then under $H_0$,
\begin{equation}
\frac{\widehat S_{n,\lambda}^{(q)}-S_{n,\lambda}^{(q)}}{a_{n,\lambda}^{(q)}}\xrightarrow{p}0,
\label{eq:feasible-statistic-metric-equivalence}
\end{equation}
\begin{equation}
\frac{\widehat\mu_{n,\lambda}^{(q)}-\mu_{n,\lambda}^{(q)}}{a_{n,\lambda}^{(q)}}\xrightarrow{p}0,
\label{eq:feasible-mu-consistency}
\end{equation}
and
\begin{equation}
\frac{\widehat\sigma_{n,\lambda}^{(q)}}{a_{n,\lambda}^{(q)}}\xrightarrow{p}1.
\label{eq:feasible-sigma-consistency}
\end{equation}
Moreover, the feasible finite basis multiplier bootstrap is conditionally consistent for the full RKHS score statistic.
\end{theorem}

\begin{proof}
Put
\[
E_{q,\lambda}=H_{q,\lambda}^{-1/2}(\widehat H_{q,\lambda}-H_{q,\lambda})H_{q,\lambda}^{-1/2}.
\]
By Theorem~\ref{thm:curvature-matrix-rate}, $\|E_{q,\lambda}\|_{\mathrm{op}}=O_p(\epsilon_{n,q,\lambda,\delta})=o_p(1)$, where the final relation follows from \eqref{eq:metric-estimation-testing-condition}. Hence, with probability tending to one,
\[
\widehat H_{q,\lambda}^{-1}=H_{q,\lambda}^{-1/2}(I_q+E_{q,\lambda})^{-1}H_{q,\lambda}^{-1/2}
\]
and
\[
\|(I_q+E_{q,\lambda})^{-1}-I_q\|_{\mathrm{op}}\leq C\|E_{q,\lambda}\|_{\mathrm{op}}.
\]
If $Y_n=\sqrt n\,H_{q,\lambda}^{-1/2}\bar s$, then
\[
S_{n,\lambda}^{(q)}=\|Y_n\|^2,\qquad \widehat S_{n,\lambda}^{(q)}=Y_n^{\prime}(I_q+E_{q,\lambda})^{-1}Y_n.
\]
Therefore
\begin{equation}
|\widehat S_{n,\lambda}^{(q)}-S_{n,\lambda}^{(q)}|\leq C\|E_{q,\lambda}\|_{\mathrm{op}}S_{n,\lambda}^{(q)}.
\label{eq:metric-statistic-pathwise}
\end{equation}
For the projected scores, $E(S_{n,\lambda}^{(q)})=\mu_{n,\lambda}^{(q)}$ and $S_{n,\lambda}^{(q)}-\mu_{n,\lambda}^{(q)}=O_p(a_{n,\lambda}^{(q)})$. Thus \eqref{eq:metric-estimation-testing-condition} and \eqref{eq:metric-statistic-pathwise} give \eqref{eq:feasible-statistic-metric-equivalence}.

Let
\[
\widehat\Omega_{n,\lambda}^{(q),0}=H_{q,\lambda}^{-1/2}\widehat C_qH_{q,\lambda}^{-1/2}.
\]
Since $\|\Pi_{q,\lambda}G_{i,\lambda}\|_\lambda\leq\|G_{i,\lambda}\|_\lambda$, the fourth moment bound and the argument of Lemma~\ref{lem:bootstrap-covariance-consistency} apply to the projected scores. The truncation condition gives $a_{n,\lambda}^{(q)}/a_{n,\lambda}\to1$, so
\[
\|\widehat\Omega_{n,\lambda}^{(q),0}-\Omega_{n,\lambda}^{(q)}\|_{\mathrm{HS}}=o_p(a_{n,\lambda}^{(q)})
\]
and the corresponding trace error is $o_p(a_{n,\lambda}^{(q)})$.

Replacing the identity metric by $(I_q+E_{q,\lambda})^{-1}$ changes the trace by at most
\[
C\|E_{q,\lambda}\|_{\mathrm{op}}\operatorname{tr}(\widehat\Omega_{n,\lambda}^{(q),0})=O_p\{\epsilon_{n,q,\lambda,\delta}(\mu_{n,\lambda}^{(q)}+a_{n,\lambda}^{(q)})\},
\]
which is $o_p(a_{n,\lambda}^{(q)})$ by \eqref{eq:metric-estimation-testing-condition}. The same perturbation argument for the squared Hilbert--Schmidt norm yields \eqref{eq:feasible-mu-consistency} and \eqref{eq:feasible-sigma-consistency}.

Conditionally on the data, the population-metric and feasible finite basis bootstrap statistics are Gaussian quadratic forms generated by the same multiplier vector, with weighting matrices differing through $(I_q+E_{q,\lambda})^{-1}$. The conditional analogue of \eqref{eq:metric-statistic-pathwise}, together with \eqref{eq:metric-estimation-testing-condition}, shows that their difference is $o_p(a_{n,\lambda}^{(q)})$. The continuous subsequential limits in Lemma~\ref{lem:spectral-subsequence} therefore imply convergence of their conditional distribution functions in Kolmogorov distance. Finally, Corollary~\ref{cor:spline-truncation-scale} and Theorem~\ref{thm:uniform-bootstrap-validity} transfer this conclusion from the finite basis statistic to the full RKHS score statistic.
\end{proof}

\begin{corollary}\label{cor:fully-feasible-basis}
Under the loss and sampling assumptions of Theorem~\ref{thm:feasible-finite-basis}, let $r\geq2$, choose $s_\lambda\asymp s_n^*$ as in \eqref{eq:unified-optimal-s}, and let $L_n\to\infty$ satisfy $\log L_n=O(\log n)$. Choose a quasi-uniform B-spline space of fixed degree $p\geq2r-1$ with
\begin{equation}
q_n\asymp L_n s_\lambda^{-2r/(2r-1)}.
\label{eq:explicit-optimal-q}
\end{equation}
For a Lipschitz score, choose $\delta$ such that
\[
\delta^\alpha\leq\{\log(2q_n)/(ns_\lambda)\}^{1/2}.
\]
For a bounded jump score, choose
\begin{equation}
\delta\asymp\{\log(2q_n)/(ns_\lambda)\}^{1/(2\alpha+1)}
\label{eq:explicit-jump-delta}
\end{equation}
and suppose
\begin{equation}
\alpha>\frac{1}{4(r-1)}.
\label{eq:jump-holder-threshold}
\end{equation}
Then the spline truncation and metric estimation conditions hold throughout the sparse, transition, and dense regimes. Consequently, the test based on $\{s_i,\widehat H_{q_n,\lambda},\widehat C_{q_n}\}$ and Gaussian subject multipliers has asymptotic size $\gamma$ in all three regimes.
\end{corollary}

\begin{proof}
By \eqref{eq:explicit-optimal-q},
\[
\vartheta_{q_n,\lambda}=q_n^{1-2r}s_\lambda^{-2r}\asymp L_n^{1-2r}\to0.
\]
Hence \eqref{eq:growing-q-condition} holds whenever the trajectory component is first order because $a_{n,\lambda}$ is bounded away from zero. In the bounded sparse regime, $a_{n,\lambda}\asymp m_H^{-1}s_\lambda^{-1/2}\to\infty$, so \eqref{eq:sparse-q-condition} also holds. Moreover,
\[
q_ns_\lambda\asymp L_n s_\lambda^{-1/(2r-1)}\to\infty,
\]
so the spline space resolves the effective testing frequencies.

The trace bounds from the sparse and growing cluster theories imply, uniformly across the three regimes,
\begin{equation}
\frac{\mu_{n,\lambda}^{(q_n)}+a_{n,\lambda}^{(q_n)}}{a_{n,\lambda}^{(q_n)}}\leq Cs_\lambda^{-1/2}.
\label{eq:mu-over-a-universal}
\end{equation}
Indeed, when $m_H^2s_\lambda$ is bounded above, the diagonal component controls both the centering and the quadratic scale, giving a ratio of order $s_\lambda^{-1/2}$. When $m_H^2s_\lambda$ is bounded away from zero, $a_{n,\lambda}$ is bounded away from zero and
$m_H^{-1}s_\lambda^{-1}\leq Cs_\lambda^{-1/2}$; the off diagonal trace is also $O(s_\lambda^{-1/2})$. Corollary~\ref{cor:spline-truncation-scale} transfers these orders to the projected quantities.

For a Lipschitz score, $L_\delta=O(1)$, and the stated choice of $\delta$ gives
\[
\epsilon_{n,q_n,\lambda,\delta}
=O\left[\{\log(2q_n)/(ns_\lambda)\}^{1/2}\right].
\]
Combining this with \eqref{eq:mu-over-a-universal}, the left hand side of \eqref{eq:metric-estimation-testing-condition} is bounded by
\[
C\left\{\frac{\log(2q_n)}{ns_\lambda^2}\right\}^{1/2}.
\]
At the optimal smoothing scale, $s_\lambda\gtrsim n^{-1/(2r)}$, so
\[
ns_\lambda^2\gtrsim n^{1-1/r}\to\infty.
\]
Since $\log q_n=O(\log n)$, the logarithmic factor is negligible for $r\geq2$. Thus \eqref{eq:metric-estimation-testing-condition} holds for Lipschitz scores.

For a bounded jump score, $L_\delta=O(\delta^{-1})$. Substituting \eqref{eq:explicit-jump-delta} into the definition of $\epsilon_{n,q,\lambda,\delta}$ gives
\[
\epsilon_{n,q_n,\lambda,\delta}
=O\left[\left\{\frac{\log(2q_n)}{ns_\lambda}\right\}^{\alpha/(2\alpha+1)}\right].
\]
After multiplication by the factor $s_\lambda^{-1/2}$ in \eqref{eq:mu-over-a-universal}, the largest order occurs at $s_\lambda\asymp n^{-1/(2r)}$. Ignoring the logarithmic factor, which is negligible under a strict polynomial inequality, convergence to zero is equivalent to
\[
\frac{\alpha}{2\alpha+1}\left(1-\frac1{2r}\right)>\frac1{4r}.
\]
This reduces to \eqref{eq:jump-holder-threshold}. Hence \eqref{eq:metric-estimation-testing-condition} also holds for the bounded jump score class.

Theorem~\ref{thm:feasible-finite-basis} and Corollary~\ref{cor:bootstrap-optimal-scale} now give asymptotic size $\gamma$ throughout the sparse, transition, and dense regimes.
\end{proof}

\begin{remark}
For absolute loss, the score is bounded and discontinuous, so $L_\delta=O(\delta^{-1})$. If the conditional residual density is symmetric about zero and locally Lipschitz there, then $\alpha=1$. Since
\(1>\frac{1}{4(r-1)}\)
for every $r\geq2$, condition \eqref{eq:jump-holder-threshold} is automatically satisfied. Thus the fully feasible calibration above applies to absolute loss throughout the sparse to dense transition.
\end{remark}

\subsection{Power of the fully feasible statistic}
\label{subsec:feasible-power}

Under an alternative $g=g_0+h$, the residual entering the curvature estimator is $e_{ij}^0=e_{ij}+h(T_{ij})$. Hence the expectation of $w_\delta(e_{ij}^0)$ is a shifted secant of the population score map rather than $I(t)$.

Recall the Sobolev ball $\mathcal B_R$ and let
\(M_R=\sup_{h\in\mathcal B_R}\|h\|_\infty<\infty.\)

\begin{assumption}\label{ass:shifted-secant}
There exist constants $0<c_{\mathrm{sc}}\leq C_{\mathrm{sc}}<\infty$ and $\delta_0>0$ such that, uniformly in $t\in[0,1]$, $|a|\leq M_R$, and $0<\delta\leq\delta_0$,
\begin{equation}
c_{\mathrm{sc}}I(t)\leq I_{\delta,t}(a):=\frac{\Psi_t(a-\delta)-\Psi_t(a+\delta)}{2\delta}\leq C_{\mathrm{sc}}I(t).
\label{eq:shifted-secant-comparability}
\end{equation}
\end{assumption}

Since $a\mapsto\Psi_t(a)$ is nonincreasing, $I_{\delta,t}(a)\geq0$. Assumption~\ref{ass:shifted-secant} strengthens the local identification condition in Assumption~\ref{ass:loss-primitive} only over the compact range of shifts generated by $\mathcal B_R$. It is used to obtain uniform power over the full Sobolev ball. For shrinking alternatives, a weaker assumption may suffice, but the local expansion of $\Psi_t$ alone does not control secants at an independently shrinking bandwidth. We retain Assumption~\ref{ass:shifted-secant} in the feasible power result.

For $h\in\mathcal B_R$, define
\begin{equation*}
A_{q,\delta}(h)=\int_0^1 I_{\delta,t}\{-h(t)\}B(t)B(t)^{\prime}\pi(t)\,dt
\end{equation*}
and
\(H_{q,\lambda,\delta}(h)=A_{q,\delta}(h)+\lambda P_q.\)
The argument $-h(t)$ follows from $e_{ij}^0=e_{ij}+h(T_{ij})$ and the definition $\Psi_t(a)=E\{\psi(e_{ij}-a)\mid T_{ij}=t\}$.

\begin{lemma}\label{lem:shifted-metric-concentration}
Suppose Assumptions~\ref{ass:loss-primitive} and~\ref{ass:shifted-secant} and \eqref{eq:curvature-envelope} hold. Uniformly over $h\in\mathcal B_R$,
\begin{equation}
c_HH_{q,\lambda}\preceq H_{q,\lambda,\delta}(h)\preceq C_HH_{q,\lambda}.
\label{eq:shifted-H-comparison}
\end{equation}
Let
\[
\zeta_{n,q,\lambda,\delta}
=\left\{\frac{L_\delta\log(2q)}{ns_\lambda}\right\}^{1/2}
+\frac{L_\delta\log(2q)}{ns_\lambda}.
\]
Then, uniformly over $h\in\mathcal B_R$,
\begin{equation}
\left\|H_{q,\lambda,\delta}(h)^{-1/2}
\{\widehat H_{q,\lambda}-H_{q,\lambda,\delta}(h)\}
H_{q,\lambda,\delta}(h)^{-1/2}\right\|_{\mathrm{op}}
=O_{P_h}(\zeta_{n,q,\lambda,\delta}).
\label{eq:shifted-H-concentration}
\end{equation}
\end{lemma}

\begin{proof}
For every $u\in\mathbb R^q$, Assumption~\ref{ass:shifted-secant} gives
\[
c_{\mathrm{sc}}u^{\prime}A_qu\leq u^{\prime}A_{q,\delta}(h)u\leq C_{\mathrm{sc}}u^{\prime}A_qu.
\]
Adding $\lambda u^{\prime}P_qu$ to all three terms proves \eqref{eq:shifted-H-comparison} with
$c_H=\min(c_{\mathrm{sc}},1)$ and $C_H=\max(C_{\mathrm{sc}},1)$.

Under $P_h$, conditional on $T_{ij}=t$,
\[
E_h\{w_\delta(e_{ij}^0)\mid T_{ij}=t\}=I_{\delta,t}\{-h(t)\}.
\]
Set $x_{q,h}(t)=H_{q,\lambda,\delta}(h)^{-1/2}B(t)$. By \eqref{eq:shifted-H-comparison} and the evaluation inequality,
\[
\|x_{q,h}(t)\|^2\leq C B(t)^{\prime}H_{q,\lambda}^{-1}B(t)\leq Cs_\lambda^{-1}
\]
uniformly over $h\in\mathcal B_R$.

Define
\[
Y_i(h)=\frac1{m_i}\sum_{j=1}^{m_i}w_\delta(e_{ij}^0)x_{q,h}(T_{ij})x_{q,h}(T_{ij})^{\prime}.
\]
Then $0\preceq Y_i(h)\preceq CL_\delta s_\lambda^{-1}I_q$ uniformly in $h$. Moreover,
\[
E_hY_i(h)=H_{q,\lambda,\delta}(h)^{-1/2}A_{q,\delta}(h)H_{q,\lambda,\delta}(h)^{-1/2},
\]
so $\|E_hY_i(h)\|_{\mathrm{op}}\leq1$. Since the penalty term is identical in $\widehat H_{q,\lambda}$ and $H_{q,\lambda,\delta}(h)$, it cancels from their difference. Therefore the same self adjoint matrix Bernstein argument as in Theorem~\ref{thm:curvature-matrix-rate} gives
\[
\left\|\frac1n\sum_{i=1}^n\{Y_i(h)-E_hY_i(h)\}\right\|_{\mathrm{op}}
=O_{P_h}(\zeta_{n,q,\lambda,\delta})
\]
with constants independent of $h\in\mathcal B_R$. This is exactly \eqref{eq:shifted-H-concentration}.
\end{proof}

\begin{lemma}\label{lem:stable-spline-direction}
Let $\mathcal S_q$ be a quasi-uniform B-spline space of fixed degree $p\geq2r-1$. There exists a linear spline projector $Q_q:W^{r,2}[0,1]\to\mathcal S_q$ such that
\begin{equation}
\|Q_qh-h\|_{L^2}\leq Cq^{-r}\|h^{(r)}\|_{L^2}
\label{eq:spline-direction-L2}
\end{equation}
and
\begin{equation}
J(Q_qh,Q_qh)\leq CJ(h,h).
\label{eq:spline-direction-stability}
\end{equation}
Uniformly over $h\in\mathcal B_R$,
\begin{equation}
V(Q_qh-h,Q_qh-h)^{1/2}\leq C_Rq^{-r}.
\label{eq:spline-direction-V}
\end{equation}
\end{lemma}

\begin{proof}
Define $Q_qh$ by $J(Q_qh,v)=J(h,v)$ for every $v\in\mathcal S_q$, with the additional moment conditions
\[
\int_0^1(Q_qh-h)p\,dt=0,\qquad p\in\mathcal P_{r-1}.
\]
The polynomial moments determine the null-space component, and coercivity of $J$ on its orthogonal complement gives a unique linear projector. Taking $v=Q_qh$ and applying Cauchy--Schwarz proves \eqref{eq:spline-direction-stability}, with constant one.

Put $e=h-Q_qh$. Since $e$ is $L^2$-orthogonal to $\mathcal P_{r-1}$, the weak problem $J(z,v)=\int_0^1ev\,dt$ has a solution in the same orthogonal complement. The one-dimensional natural boundary value problem gives $\|z\|_{H^{2r}}\leq C\|e\|_{L^2}$. The spline approximation bound in Corollary~2 of \citet{sandeetal2020} supplies $z_q\in\mathcal S_q$ with $\|(z-z_q)^{(r)}\|_{L^2}\leq Cq^{-r}\|z\|_{H^{2r}}$. Galerkin orthogonality then yields
\[
\|e\|_{L^2}^2=J(e,z-z_q)
\leq Cq^{-r}\|e^{(r)}\|_{L^2}\|e\|_{L^2}.
\]
The derivative orthogonality also gives $\|e^{(r)}\|_{L^2}\leq\|h^{(r)}\|_{L^2}$. This proves \eqref{eq:spline-direction-L2}. The upper bound on $I\pi$ and $J(h,h)\leq R^2$ give \eqref{eq:spline-direction-V}.
\end{proof}

For $h\in\mathcal B_R$, let
\[
b_q(h)=E_h(s_i)=\int_0^1\Psi_t\{-h(t)\}B(t)\pi(t)\,dt.
\]
Although the cluster sizes may vary, this expectation is common to all subjects because each subject score is an average over identically distributed design points.

\begin{lemma}\label{lem:feasible-signal-lower}
Under Assumption~\ref{ass:shifted-secant}, suppose $h\in\mathcal B_R$ satisfies
\begin{equation}
V(h,h)^{1/2}\geq\rho_n,\qquad q^{-r}=o(\rho_n).
\label{eq:finite-signal-resolution}
\end{equation}
Then, uniformly over such sequences,
\begin{equation}
b_q(h)^{\prime}H_{q,\lambda,\delta}(h)^{-1}b_q(h)
\geq c_R\frac{V(h,h)^2}{V(h,h)+\lambda R^2}\{1+o(1)\}.
\label{eq:finite-shifted-signal-lower}
\end{equation}
\end{lemma}

\begin{proof}
Assumption~\ref{ass:shifted-secant} implies a uniform Lipschitz bound for $\Psi_t$ on $[-M_R,M_R]$. Indeed, if $x<y$ and $y-x\leq2\delta_0$, take $a=(x+y)/2$ and $\delta=(y-x)/2$ in \eqref{eq:shifted-secant-comparability}. Then
\[
0\leq\Psi_t(x)-\Psi_t(y)\leq C I(t)(y-x).
\]
Partitioning longer intervals and using $\Psi_t(0)=0$ gives
\begin{equation}
|\Psi_t(a)|\leq CI(t)|a|,\qquad |a|\leq M_R.
\label{eq:Psi-global-linear-bound}
\end{equation}

Let $v_q=Q_qh$ and write $v_q(t)=B(t)^{\prime}\theta_q$. By Lemma~\ref{lem:global-identification},
\[
\int_0^1\Psi_t\{-h(t)\}h(t)\pi(t)\,dt\geq c_RV(h,h).
\]
Using \eqref{eq:Psi-global-linear-bound}, Cauchy--Schwarz in the $V$ inner product, and Lemma~\ref{lem:stable-spline-direction},
\[
\left|\int_0^1\Psi_t\{-h(t)\}\{v_q(t)-h(t)\}\pi(t)\,dt\right|
\leq C_RV(h,h)^{1/2}q^{-r}=o\{V(h,h)\},
\]
where the last relation follows from \eqref{eq:finite-signal-resolution}. Hence
\begin{equation}
b_q(h)^{\prime}\theta_q\geq\{c_R+o(1)\}V(h,h).
\label{eq:finite-direction-numerator}
\end{equation}

By \eqref{eq:shifted-H-comparison}, Lemma~\ref{lem:stable-spline-direction}, and \eqref{eq:spline-direction-V},
\[
\theta_q^{\prime}H_{q,\lambda,\delta}(h)\theta_q
\leq C\{V(v_q,v_q)+\lambda J(v_q,v_q)\}
\leq C_R\{V(h,h)+q^{-2r}+\lambda R^2\}.
\]
Since $q^{-r}=o(\rho_n)$ and $V(h,h)^{1/2}\geq\rho_n$,
\(q^{-2r}=o\{V(h,h)\}.\)
Therefore
\[
\theta_q^{\prime}H_{q,\lambda,\delta}(h)\theta_q
\leq C_R\{V(h,h)+\lambda R^2\}\{1+o(1)\}.
\]
Finally, using
\[
b^{\prime}H^{-1}b=\sup_{\theta\neq0}\frac{(b^{\prime}\theta)^2}{\theta^{\prime}H\theta}
\]
with $\theta=\theta_q$ and applying \eqref{eq:finite-direction-numerator} proves \eqref{eq:finite-shifted-signal-lower}.
\end{proof}

Let $\widehat c_{n,\gamma}^{*(q)}$ denote the conditional $(1-\gamma)$ quantile of the feasible multiplier statistic $\widehat S_{n,\lambda}^{*(q)}$, computed with the same $\widehat H_{q,\lambda}$ as the observed statistic.

\begin{theorem}\label{thm:fully-feasible-power}
Let $r\geq2$ and suppose the loss and sampling assumptions of Corollary~\ref{cor:fully-feasible-basis} and Assumption~\ref{ass:shifted-secant} hold. Choose $s_\lambda\asymp s_n^*$ and $q_n,\delta_n$ as in Corollary~\ref{cor:fully-feasible-basis}; for bounded jump scores also suppose \eqref{eq:jump-holder-threshold} holds. Let $M_n\to\infty$ satisfy $M_n\rho_n^*\to0$. Then
\begin{equation}
\inf_{\substack{h\in\mathcal B_R:\\ V(h,h)^{1/2}\geq M_n\rho_n^*}}
P_h\{\widehat S_{n,\lambda}^{(q_n)}>\widehat c_{n,\gamma}^{*(q_n)}\}\to1.
\label{eq:fully-feasible-uniform-power}
\end{equation}
Together with Corollary~\ref{cor:fully-feasible-basis}, the fully feasible test has asymptotic level $\gamma$ and consistency separation rate $\rho_n^*$. This rate matches the lower bound over model classes containing the Gaussian submodel of Theorem~\ref{thm:unified-minimax-lower}.
\end{theorem}

\begin{proof}
Write $H_h=H_{q_n,\lambda,\delta_n}(h)$ and define the ideal shifted metric statistic
\(S_{n,h}^{(q_n)}=n\bar s^{\prime}H_h^{-1}\bar s.\)
Lemma~\ref{lem:shifted-metric-concentration} and the choices in Corollary~\ref{cor:fully-feasible-basis} give, uniformly over $h\in\mathcal B_R$,
\[
\|H_h^{-1/2}(\widehat H_{q_n,\lambda}-H_h)H_h^{-1/2}\|_{\mathrm{op}}=o_{P_h}(1).
\]
Moreover, the rate calculation in that corollary gives
\begin{equation}
\zeta_{n,q_n,\lambda,\delta_n}s_\lambda^{-1/2}\to0.
\label{eq:shifted-metric-testing-small}
\end{equation}
Indeed, $\zeta_{n,q,\lambda,\delta}\leq\epsilon_{n,q,\lambda,\delta}$ and the metric condition was established after multiplication by $s_\lambda^{-1/2}$.

Put $E_h=H_h^{-1/2}(\widehat H_{q_n,\lambda}-H_h)H_h^{-1/2}$. Then
\[
\widehat H_{q_n,\lambda}^{-1}=H_h^{-1/2}(I+E_h)^{-1}H_h^{-1/2}
\]
and, with probability tending uniformly to one,
\[
\|(I+E_h)^{-1}-I\|_{\mathrm{op}}\leq C\|E_h\|_{\mathrm{op}}.
\]
Thus replacing $H_h^{-1}$ by $\widehat H_{q_n,\lambda}^{-1}$ produces a relative $o_{P_h}(1)$ perturbation of any corresponding quadratic form. For the centered stochastic quadratic term and the bootstrap critical value, \eqref{eq:shifted-metric-testing-small} makes the perturbation $o_{P_h}(a_{n,\lambda})$. The deterministic signal has a relative $o_{P_h}(1)$ error.

We next lower bound the deterministic signal. By \eqref{eq:explicit-optimal-q},
\[
q_n^{-r}\asymp L_n^{-r}s_\lambda^{2r^2/(2r-1)}=o(s_\lambda^r).
\]
At the optimal smoothing rate, $s_\lambda^r\asymp\rho_n^*$, so $q_n^{-r}=o(\rho_n^*)$. Lemma~\ref{lem:feasible-signal-lower} therefore gives uniformly over the alternative set,
\[
b_q(h)^{\prime}H_h^{-1}b_q(h)
\geq c_R\frac{V(h,h)^2}{V(h,h)+\lambda R^2}\{1+o(1)\}.
\]
Since $\lambda=s_\lambda^{2r}\asymp(\rho_n^*)^2$ and $V(h,h)^{1/2}\geq M_n\rho_n^*$ with $M_n\to\infty$,
\[
b_q(h)^{\prime}H_h^{-1}b_q(h)\geq c_RV(h,h)
\]
for all sufficiently large $n$. Hence
\begin{equation}
n\,b_q(h)^{\prime}H_h^{-1}b_q(h)
\geq c_RM_n^2n(\rho_n^*)^2.
\label{eq:feasible-deterministic-signal}
\end{equation}

Let $d_h=H_h^{-1/2}b_q(h)$ and let $\Omega_{q,h}$ denote the covariance operator of the centered transformed subject score. The shifted moment bound is uniform over $\mathcal B_R$. Under growing clusters, it bounds the shifted pair kernel in $L^2([0,1]^2)$, while the diagonal multiplication operator remains bounded. Conjugation by $R_\lambda$ therefore gives a pair Hilbert--Schmidt norm of order one and a diagonal Hilbert--Schmidt norm of order $s_\lambda^{-1/2}$. The corresponding trace bounds follow by pairing with $R_\lambda^2$. Projection and the metric comparison \eqref{eq:shifted-H-comparison} preserve these bounds. Under bounded clusters the earlier comparison with $A_\lambda$ gives the same bounds, with constants depending on $m_{\max}$. Thus
\[
\operatorname{tr}(\Omega_{q,h})
\leq C\{m_H^{-1}s_\lambda^{-1}+s_\lambda^{-1/2}\},
\]
whereas
\[
\{2\operatorname{tr}(\Omega_{q,h}^2)\}^{1/2}
\leq C\{m_H^{-2}s_\lambda^{-1}+1\}^{1/2}
\asymp Ca_{n,\lambda}.
\]
The fourth moment argument of Lemma~\ref{lem:bootstrap-covariance-consistency} depends only on the uniform shifted moment bound and therefore applies uniformly under $P_h$. Consequently, for the empirical covariance in the shifted metric,
\[
\widehat\mu_{q,h}
=\operatorname{tr}(\Omega_{q,h})+o_{P_h}(a_{n,\lambda})
\]
and
\[
\widehat a_{q,h}=O_{P_h}(a_{n,\lambda})
\]
uniformly over $h\in\mathcal B_R$.

Conditionally on the data, the ideal shifted metric multiplier statistic is a Gaussian quadratic form. Chebyshev's inequality therefore gives
\begin{equation}
c_{q,h,\gamma}^*
\leq\widehat\mu_{q,h}+\gamma^{-1/2}\widehat a_{q,h}
=\operatorname{tr}(\Omega_{q,h})+O_{P_h}(a_{n,\lambda}).
\label{eq:feasible-bootstrap-critical-order}
\end{equation}
By \eqref{eq:shifted-metric-testing-small}, replacing $H_h$ by $\widehat H_{q_n,\lambda}$ changes this bootstrap critical value by $o_{P_h}(a_{n,\lambda})$ on the testing scale.

Finally, write
\[
S_{n,h}^{(q_n)}
=n\|d_h\|^2+2\sqrt n\langle d_h,Z_h\rangle+\|Z_h\|^2,
\]
where $Z_h$ is the centered normalized transformed score sum. Uniformly over $h\in\mathcal B_R$,
\[
\|Z_h\|^2=\operatorname{tr}(\Omega_{q,h})+O_{P_h}(a_{n,\lambda}).
\]
Also, $\|\Omega_{q,h}\|_{\mathrm{op}}\leq C$, so
\[
\operatorname{Var}_h\{2\sqrt n\langle d_h,Z_h\rangle\}\leq Cn\|d_h\|^2.
\]
Thus the linear term is $O_{P_h}\{(n\|d_h\|^2)^{1/2}\}$.

At the optimal smoothing rate,
\[
n(\rho_n^*)^2\asymp a_{n,\lambda}
\]
up to constants in the three sampling regimes. Hence \eqref{eq:feasible-deterministic-signal} and $M_n\to\infty$ imply
\[
\frac{n\|d_h\|^2}{a_{n,\lambda}}\to\infty
\]
and
\[
\frac{(n\|d_h\|^2)^{1/2}}{n\|d_h\|^2}\to0
\]
uniformly over the stated alternative set. Comparing the preceding decomposition with \eqref{eq:feasible-bootstrap-critical-order} shows that the ideal shifted metric rejection probability tends uniformly to one. The metric perturbation is negligible relative to the stochastic scale or, when the signal is larger, relative to the signal itself. Therefore the same conclusion holds for the fully feasible statistic and its feasible bootstrap critical value, proving \eqref{eq:fully-feasible-uniform-power}.
\end{proof}

\begin{proposition}\label{prop:shifted-secant-examples}
Assumption~\ref{ass:shifted-secant} holds for squared loss; for absolute loss if the conditional residual density $f_t$ is uniformly bounded above and away from zero on $|x|\leq M_R+\delta_0$; and for Huber loss if
\[
p_t(u)=P(|e_{ij}-u|<c\mid T_{ij}=t)
\]
is uniformly bounded above and away from zero for $|u|\leq M_R+\delta_0$.
\end{proposition}

\begin{proof}
For squared loss, $\Psi_t(a)=-a$, so $I_{\delta,t}(a)=I(t)=1$.

For absolute loss,
\[
I_{\delta,t}(a)=\frac{F_t(a+\delta)-F_t(a-\delta)}{\delta}
=\frac1\delta\int_{a-\delta}^{a+\delta}f_t(u)\,du.
\]
Since $I(t)=2f_t(0)$, the assumed uniform upper and lower bounds on $f_t$ imply \eqref{eq:shifted-secant-comparability}.

For Huber loss,
\[
I_{\delta,t}(a)=\frac1{2\delta}\int_{a-\delta}^{a+\delta}p_t(u)\,du.
\]
Since $I(t)=p_t(0)$, the uniform upper and lower bounds on $p_t$ again imply \eqref{eq:shifted-secant-comparability}.

A sufficient primitive condition for the Huber lower bound is
\[
f_t(x)\geq c_R>0,\qquad |x|\leq M_R+\delta_0+c,
\]
uniformly in $t$. Indeed, for $|u|\leq M_R+\delta_0$,
\[
p_t(u)=\int_{u-c}^{u+c}f_t(x)\,dx\geq2cc_R.
\]
Since $p_t(u)\leq1$, this gives the required uniform comparability and completes the proof.
\end{proof}

\section{Simulation design and implementation}
\label{sec:simulation}

The simulations examine size, power against three signal shapes, and sensitivity to heavy tails and contamination. We also vary the number of measurements per subject. Every multiplier replicate assigns one Gaussian weight to each subject, preserving the grouping of repeated observations.

The data are generated according to the observation model used in the theory. The null mean is $g_0(t)=0$ on $[0,1]$. Subject level latent variables are independent across subjects, and the measurement errors are independent of the latent trajectories and observation times.

Given a target sampling level $m_0$, we sample $m_i$ uniformly from the integers between $\max(2,m_0-1)$ and $m_0+1$, inclusive. Thus $m_0=2$ gives equal probabilities to sizes two and three. Conditional on $m_i$, the observation times are independent $\mathrm{Unif}(0,1)$ draws and are sorted only for computation.

An independent latent trajectory is then generated for each subject from the 50 component Karhunen--Lo\`eve expansion
\[
Z_i(t)=\sum_{k=1}^{50}\zeta_k\xi_{ik}\phi_k(t),
\]
where
\[
\phi_1(t)=1,\qquad \phi_k(t)=\sqrt2\cos\{(k-1)\pi t\},\quad k=2,\ldots,50.
\]
The scores $\xi_{ik}$ are independent $\mathrm{Unif}[-\sqrt3,\sqrt3]$ variables and therefore have mean zero and variance one. The trajectory eigenvalues satisfy $\zeta_k^2\propto k^{-1.1}$ and are normalized so that $\sum_{k=1}^{50}\zeta_k^2=0.5$. Since the cosine system is orthonormal in $L^2[0,1]$,
\[
E\{\|Z_i\|_{L^2}^2\}=0.5.
\]
The 50 components give a range of within subject covariance directions. This construction follows the broad simulation principle used by \citet{caiyuan2011}.

Measurement errors are generated independently across observations from one of three symmetric unit variance distributions: standard Gaussian, standardized Student $t_5$, or the centered mixture $0.95N(0,1)+0.05N(0,25)$ rescaled to variance one. The observation level error entering the response is $0.60\,\varepsilon_{ij}$. The latent KL process is identical across the three error settings, so changing the error distribution isolates sensitivity to heavy tails and contamination without changing the within subject covariance. Symmetry also aligns the mean, median, and Huber center under the null.

Under the alternative, a normalized departure $h$ is added to $g_0$. Three signal shapes are considered. The diffuse global departure is
\[
h_{\mathrm G}(t)=C_{\mathrm G}\sum_{k=1}^{10}(-1)^{k+1}k^{-1}\sqrt2\cos(k\pi t),
\]
where
\[
C_{\mathrm G}=\left(\sum_{k=1}^{10}k^{-2}\right)^{-1/2}.
\]
This alternative spreads signal across ten low frequency directions while assigning progressively less weight to higher frequencies. The oscillatory departure is
\[
h_{\mathrm O}(t)=\sqrt2\sin(6\pi t).
\]
The localized departure is
\[
h_{\mathrm L}(t)=C_{\mathrm L}\exp\left[-\frac{(t-0.55)^2}{2(0.07)^2}\right],
\]
where $C_{\mathrm L}$ is chosen so that $\|h_{\mathrm L}\|_{L^2}=1$. Hence all three alternatives have unit $L^2$ norm and differ only in their spatial and spectral structure.

For a nonnull cell, the departure amplitude is tied to the minimax separation rate established in Section~\ref{sec:growing-m}. Because the generated cluster sizes remain comparable to the target sampling level $m_0$, their harmonic mean satisfies $m_H\asymp m_0$. We therefore set
\[
a=c_{\mathrm{sig}}\rho_{n,m_0}^*,
\]
where $c_{\mathrm{sig}}>0$ is a dimensionless signal multiplier and, for $r=2$,
\[
\rho_{n,m_0}^*=(nm_0)^{-2r/(4r+1)}\vee n^{-1/2}.
\]
The multiplier $c_{\mathrm{sig}}$ sets the signal size relative to the separation rate, while $h$ determines its shape. Under the null, $a=0$.

After $m_i$, $T_{ij}$, $Z_i$, $\varepsilon_{ij}$, and the mean departure have been generated, the response is
\[
U_{ij}=g_0(T_{ij})+a\,h(T_{ij})+Z_i(T_{ij})+0.60\,\varepsilon_{ij}.
\]
Consequently,
\[
E\{U_{ij}\mid T_{ij}=t\}=g_0(t)+a h(t),
\]
while dependence among repeated measurements from the same subject is induced by the common latent trajectory $Z_i$.

The proposed score test is implemented with squared, Huber, and absolute loss and is denoted RGT-Squared, RGT-Huber, and RGT-Absolute. The Huber clipping constant is $1.345$. For absolute loss, the secant bandwidth is
\[
\delta=\{\log(2q)/(ns_\lambda)\}^{1/3},
\]
truncated to $[0.05,0.75]$. The squared loss curvature is exact, while the Huber secant uses $\delta=0.05$.

Two nonrobust global benchmarks are constructed from the same squared loss B spline smoother and calibrated with the same subject level multiplier bootstrap. L2-BS uses the integrated squared fitted departure, whereas Linf-BS uses the maximum absolute fitted departure standardized pointwise by its estimated standard deviation. Integrated $L^2$ testing is standard in functional mean comparisons; see \citet{staicu2015}. Supremum type procedures are more sensitive to localized alternatives, with the $F_{\max}$ construction of \citet{zhangetal2019} providing a representative example. L2-BS and Linf-BS are one sample B spline analogues adapted to the present irregular design and are not reproductions of those published statistics.

A separate paired experiment includes two published comparators. The first is pLRT-Staicu, the pseudo likelihood ratio test of \citet{staicu2014plrt}. The method compares a parametric null mean with a penalized spline alternative after estimating the within subject covariance and pre-whitening the observations. We apply it to $U_{ij}-g_0(T_{ij})$, so the null mean is identically zero. The alternative is represented by a linear penalized spline with 20 interior knots, matching the sparse data implementation used in \citet{staicu2014plrt}.

To estimate the covariance required by pLRT-Staicu, a flexible pooled spline is first removed from the observations. Off diagonal products of residuals from the same subject are pooled, symmetrized, binned on a $25\times25$ grid for computational efficiency, and smoothed with a bivariate thin plate regression spline using modified GCV with $\gamma=1.5$. A separate one dimensional smooth of the squared residuals estimates the total pointwise variance. The measurement error variance is obtained from the integrated positive part of the difference between the total variance smooth and the diagonal of the estimated latent covariance. Functional principal components with positive eigenvalues are retained until at least $99\%$ of the estimated latent covariance variation is explained. The resulting components define a positive semidefinite covariance matrix for each subject and are used for pre-whitening.

After pre-whitening, the pLRT statistic is the likelihood ratio statistic associated with the linear spline mixed model representation, with both the fixed polynomial component and the spline variance component restricted under the null. Its finite sample null distribution is simulated with the spectral algorithm of \citet{crainiceanuruppert2004}, implemented through \texttt{RLRsim::LRTSim}; see \citet{scheipl2008}. Each Monte Carlo data set uses 999 simulated null statistics. Because covariance smoothing and pLRT computation are substantially more expensive than the subject level score procedures, this published comparator is evaluated in a focused paired experiment rather than throughout the full simulation grid.

The second published comparator is the simultaneous confidence band procedure of \citet{louwu2026}, denoted LW-SCB. We follow their local linear construction with Epanechnikov kernel, subject weighting by $1/m_i$, jackknife bias correction based on bandwidths $b$ and $\sqrt{2}b$, and one Gaussian multiplier per subject. The multiplier process uses the corresponding higher order kernel $2K(u)-K(u/\sqrt{2})/\sqrt{2}$. For each combination of $m_0\in\{3,40\}$ and error distribution, the production bandwidth is the median of 100 independently selected leave one subject out cross-validation bandwidths over the published search range $0<b\leq0.30$. The resulting median is $b=0.30$ in all six design settings. The selected bandwidth is then held fixed across the null and the three alternative shapes within that design setting. LW-SCB is evaluated on the same generated samples, using 399 multiplier draws per replication. Its reported simulation rejection rates use the observed statistic exceeding the empirical $0.95$ bootstrap quantile. With finitely many draws, this rule can differ from thresholding the plus-one $p$-value at $0.05$.

For the five spline score procedures, cubic B splines are used with $r=2$. The smoothing scale follows the theoretical sparse to dense rule,
\[
s_\lambda=\max\{(nm_H)^{-2/(4r+1)},n^{-1/(2r)}\},
\]
where $m_H$ is the realized harmonic mean of the cluster sizes, and $\lambda=s_\lambda^{2r}$. The basis dimension is
\[
q=\left\lceil1.5s_\lambda^{-2r/(2r-1)}\right\rceil,
\]
truncated to $[12,35]$. Numerical quadrature uses 401 equally spaced points. These are finite sample tuning choices: the fixed basis cap, constant factor $1.5$, and positive secant bandwidth floor do not satisfy the diverging-dimension and vanishing-bandwidth conditions of Corollary~\ref{cor:fully-feasible-basis} as $n\to\infty$. The simulations assess their practical performance.

\subsection{Simulation settings}

The main size experiment uses $n\in\{100,150,200\}$, target within subject sample sizes $m_0\in\{3,10,40\}$, and all three measurement error distributions. Each of the 27 null settings uses 1000 Monte Carlo replications. The nominal level is $0.05$. For a cell with $R$ replications and rejection proportion $\widehat p$, its Monte Carlo standard error is $\{\widehat p(1-\widehat p)/R\}^{1/2}$. At rejection probability $0.05$, this is about $0.0069$ for $R=1000$; the largest possible standard errors are $0.0224$ for $R=500$ and $0.0289$ for $R=300$. The tables summarize cells and do not report uncertainty for paired differences.

The balanced core power experiment uses $n\in\{100,150,200\}$ and $m_0\in\{3,40\}$ and crosses the three error mechanisms with the diffuse, localized, and oscillatory alternatives. The signal grid depends on the alternative shape. For the diffuse alternative, $c_{\mathrm{sig}}\in\{0.5,1,1.5,2\}$; for the localized alternative, $c_{\mathrm{sig}}\in\{1.5,2,2.5,3.5\}$; and for the oscillatory alternative, $c_{\mathrm{sig}}\in\{3,4,5,6\}$. Each of the 216 cells uses 500 Monte Carlo replications. These shape specific grids cover nondegenerate power ranges despite the stronger smoothing attenuation of the oscillatory alternative. Comparisons among procedures are made within identical cells; the values of $c_{\mathrm{sig}}$ are not intended to compare intrinsic difficulty across alternative shapes.

Because the amplitude is scaled as $a=c_{\mathrm{sig}}\rho_{n,m_0}^*$, the absolute departure decreases as $n$ increases. The sample size dimension therefore evaluates performance at comparable multiples of the theoretical separation boundary rather than under a fixed nonvanishing effect. Approximate stability of power across $n$ is consequently compatible with the predicted rate, although finite sample changes may remain because the alternative geometry and spline approximation interact with $n$ and the sampling density.

For compact presentation, we report a robustness focused summary of the sparse global configurations together with separate shape specific challenges. The summary uses $m_0=3$, the diffuse alternative, $c_{\mathrm{sig}}\in\{1.5,2\}$, and $n\in\{100,150,200\}$, and combines Gaussian errors, standardized $t_5$ errors, the $5\%$ observation level contamination model, and two subject level functional contamination regimes with amplitudes $A=5$ and $A=7.5$. It therefore contains 30 power configurations drawn from the broader simulation program. The full balanced core design is described above, while the localized and oscillatory alternatives are retained separately so that settings favoring specialized competitors remain visible.

A further experiment isolates the effect of within trajectory sampling. It uses $n=200$, Gaussian errors, the diffuse alternative, $c_{\mathrm{sig}}=1.5$, and target values $m_0\in\{2,3,5,10,20,40\}$. Each setting uses 500 replications. The experiment records the realized harmonic mean $m_H$, the selected $s_\lambda$, and the empirical rejection probability, allowing the numerical transition to be compared directly with the covariance decomposition in Section~\ref{sec:growing-m}.

The named method comparison is paired within each generated data set. Its size component fixes $n=150$, uses $m_0\in\{3,40\}$, and includes Gaussian, standardized $t_5$, and contaminated errors. Its power component uses the same sampling levels and error laws and considers all three alternative shapes at representative nonsaturated signals: $c_{\mathrm{sig}}=1.5$ for the diffuse alternative, $c_{\mathrm{sig}}=2.5$ for the localized alternative, and $c_{\mathrm{sig}}=5$ for the oscillatory alternative. Each cell uses 300 Monte Carlo replications. The comparison includes RGT-Squared, RGT-Huber, RGT-Absolute, L2-BS, Linf-BS, pLRT-Staicu, and LW-SCB on exactly the same simulated samples.

Subject level functional contamination is qualitatively different from isolated measurement error contamination. Independently across subjects, let $B_i\sim\mathrm{Bernoulli}(0.10)$ and let $S_i$ be a Rademacher variable. The response receives the additional perturbation
\[
C_i(t)=B_iS_iAq(t),
\]
where
\[
q(t)=\{\sqrt2\cos(11\pi t)+\sqrt2\cos(13\pi t)\}/\sqrt2
\]
and $\|q\|_{L^2}=1$. Since $E(S_i)=0$, this contamination leaves the population mean unchanged under both the null and the alternative, while perturbing entire trajectories in directions outside the ten cosine frequencies used by $h_{\mathrm G}$. The focused pLRT stress experiment fixes $n=150$, $m_0=3$, Gaussian measurement errors, and $c_{\mathrm{sig}}=2$, with $A\in\{2.5,5,7.5\}$. For each value of $A$, one null and one alternative setting are evaluated using 500 replications. A cross sample size extension uses $A\in\{5,7.5\}$, $n\in\{100,150,200\}$, and $c_{\mathrm{sig}}\in\{1.5,2\}$ for the five spline score procedures. These cells form the subject contamination component of the robustness focused summary.

Each Monte Carlo replication of the spline score procedures uses 399 Gaussian multiplier draws. The production code stores scenarios separately, records numerical stabilization and failures, and can resume without recomputing completed settings. No ridge stabilization was required in the balanced production runs, and no pLRT fitting failures occurred in the focused named method comparison.

\subsection{Size and power comparisons}

Across the 27 main null settings, all 1000 replications were completed for every procedure and the numerical stabilization frequency was zero. Table~\ref{tab:main-size} summarizes the empirical rejection probabilities over $n\in\{100,150,200\}$, $m_0\in\{3,10,40\}$, and the three error distributions. All five spline score procedures remain close to the nominal level $0.05$. The largest rejection probability among the individual settings is $0.070$ for Linf-BS. RGT-Huber ranges from $0.037$ to $0.065$, while RGT-Squared ranges from $0.036$ to $0.065$.

\begin{table}[t]
\centering
\caption{Empirical size over the 27 main null settings. The first numeric column is the average rejection probability, and the last two columns give the minimum and maximum across settings.}
\label{tab:main-size}
\small
\begin{tabular}{lccc}
\hline
Procedure & Average & Minimum & Maximum\\
\hline
L2-BS         & 0.050 & 0.038 & 0.065\\
Linf-BS       & 0.053 & 0.043 & 0.070\\
RGT-Absolute  & 0.051 & 0.041 & 0.066\\
RGT-Huber     & 0.050 & 0.037 & 0.065\\
RGT-Squared   & 0.049 & 0.036 & 0.065\\
\hline
\end{tabular}
\end{table}

Table~\ref{tab:primary-global} summarizes the 30 sparse global configurations by averaging power over the six combinations of $n\in\{100,150,200\}$ and $c_{\mathrm{sig}}\in\{1.5,2\}$ within each data regime. The last two columns count the configurations in which RGT-Huber and, more generally, one of the three RGT procedures achieves the highest empirical power among the five spline score procedures. An RGT procedure is best in all 30 configurations: RGT-Huber leads in 20 and RGT-Squared in the remaining 10. Under Gaussian and standardized $t_5$ errors, RGT-Squared most often has the highest power, whereas RGT-Huber leads every observation level and subject level contamination cell. These counts concern only the five spline score procedures. The two published comparators, pLRT-Staicu and LW-SCB, are assessed in the focused paired experiment and are not included in these counts. The counts describe empirical rankings, not statistically established differences in power.

\begin{table}[t]
\centering
\caption{Robustness focused sparse global summary. Power columns average the six cells formed by $n\in\{100,150,200\}$ and $c_{\mathrm{sig}}\in\{1.5,2\}$ within each regime. ``Huber wins'' and ``RGT wins'' are counts out of six cells.}
\label{tab:primary-global}
\small
\resizebox{\textwidth}{!}{%
\begin{tabular}{lccccccc}
\hline
Regime & L2-BS & Linf-BS & RGT-Abs. & RGT-Huber & RGT-Sq. & Huber wins & RGT wins\\
\hline
Gaussian                         & 0.556 & 0.549 & 0.396 & 0.551 & 0.573 & 1 & 6\\
Standardized $t_5$               & 0.545 & 0.545 & 0.418 & 0.559 & 0.571 & 1 & 6\\
Observation contamination        & 0.558 & 0.558 & 0.445 & 0.614 & 0.582 & 6 & 6\\
Subject contamination $A=5$      & 0.170 & 0.175 & 0.356 & 0.440 & 0.170 & 6 & 6\\
Subject contamination $A=7.5$    & 0.104 & 0.113 & 0.330 & 0.420 & 0.105 & 6 & 6\\
\hline
\end{tabular}%
}
\end{table}

The subject contamination results are particularly stable across sample size. Over the 12 cells formed by $A\in\{5,7.5\}$, $n\in\{100,150,200\}$, and $c_{\mathrm{sig}}\in\{1.5,2\}$, RGT-Huber has the highest empirical power in every cell. Its advantage over the next strongest spline score procedure ranges from $0.046$ to $0.116$, with an average margin of $0.087$. The improvement therefore persists across both contamination amplitudes, both signal levels, and all three sample sizes rather than being driven by a single stress configuration.

Table~\ref{tab:shape-challenges} reports sparse sampling power for the localized and oscillatory alternatives, averaged over $n\in\{100,150,200\}$ and the three measurement error distributions. The localized alternative consistently favors Linf-BS, as expected from a standardized supremum statistic designed to emphasize spatially concentrated departures. RGT-Huber and RGT-Squared remain competitive as the localized signal becomes stronger, but neither matches Linf-BS on this challenge. For the oscillatory alternative, RGT-Huber or RGT-Squared has the highest average power at every reported signal level, although the differences from L2-BS and Linf-BS are generally modest. RGT-Absolute is less powerful in these regular and moderately heavy tailed settings, illustrating the efficiency cost associated with its more robust nonsmooth score.

\begin{table}[t]
\centering
\caption{Sparse sampling challenge power for $m_0=3$, averaged over $n\in\{100,150,200\}$ and the three measurement error distributions.}
\label{tab:shape-challenges}
\small
\begin{tabular}{llccccc}
\hline
Shape & $c_{\mathrm{sig}}$ & L2-BS & Linf-BS & RGT-Abs. & RGT-Huber & RGT-Sq.\\
\hline
Localized  & 1.5 & 0.190 & 0.246 & 0.161 & 0.215 & 0.208\\
           & 2.0 & 0.312 & 0.410 & 0.268 & 0.355 & 0.358\\
           & 2.5 & 0.446 & 0.582 & 0.382 & 0.521 & 0.522\\
           & 3.5 & 0.742 & 0.867 & 0.639 & 0.822 & 0.845\\
Oscillatory & 3  & 0.181 & 0.181 & 0.143 & 0.194 & 0.186\\
            & 4  & 0.279 & 0.278 & 0.212 & 0.294 & 0.290\\
            & 5  & 0.399 & 0.408 & 0.313 & 0.420 & 0.421\\
            & 6  & 0.513 & 0.514 & 0.413 & 0.551 & 0.544\\
\hline
\end{tabular}
\end{table}

The balanced sample size extension further illustrates the role of the theoretical signal scaling. Since the simulated amplitude is $a=c_{\mathrm{sig}}\rho_{n,m_0}^*$, it decreases with $n$ at fixed $m_0$ and $c_{\mathrm{sig}}$. Power therefore need not increase monotonically with sample size at a fixed signal multiplier. For the diffuse alternatives, the rejection probabilities remain broadly comparable across $n$, which is consistent with testing at fixed multiples of the separation boundary. The localized alternatives exhibit somewhat stronger finite sample gains as $n$ increases, while several cells with $m_0=40$ approach saturation at the larger signal multipliers.

Increasing within subject sampling produces a substantial gain in power. In the dedicated sampling density experiment, the average rejection probability across the five spline score procedures increases from $0.367$ at target $m_0=2$ to $0.410$, $0.594$, $0.776$, $0.912$, and $0.953$ at $m_0=3,5,10,20,40$, respectively. The corresponding realized harmonic means are $2.40$, $2.77$, $4.87$, $9.93$, $19.97$, and $39.98$. For $r=2$ and $n=200$, the asymptotic transition scale $n^{1/(4r)}$ is about $1.94$.
Thus even the sparsest setting lies near or just beyond the asymptotic transition scale, and the experiment primarily illustrates the continued finite sample benefit of increasing repeated measurements rather than providing a literal comparison of the two asymptotic regimes.

The focused paired comparison is summarized in Table~\ref{tab:plrt-power}, with each entry averaged over the three measurement error distributions. The two published procedures display distinct strengths. LW-SCB is especially effective for the localized alternative, reaching power $0.872$ at $m_0=3$, compared with $0.580$ for Linf-BS, $0.646$ for pLRT-Staicu, and $0.526$ for RGT-Huber. pLRT-Staicu is strongest for the selected oscillatory departure and reaches power one already at $m_0=3$. For the sparse diffuse departure, LW-SCB is less powerful, with power $0.349$, whereas pLRT-Staicu, RGT-Squared, and RGT-Huber attain $0.529$, $0.427$, and $0.423$, respectively. At $m_0=40$, all procedures have high power and several are essentially saturated. These results illustrate the expected geometry: a studentized supremum procedure can be highly sensitive to localized structure, whereas a quadratic score aggregates distributed departures.

\begin{table}[t]
\centering
\caption{Paired named method power averaged over the three measurement error distributions. Signals use $c_{\mathrm{sig}}=1.5$ for the diffuse alternative, $2.5$ for the localized alternative, and $5$ for the oscillatory alternative.}
\label{tab:plrt-power}
\small
\resizebox{\textwidth}{!}{%
\begin{tabular}{llccccccc}
\hline
Shape & $m_0$ & L2-BS & Linf-BS & LW-SCB & pLRT-Staicu & RGT-Abs. & RGT-Huber & RGT-Sq.\\
\hline
Diffuse     & 3  & 0.406 & 0.412 & 0.349 & 0.529 & 0.321 & 0.423 & 0.427\\
Diffuse     & 40 & 0.951 & 0.969 & 0.959 & 0.999 & 0.909 & 0.958 & 0.964\\
Localized   & 3  & 0.452 & 0.580 & 0.872 & 0.646 & 0.391 & 0.526 & 0.520\\
Localized   & 40 & 0.930 & 0.974 & 1.000 & 1.000 & 0.950 & 0.974 & 0.978\\
Oscillatory & 3  & 0.420 & 0.424 & 0.789 & 1.000 & 0.308 & 0.442 & 0.447\\
Oscillatory & 40 & 0.971 & 0.988 & 1.000 & 1.000 & 0.953 & 0.982 & 0.981\\
\hline
\end{tabular}%
}
\end{table}

The six LW-SCB null rejection probabilities range from $0.037$ to $0.073$ and average $0.054$. The two largest values, $0.073$ at $m_0=3$ and $0.070$ at $m_0=40$, occur under standardized $t_5$ errors; under Gaussian errors both values are $0.053$, and under the point contamination model both are $0.037$. Thus LW-SCB is close to nominal overall, with mild liberalness under the heavy tailed setting. The corresponding pLRT null rejection probabilities range from $0.057$ to $0.093$ across the six paired size settings, whereas those of the five spline score procedures range from $0.030$ to $0.067$. The largest pLRT rejection probability, $0.093$, occurs under Gaussian errors with $m_0=40$. Repeating this null setting with a finer $50\times50$ covariance binning grid gives a rejection probability of $0.087$ over 150 replications. Although the latter estimate has appreciable Monte Carlo uncertainty, it indicates that the moderate liberalness is not attributable solely to the original $25\times25$ binning approximation. The average number of retained covariance principal components is approximately $2.1$ for $m_0=3$ and $5.0$ for $m_0=40$. Neither published comparator has fitting failures in the production comparison, and their power values should be interpreted together with the corresponding null calibration.

Table~\ref{tab:functional-contamination} gives the focused comparison at $n=150$ that includes pLRT-Staicu. Under mild subject contamination with $A=2.5$, pLRT-Staicu and RGT-Huber have empirical powers $0.554$ and $0.564$, respectively. As the contamination amplitude increases, pLRT power decreases to $0.276$ at $A=5$ and $0.162$ at $A=7.5$, whereas RGT-Huber retains powers $0.512$ and $0.474$. RGT-Absolute is also comparatively resistant, with powers $0.420$ and $0.358$ in the two stronger contamination settings. By contrast, RGT-Squared, L2-BS, and Linf-BS lose most of their power when $A=7.5$. The corresponding RGT-Huber null rejection probabilities are $0.038$, $0.056$, and $0.048$, showing that its power advantage is not explained by inflated type I error.

\begin{table}[t]
\centering
\caption{Subject level functional contamination at $n=150$ and $m_0=3$. Each entry gives empirical size/power at $c_{\mathrm{sig}}=2$, based on 500 null and 500 alternative replications for each contamination amplitude.}
\label{tab:functional-contamination}
\small
\resizebox{\textwidth}{!}{%
\begin{tabular}{lcccccc}
\hline
$A$ & L2-BS & Linf-BS & pLRT-Staicu & RGT-Abs. & RGT-Huber & RGT-Sq.\\
\hline
2.5 & 0.036 / 0.410 & 0.046 / 0.392 & 0.072 / 0.554 & 0.040 / 0.432 & 0.038 / 0.564 & 0.036 / 0.432\\
5.0 & 0.056 / 0.180 & 0.052 / 0.200 & 0.078 / 0.276 & 0.050 / 0.420 & 0.056 / 0.512 & 0.046 / 0.178\\
7.5 & 0.042 / 0.114 & 0.040 / 0.126 & 0.074 / 0.162 & 0.044 / 0.358 & 0.048 / 0.474 & 0.038 / 0.116\\
\hline
\end{tabular}%
}
\end{table}

The comparisons show different strengths across signal shapes and error laws. Huber scores retain power under the contamination settings, while LW-SCB is strongest for the sparse localized signal and pLRT-Staicu for the selected oscillatory signal. These results support examining robustness and signal shape together when interpreting a method's power.

\section{Application to longitudinal CD4 counts}
\label{sec:cd4-application}

We illustrate the procedure using the \texttt{aids} data distributed with the \texttt{joineR} package described by \citet{philipson2018joiner}. The data arise from a randomized clinical trial comparing didanosine (ddI) with zalcitabine (ddC) in HIV infected patients who had failed or were intolerant of zidovudine therapy; see \citet{goldman1996}. The full data contain 1405 longitudinal CD4 measurements from 467 patients, with scheduled measurements at study entry and at 2, 6, 12, and 18 months. We use only the longitudinal CD4 component.

For each patient, the earliest recorded CD4 measurement is taken as baseline. If $C_{ij}$ denotes a subsequent CD4 count at month $t_{ij}$ and $C_{i0}$ denotes the baseline count at $t_{i0}$, we analyze
\[
U_{ij}=\log(1+C_{ij})-\log(1+C_{i0})
\]
at rescaled time
\[
\tau_{ij}=\frac{t_{ij}-t_{i0}}{18}.
\]
Because the same baseline measurement enters every post-baseline response from a subject, baseline subtraction induces additional within subject dependence. This dependence is retained by the subject level score construction and multiplier bootstrap. Time is rescaled to $[0,1]$. Within each analysis, the transformed responses are also divided by a common robust MAD scale so that the Huber clipping constant is dimensionless. Multiplication by a positive scale preserves the zero-reference null hypothesis, although the estimated scale is an additional nuisance quantity not covered by the fixed-score theory here.

The primary hypothesis is
\[
H_0:g(\tau)=0,\qquad 0<\tau\leq1.
\]
For squared loss, this is directly a mean function hypothesis. For Huber and absolute loss, it is a robust location hypothesis and coincides with the mean function hypothesis under the consistency condition in Assumption~\ref{ass:mean-consistency}. The pooled analysis is primary, while separate ddI and ddC analyses provide treatment arm sensitivity checks.

Only patients with at least one post-baseline measurement contribute to the analysis. Table~\ref{tab:cd4-sample} shows that the resulting design is very sparse. The pooled sample contains 406 patients and 938 post-baseline observations, with one to four measurements per patient and harmonic mean $m_H=1.894$. The two treatment arms have similar sampling frequencies. The five spline score procedures use the same tuning rule as in the simulation study, without application specific adjustment, and obtain $q=12$ in all three analyses. LW-SCB uses its own leave one subject out cross-validation criterion. Because the trial has only four scheduled post-baseline visit times, we allow a wider bandwidth sensitivity grid in the real data analysis. The selected bandwidths are $0.80$, $2.00$, and $0.68$ for the pooled, ddI, and ddC analyses, respectively.

\begin{table}[t]
\centering
\caption{CD4 application sample sizes and sampling frequencies after baseline subtraction.}
\label{tab:cd4-sample}
\small
\begin{tabular}{lrrrrrr}
\hline
Analysis & Subjects & Observations & Min.\ $m_i$ & Median $m_i$ & Max.\ $m_i$ & $m_H$\\
\hline
Pooled & 406 & 938 & 1 & 2 & 4 & 1.894\\
ddI    & 198 & 458 & 1 & 2 & 4 & 1.884\\
ddC    & 208 & 480 & 1 & 2 & 4 & 1.904\\
\hline
\end{tabular}
\end{table}

Each test uses 1999 subject level Gaussian multiplier draws. Table~\ref{tab:cd4-results} reports the resulting Monte Carlo $p$-values. For an observed statistic $T_{\mathrm{obs}}$ and bootstrap statistics $T_1^*,\ldots,T_B^*$, the Monte Carlo $p$-value is computed as
\[
\widehat p
=
\frac{1+\sum_{b=1}^{B} I(T_b^*\geq T_{\mathrm{obs}})}
{B+1},
\qquad B=1999.
\]
All five spline score procedures and LW-SCB reject the zero trajectory in the pooled sample and in both treatment arms. In every reported analysis, none of the 1999 multiplier statistics exceeds the corresponding observed statistic. Hence $\widehat p=1/2000=0.0005$, the smallest value attainable under this Monte Carlo convention. The equality of the reported $p$-values therefore reflects the finite Monte Carlo resolution and should not be interpreted as indicating identical strengths of evidence across the procedures. For LW-SCB, the widened bandwidth sensitivity gives statistics $0.396$, $0.309$, and $0.520$ for pooled, ddI, and ddC, compared with corresponding $0.95$ bootstrap critical values $0.116$, $0.147$, and $0.169$. The ddI cross-validation criterion selects the upper endpoint $b=2.00$, but the observed statistic remains more than twice its critical value. The empirical conclusion is therefore stable across the quadratic, robust, spline supremum, and published local linear supremum constructions.

\begin{table}[t]
\centering
\caption{CD4 application Monte Carlo $p$-values based on 1999 subject level multiplier draws.}
\label{tab:cd4-results}
\small
\begin{tabular}{lcccccc}
\hline
Analysis & L2-BS & Linf-BS & LW-SCB & RGT-Absolute & RGT-Huber & RGT-Squared\\
\hline
Pooled & 0.0005 & 0.0005 & 0.0005 & 0.0005 & 0.0005 & 0.0005\\
ddI    & 0.0005 & 0.0005 & 0.0005 & 0.0005 & 0.0005 & 0.0005\\
ddC    & 0.0005 & 0.0005 & 0.0005 & 0.0005 & 0.0005 & 0.0005\\
\hline
\end{tabular}
\end{table}

Figure~\ref{fig:cd4} shows the corresponding scatterplots and descriptive lowess smooths on the original log-change scale. The pooled and ddI smooths show a progressive decline from baseline. The ddC smooth declines through the intermediate visits and then partially recovers by month 18, while remaining below zero. These smooths are descriptive only and are not used in the construction of the test statistics or bootstrap critical values.

The trial uses a discrete visit schedule, whereas the asymptotic theory assumes random design points with a density on $[0,1]$. The CD4 analysis should therefore be interpreted as an illustration of the finite sample subject level procedure under severe sparsity, within subject dependence, and unbalanced follow-up, rather than as a direct verification of the continuous-design asymptotic assumptions. The analyses describe patients with observed follow-up and do not establish a treatment contrast or address potentially informative dropout.

\begin{figure}[t]
\centering
\includegraphics[width=0.96\textwidth]{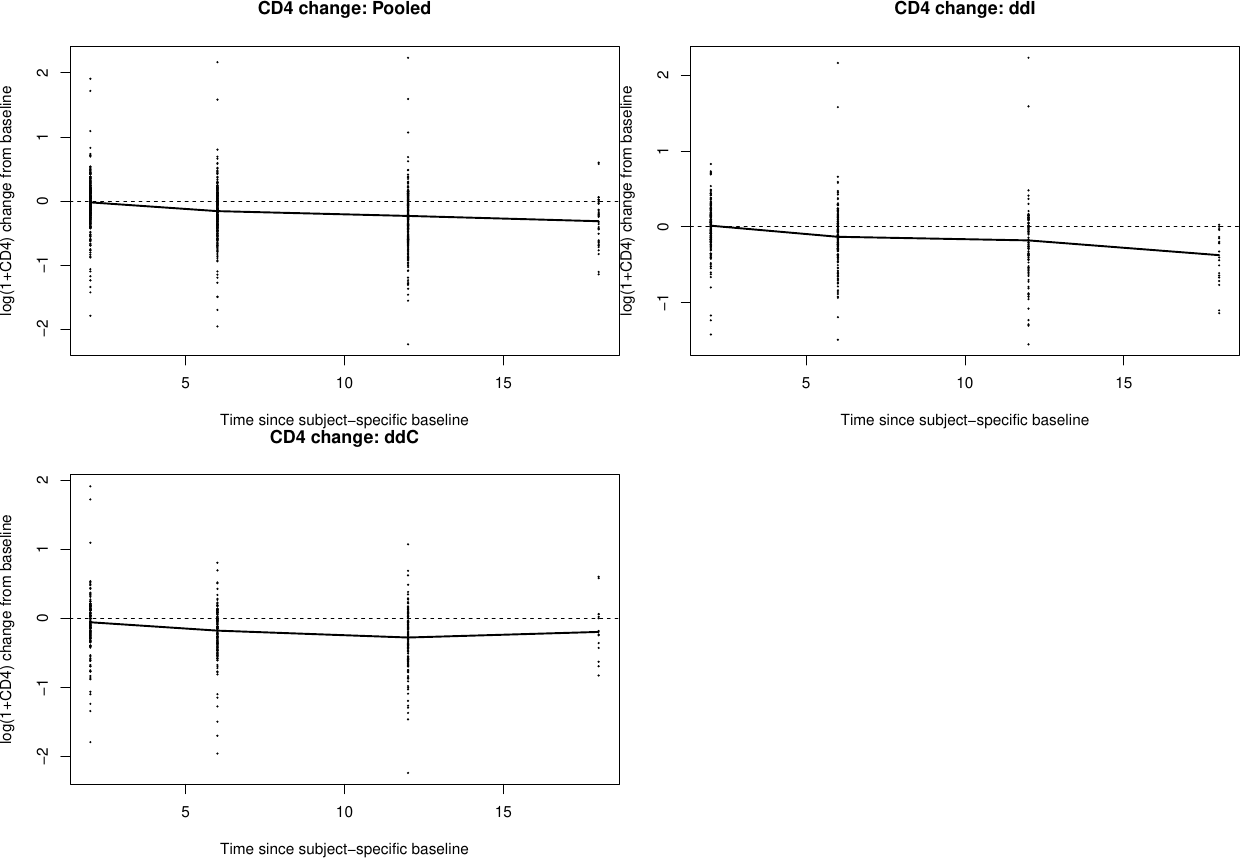}
\caption{Post-baseline change in $\log(1+\mathrm{CD4})$ for the pooled sample and the two randomized treatment arms. The dashed horizontal line denotes zero and the solid curve is a descriptive lowess smooth.}
\label{fig:cd4}
\end{figure}

The fixed pLRT-Staicu implementation is omitted from this application because its 20 interior knots produce a rank deficient design on the four scheduled post-baseline visits. An implementation adapted to this grid could be studied separately. This computational issue concerns the chosen pLRT specification, not pseudo likelihood ratio testing in general.

\section{Discussion}\label{sec:discussion}

The proposed method provides a common framework for global mean testing with squared, Huber, and absolute losses. Its main practical
benefit is the combination of flexible loss functions with subject
level calibration. Scores are averaged within each subject, retaining
both measurement variation and dependence among repeated observations.
The resulting test does not require a separate model for the full
within subject covariance function. Moreover, the score statistic is
evaluated directly under the null, so its computation does not require
fitting the penalized estimator under the alternative. Interpreting
the robust versions as tests of the population mean still requires
the stated score centering condition.

A further benefit is that the same multiplier procedure applies
across sparse, transition, and dense sampling. The user need not
choose between Gaussian and weighted quadratic reference
distributions according to the sampling regime. Instead, subject
multipliers reproduce the empirical covariance structure relevant
to the quadratic statistic. The theory explains why this matters.
At the optimal smoothing scale, the testing transition occurs at
$m_H\asymp n^{1/(4r)}$. Below this boundary, measurement variability
determines the nonparametric separation rate. Above it, the rate
reaches $n^{-1/2}$, with a matching lower bound supplied by the
Gaussian random intercept submodel. The feasible spline procedure
preserves calibration and attains the stated consistency separation
rate under the corresponding approximation, curvature, and power
conditions.

The numerical results illustrate the value of choosing the loss
according to the error distribution. In the reported settings,
Huber scores offer a useful balance between efficiency under regular
errors and resistance to contamination. The quadratic statistic
accumulates evidence across the domain, making it a useful choice
when departures are spread over the mean curve. The comparisons
also identify settings in which other procedures perform better.
LW-SCB is effective for localized departures, while pLRT-Staicu
performs well for the selected oscillatory departure. These findings
support the proposed test as a flexible option rather than a
uniformly dominant procedure. Small differences in empirical power
should also be interpreted alongside Monte Carlo uncertainty.

The CD4 application illustrates how the procedure can be used with
irregular repeated measurements and different loss functions.
The methods considered detect departures from a zero post-baseline
trajectory. Agreement across losses provides a useful sensitivity
comparison, although it does not establish a treatment effect.
The discrete visit schedule, selected follow-up sample, and
estimated response scale are not fully covered by the present
theory. Similarly, the fixed computational caps used in the
simulations are practical choices rather than tuning sequences
satisfying the asymptotic conditions indefinitely. The
sampling-density experiment demonstrates finite sample gains from
additional measurements, but its settings do not directly establish
a crossing of the theoretical transition boundary.

Several extensions would broaden the method's applicability.
Adaptive selection of the smoothing parameter could reduce
dependence on a specified smoothness order. Combining quadratic
and supremum statistics could improve sensitivity to both diffuse
and localized departures. Extending the theory to estimated
response scales, informative observation schedules, and dependent
subjects would also bring the framework closer to the conditions
encountered in longitudinal applications.

\end{document}